\documentclass[12pt, english]{article}

\usepackage{lecture}

\begin{document}
\printtitle									%
  	\vfill
\printauthor								%
\thispagestyle{empty}

\newpage

\section*{Preamble}
Around the year 1927 physicist, lead by Heisenberg and Schr\"odinger, discovered a new theory for subatomic particles, called quantum mechanics. This theory radically differed from classical mechanics in its physical assumptions as well as in its mathematical formulation.
The formal developments of quantum mechanics rapidly made it one of the most successful physical theories.
 Nevertheless, its physical foundations were very obscure and unjustifiable, and its mathematical basis were often shaky.
To provide a clear physical justification and a rigorous mathematical foundation for the theory, several mathematicians, lead by von Neumann, Weyl, and Jordan, set out to define an algebraic approach for quantum mechanics. In this lecture we follow their path and see how the physical formulation of the standard Schr\"odinger quantum mechanics appears from general physical arguments, paired with the brilliant observations of Heisenberg, and impeccable mathematics of von Neumann and friends...

\vspace{2cm}
In the first part of the lecture, we will develop the basic theory of operator algebras. Here we will not make any connections to physics and treat the subject in its pure form. In particular, we  discuss the function calculus and the Gelfand-Naimark-Segal construction of representations for \Cs-algebras. These tools will give us the necessary mathematical background to deal with \Cs-algebras when they appear in physics.

 In the second part of the lecture we establish a connection between pure mathematics and physics. We will see how general physical principles lead us to an algebraic description of physical theories. Combining this description with the Heisenberg's uncertainty principle and the canonical commutation relations developed by Born and Jordan we end up with the \Cs-algebra that describes a quantum mechanical particle. The final piece of the puzzle follows from the Stone-von Neumann theorem that  justifies the almost exclusive use of Schr\"odinger's formulation for quantum mechanics.

\newpage

\tableofcontents
\newpage

\section{Operator algebras: Banach and \Cs-algebras}
\subsection{Basic properties of operator algebras}
\begin{sectionmeta}
    This section is based on the  volume 1 and 3 of the series  by Kadison and Ringrose
    \cite{kadison1997fundamentalsI, kadison1997fundamentalsIII}.
    Some of the reminders and definitions that concern more standard statements
    of functional analysis are mostly from \cite{reed2012methods}. Further reading: for the abstract approach to \Cs-algebras and von Neumann
    algebras the book by Sakai \cite{sakai2012c} is great but is a little more advanced; another very satisfying read is the book by Pedersen
    \cite{pedersen2012analysis}; the four volumes of the aforementioned Kadison and Ringrose bible are anyway a treat  for the soul: the books cover most of the topics of old fashioned \Cs-algebras and develop the theory very slowly and self-consistently.
\end{sectionmeta}
In a nutshell, an algebra is a vector space in which we can multiply vectors. More formally we define the algebra as follows.

\begin{defi}[\textbf{Algebra}]
An algebra \( \mcA \) over a field \( \bbK = (\bbR, \bbC) \) is a vector space over \( \bbK \), equipped with multiplication that
\begin{enumerate}
    \item
    distributes over addition (from left and right):
    \begin{align}
         f \cdot (g + h) = f\cdot g + f \cdot h,
         &&
        (f + g) \cdot h = f \cdot h + g \cdot h;
    \end{align}

    \item
    is associative:
    \begin{align}
        (f \cdot g )\cdot h  = f \cdot (g \cdot h);
    \end{align}

    \item
    is compatible with scalar multiplication:
    \begin{align}
        \lambda (AB) = (\lambda A) B
        = A (\lambda B)
        \qquad (\lambda \in \bbK, \ A,B \in \mcA).
    \end{align}
\end{enumerate}
\end{defi}

A unit of an algebra \( \mcA \) is an element \( I\in \mcA \) such that \( I \cdot A = A \cdot I = A \) for all \( A \in \mcA \). It is clear that if such an element exists it must be unique. However, not every algebra contains a unite; if it does we call the algebra \textit{unital}. We will always use unital algebras in the following, even if we do not mention it explicitly. In one of the exercises, we will see that one can always adjoint a unite to an algebra to make it unital.
For unital algebras we can shorten the definition and state.

\begin{defi}
An algebra \( \mcA \) over a field \( \bbK \) is a vector space, whose underlying additive group has a ring structure compatible with multiplication with scalars.
\end{defi}

To recall the definitions of the vector spaces and rings consult the following remarks.
\begin{info} [Reminder: Vector space.]
A vector space over a field \( \bbK \) is a set of points \( V \) equipped with vector addition \( + \colon V \times V \to V \) and scalar multiplication \(\cdot \colon F \times V \to V \). These operations have to satisfy the following requirements.

For \(f,g \in V \) and \( \alpha, \beta \in \bbK \)
 the scalar multiplication has to
 \begin{enumerate}
    \item
    distribute over the vector addition:
    \( \alpha (f + g) = \alpha f + \alpha g \),

    \item
    distribute over the field addition:
    \( (\alpha + \beta ) f = \alpha f + \beta f\),

    \item
    be compatible with the field multiplication:
    \( (\alpha \beta)  f = \alpha (\beta f) \),

    \item
    has the neutral element:
    \( 1 f  = f \).
 \end{enumerate}

The vector addition has to be commutative and \( V \) with \( +\colon V \to V \) must be a commutative group. That is for \(f,g,h \in V \) and \( \alpha, \beta \in \bbK \) the addition
\begin{enumerate}
    \item
    is commutative (aka. Abelian):
    \( f + g  = g + f\),

    \item
    is associative:
    \( (f+g)+h = f + (g + h) \),

    \item
    has a neutral element \( e \in V \):
    \( e + f = f \)

    \item
    is such that for every \( f \in V \) there is an inverse \( f^{-1} \in V \):
    \( f + f^{-1} = e \).
\end{enumerate}
\end{info}

\begin{info}[Reminder: Ring.]
    A ring is a commutative group \( (R, +) \) equipped with multiplication \( \cdot \colon R \times R \to R \), such that with \( f,g,h \in R \) the multiplication
    \begin{enumerate}
        \item
        distributes from the left over addition:
        \( f \cdot (g + h) = f\cdot g + f \cdot h \)

        \item
        distributes from the right over addition:
        \( (f + g) \cdot h = f \cdot h + g \cdot h \)

        \item
        is associative:
        \( (f \cdot g )\cdot h  = f \cdot (g \cdot h)\)

        \item
        has a neutral element \( I \in R \):
        \( I \cdot f = f \cdot I = f \).
    \end{enumerate}
\end{info}

\begin{expl}
To visualize the concept let us recall some algebras that we know and love.
\begin{enumerate}
    \item
    The simplest example of an algebra is the algebraic field \( \bbK \) itself.

    \item
    A more elaborate example is the space \( \mcM (n, \bbC) \) of \( n \times n \) matrices with complex coefficients. This algebra is unital. The unit being the identity matrix.

    \item
    Yet another example is the space \( \bbR \times \bbR \) with the Hadamard product defined by
    \begin{align}
        (1,0) \cdot (1, 0) = (1,0), && (0,1) \cdot (0,1) = (0,1), && (1,0) \cdot (0,1) = (0,0),
    \end{align}
    and required to be compatible with multiplication by scalars. Notice, however, that this is not quite natural on the vector space \( \bbR^{2} \). This is because the definition of the product depends on the choice of the basis. In particular, this product changes under rotations, which is why it rarely appears in applications to physics.

    \item
    Finally, consider the space of all polynomials in a single variable. We will call this space \( \pol \). It is a vector space, and equipped with the pointwise product it becomes a unital algebra, whose unit is the constant \( 1 \).
\end{enumerate}
\end{expl}

\begin{defi}
A norm \( \| \cdot \|  \) on an algebra \( \mcA \) is said to be continuous if for \( A,B \in \mcA \) it satisfies \( \| A \cdot B \| \leq \| A \| \ \| B \| \).
\begin{enumerate}
    \item
    An algebra with a continuous norm is called a \textit{normed algebra}.

    \item
    A normed algebra that is complete (in the metric induced by the norm) is called a \textit{Banach algebra}.
\end{enumerate}
\end{defi}

As we will see, Banach algebras can be quite small (that is contain few elements). Nevertheless, they are very general, as they have only little structure compared to the algebras we will primarily use. We define this additional structure as follows.

\begin{defi}\label{def:involution}
Let \( \mcA \) be a Banach algebra. A mapping \( x \mapsto x^{\star} \) of \( \mcA \) into itself is called an involution if it satisfies the following conditions:
\begin{enumerate}
    \item
    \( (x^{\star})^{\star} = x \),

    \item
    \( (x + y)^{\star} = x^{\star} + y^{\star} \),

    \item
    \( (xy)^{\star} = y^{\star} x^{\star} \),

    \item
     \( (\lambda x)^{\star} = \bar{\lambda} x^{\star} \),     for \( \lambda \in \bbC \) and \( \bar{\lambda} \) being its complex conjugate.
\end{enumerate}
\end{defi}

\begin{info}[Reminder: Adjoint.]
Let \( B(H) \) be the space of bounded operators on the Hilbert space \( H \), and let \( \langle \cdot , \cdot \rangle \) denote the inner product on \( H \). Then, for \( A \in B(H) \) we define its adjoint as the operator \( A^{\dagger} \in B(H)\) such that
\begin{align}
    \langle x , A y \rangle = \langle A^{\dagger} x , y \rangle
    \qquad
    (x, y \in H).
\end{align}
\end{info}

It is not difficult to see that the adjoint on the space of bounded linear operators on a Hilbert space satisfies all the requirements of Definition \ref{def:involution}. Indeed, the adjoint is the reason why we made these requirements in the first place. However, in the case of bounded linear operators, we define an adjoint using the scalar product and then check that it satisfies the properties in the Definition. For algebras, the logic is reversed. We don't have the scalar product at our disposal, and so we require the properties as the defining features of the mapping.
In this sense the involution, sometimes simply called the \textit{star (operation)}, generalizes the concept of the adjoint for bounded linear operators.

With this additional structure we can now define the algebras, that will accompany us for the rest of the lecture.

\begin{defi}
    A \textit{Banach \( \star \)-algebra \( \mfB \)} is a Banach algebra with an (isometric) involution, that is with the property that \( \| A^{\star} \| = \| A \| \) for every \( A \in \mfB \).
\end{defi}
\begin{defi}
    A \textit{\Cs-algebra \( \mfA \)} is a complex Banach \( \star \)-algebra that has the \Cs-property:
    \begin{align}
        \| A^{\star} A \| = \| A \|^{2} \qquad (A \in \mfA).
    \end{align}
\end{defi}

In the definition of Banach \( \star \)-algebras some authors do not require the involution to be isometric. In our case, this difference will not play a big role, because we will only use either Banach algebras (without the star), or \Cs-algebras. For the latter, however, the involution is always isometric: since if \( A \) is an element in the \Cs-algebra \( \mfA \) then the \Cs-property implies
\begin{align}
    \| A \|^{2} = \| A^{\star} A \| \leq \| A^{\star} \| \| A \|,
\end{align}
and thus \( \| A \| \leq \| A^{\star} \|  \). Upon replacing \( A \to A^{\star} \) we get the inverse inequality and thus \( \| A \| = \| A^{\star} \| \). Notice, however, that the converse is not true: an isometric involution is not enough to ensure the \Cs-property. Thus, in general \Cs-algebras have more structure than Banach \( \star \)-algebras.

The name giving for the elements of the Banach \( \star \)-algebra follows the usual convention. We introduce the standard notation.

\begin{info}[Common notation.]
    Let \( \mfA \) be a unital algebra with an involution.
    For \( A \in \mfA \) we call the element \( A^{\star} \) the \textit{adjoint of} \( A \). We say that \( A \) is
    \begin{enumerate}
        \item
        \textit{self-adjoint} if \( A^{\star} = A \);

        \item
        \textit{normal} if \( AA^{\star} = A^{\star}A \);

        \item
        \textit{unitary} if \( AA^{\star} = A^{\star}A = I \), with the unit \( I \in \mfA \);

        \item
        \textit{invertible} if there exists a \( B \in \mfA \) with
        \( AB = BA = I \). If such \( B \) exists it is unique and we call it the \textit{two-sided inverse} of \( A \) and denote it by \( A^{-1} \);

        \item
        \textit{self-inverse} if \( A = A^{-1} \);
    \end{enumerate}
    It follows that the unit \( I \) is self-adjoint, self-inverse, and unitary since
    \begin{align}
        (I^{\star} I)^{\star} = (I^{\star})^{\star}  = I = I^{\star} I = I^{\star}.
    \end{align}
    Further, we can state the following:
    \begin{enumerate}
        \item
        The set of all self-adjoint elements of \( \mfA \) is a real vector space.

        \item
        The set of all unitary elements of \( \mfA \) forms a multiplicative group.

        \item
        Each element \( A \in \mfA \) can be uniquely expressed as \( H + \imath K \), where \( H = \frac{1}{2}(A + A^{\star}) \) and \( K = \frac{\imath}{2}(A^{\star} - A) \) are self-adjoint elements of \( \mfA \), called the \textit{real} and the \textit{imaginary} part of \( A \), respectively.
        Moreover, \( A \) is normal if and only if \( H \) and \( K \) commute.
    \end{enumerate}
\end{info}

\begin{prop}
Let \( A \) be an element in a \Cs-algebra. Then \( A \) has a (two-sided) inverse if and only if \( A^{\star} \) is invertible and then \( (A^{\star})^{-1} =  ( A^{-1})^{\star} \).
\end{prop}

\begin{proof}
Calculate
\begin{align}
    (A^{-1})^{\star} A^{\star} = (AA^{-1})^{\star} = \mathds{1},
\end{align}
and
\begin{align}
    A^{\star} (A^{-1})^{\star} = (A^{-1}A)^{\star} = \mathds{1}.
\end{align}
It follows that \( (A^{-1})^{\star} \) is a two-sided inverse of \( A^{\star} \). By uniqueness of the two-sided inverse we get \( (A^{-1})^{\star}  = (A^{\star})^{-1}\).
\end{proof}

Before we look at  some examples of \Cs-algebras, let us quickly refresh some basic definitions in topology.
\begin{info}[Reminder: Hausdorff space.]
    A topological space \( X \) is called Hausdorff, if for every pair of distinct points \( x,y \in X \) there exists two open sets \( U,V \subset X \) with \( U \cap V  = \emptyset \) and \( x \in U \) and \( y \in V \).
\end{info}
\begin{info}[Reminder: Compact set.]
    A subset \( K \) of a topological space \( X \) is compact, if each open cover of \( K \), contains a finite sub-cover of \( K \).
\end{info}

\begin{expl}[\textbf{Commutative \Cs-algebra}]
Let \( X \) be a compact Hausdorff space. Let \( \cb{X} \) denote the space of (bounded) complex valued continuous functions on \( X \). Define a ring structure on \( \cb{X} \) by the pointwise product: for \( f,g \in \cb{X} \) define \(     (fg)(x) \coloneqq f(x) g(x) \), \( (x \in X) \).
Next, topologize \( \cb{X} \) by the supremum norm: for \( f \in \cb{X} \) set     \( \| f \| = \sup_{x \in X} | f(x) | \).
Define an involution on \( \cb{X} \) by complex conjugation: for \( f \in \cb{X} \) define \( f^{\star} = \bar{f} \).
Then \( \cb{X} \) is a commutative unital \Cs-algebra and the unit is a constant function \( I = 1 \).

\begin{proof}
Clearly, \( \cb{X} \) is a vector space. Since a pointwise product of two continuous functions is again a continuous function, the space \( \cb{X} \) is algebraically closed. Since the uniform limit of continuous functions is continuous, the space \( \cb{X} \) is complete. Hence, \( \cb{X} \) with the supremum norm and the pointwise product is a Banach algebra.
It is straight forward to check that complex conjugation defines an involution. In fact, to a great extent complex conjugation inspired our definition of the involution.
Hence, \( \cb{X} \) is a Banach \( \star \)-algebra. (Here we do not assume the involution to be continuous. The continuity will follow, after we prove the \Cs-property.)
For the \Cs-property we observe for \( f \in \cb{X} \): \(    | f(x) \bar{f}(x) | = | f(x) | |\bar{f} (x) | = |f(x)| |f(x)| = |f(x)|^{2}
  \), at any \( x \in X \).
Hence,
\begin{align}
    \| f \bar{f} \| = \sup_{x \in X} |f(x)\bar{f}(x)| = \sup_{x \in X} |f(x)|^{2} = \big(\sup_{x\in X}|f(x)| \big)^{2} = \| f \|^{2}.
\end{align}
This completes the proof.
\end{proof}
\end{expl}

\begin{info}[Comment:]
Indeed, by the Gelfand transform (that we meet in a few lectures) we will see that every unital commutative \Cs-algebra is isometrically isomorphic to a \Cs-algebra of continuous functions on some compact Hausdorff space.
\end{info}

Now we consider a non-commutative example of a \Cs-algebra.
\begin{expl}
Consider our previous example: the algebra \( \mcM (n,\bbC) \) of \( n \times n \) matrices with complex coefficients. Equipped with the operator norm and adjoint as an involution, \( \mcM (n,\bbC) \) is a unital \Cs-algebra.
\end{expl}

\begin{proof}
This is a special case of the following example.
\end{proof}

\begin{expl}
Let \( H \) be a separable Hilbert space. Denote by \( B(H) \) the space of linear bounded operators on \( H \). Then with the operator product, the operator norm, and the adjoint as involution, the space \( B(H) \) is a unital \Cs-algebra.
\end{expl}
\begin{proof}
We only show the \Cs-property as the rest is trivial. Let \( A \) be in \( B(H) \) then
\begin{align}
    \| A \| ^{2} = \sup_{x \in H; \ \| x \| = 1}  \| A x \|
        = \sup_{x \in H; \ \| x \| = 1} \langle Ax, Ax \rangle
        = \sup_{x \in H; \ \| x \| = 1} \langle x , A^{\dagger}A x \rangle
        = \| A^{\dagger} A \|.
\end{align}
The last equality holds because \( A^{\dagger}A \) is a self-adjoint operator.
\end{proof}

Whenever we speak about \Cs-algebras it  is enlightening to consider, in the back of your head, the space of linear bounded operators on a Hilbert space. However, you should keep in mind that an abstract \Cs-algebra is more general than \( B(H) \) for a fixed Hilbert space \( H \); it is more general because a \Cs-algebra is not defined by the action of operators on a Hilbert space. Instead, it is only a collection of elements that satisfy some relations under multiplication.
In particular, let \( A \) be an element of a \( B(H) \). To characterize it we often define its action on \( H \). For example, if \( A \) maps a vector \( v \) in \( H \) to a rotated vector \( \tilde{v} \), we say that \( A \) is a rotation. But how do we characterize elements in an abstract algebra if we don't define their action...
we use their spectrum.

\begin{defi}
If \( A \) is an element of a Banach algebra \( \mfA \) we say that \( \lambda \in \bbC \) is a spectral value for \( A \) (relative to \( \mfA \)) when \( A - \lambda \mathds{1} \) does not have a two-sided inverse in \( \mfA \).
The set of spectral values of \( A \) is called the spectrum of \( A \) and is denoted \( \spectr_{\mfA}(A) \).
\end{defi}

\begin{expl}
Let \( \mfA = \cb{X} \) be the \Cs-algebra of continuous real-valued functions on a compact Hausdorff space \( X \).
For \( f \in \mfA \) what is its spectrum?

For \( \lambda \notin \Img(f) \) (the image of \( f \)) consider \( f - \lambda 1\). The function \( f - \lambda 1\) is continuous and defined on a compact domain. Hence, it attains all its minima. If \( f - \lambda 1 \) never vanishes then \( \inf |f - \lambda 1| > 0 \) and the function \( g = \frac{1}{f - \lambda} \) is  continuous. Clearly, \( g \) is the inverse of \( f - \lambda 1 \). Hence \( \spectr_{\mfA}(f) \subset \Img(f) \) .

Take \( \lambda \in \Img(f) \), so for some \( x  \in X \) we have \( (f - \lambda 1)(x) = 0 \). Since \( f - \lambda 1 \) is continuous there is a sequence \( ( x_{n}) \) in \( X \) that converges to \( x \) and satisfies \( (f - \lambda 1)(x_{n}) < 1/n \). Let \( g \) be the two-sided inverse such that \( (f-\lambda1)g = g(f - \lambda 1) = 1 \). Then
\begin{align}
    1 = | \big((f-\lambda1)g\big)(x_{n}) | = | g(x_{n})| | (f-\lambda1)(x_{n}) | \leq |g(x)| \frac{1}{n},
\end{align}
for every \( x_{n} \).
Thus, \( |g(x_{n})| \geq n \), and \( g \) is not continuous at \( x \). Consequently \( g \notin \cb{X} \) and \( \Img(f) \subset \spectr_{\mfA}(f) \). Together with the opposite inclusion we get \( \spectr_{\mfA}(f) = \Img(f) \).
\end{expl}

\begin{expl}
Let \( \mfM \) be the \Cs-algebra of complex \( n\times n \) matrices. For \( M \in \mfM \) consider \( M - \lambda \mathds{1} \). This element is invertible if and only if \( \det(M - \lambda \mathds{1}) \neq 0 \). Hence, if \( M \) is diagonalizable its spectrum is the set of eigenvalues of \( M \).
\end{expl}

\begin{info}[Warning.]
For a Banach algebra the spectrum of an element does indeed depend on the algebra. Therefore, we use the subscript to specify the ambient algebra. The following example shows that changing the algebra can change the spectrum.
\end{info}

\begin{expl}
Let \( H \) be a separable Hilbert space with the basis \( \{ e_{j} \} \) for \( j \in \bbZ \).
Let \( B(H) \) be the Banach algebra of linear bounded operators on \( H \).
(Above we have shown that \( B(H) \) is a \Cs-algebra, so in particular it is a Banach algebra). Consider the left shift \( L \) in \( B(H) \)  defined on the basis of \( H \) by \(L( e_{j} ) = e_{j+1} \) and extended by linearity and continuity to the whole \( H \).
The adjoint of \( L \) is the operator \( L^{\dagger}(e_{j} )= e_{j-1} \). One can easily verify that \( L L^{\dagger} = L^{\dagger} L = \mathds{1} \). Thus, \( L^{\dagger} \) is the two-sided inverse of \( L \). In other words, \( 0 \notin \spectr_{B(H)} (L) \).

Now consider a Banach algebra \( \mfB \) generated by \( \mathds{1} \) and \( L \). It is a subalgebra of \( B(H) \).
Is it still true that \( 0 \) is not in the spectrum of \( L \) relative to \( \mfB \)? No. To see this consider a closed subspace \( S \) of \( H \) spanned by \( \{ e_{j} \ \colon \ j = 1,2,\dots \} \). Then \( L \) maps this subspace into itself, and thus, each polynomial of \( L \) leaves this subspace invariant. It follows that every element in \( \mfB \) (the norm closure of such polynomials) maps \( S \) into itself.
Assume that \( \tilde{L} \) is a two-sided inverse of \( L \) in \( \mfB \). Then \( \tilde{L} \) is also in \( B(H) \) since \( \mfB \subset B(H) \). By the uniqueness of the two-sided inverse we thus have \( L^{\dagger} = \tilde{L} \).
However, \( L^{\dagger} \) does not preserve \( S \) as it maps the element \( e_{1} \) to \( e_{0}  \notin S\). Thus, \( L^{\dagger} \) is not in the algebra \( \mfB \).
We can thus conclude that \( L \) does not have a two-sided inverse in \( \mfB \) and \( 0 \in \spectr_{\mfB}(L) \).
\end{expl}

The above example shows that the spectrum of an element depends on the ambient Banach algebra. One of the main results that we will obtain in the first part of the lecture will be that for \Cs-algebras this cannot happen.

\begin{lem}\label{lem:spectrum of stared}
For a Banach \( \star \)-algebra \( \mfA \), let \( A \) be in \( \mfA \) and let \( \spectr_{\mfA}(A) \) be its spectrum. Then the spectrum of \( A^{\star} \) is
\begin{align}
    \spectr_{\mfA}(A^{\star}) = \big\{\  \bar{\lambda} \ \colon \ \lambda \in \spectr_{\mfA}(A) \big\} \Big(  = \overline{\spectr_{\mfA}(A)}\Big).
\end{align}
\end{lem}

\begin{proof}
If \( \lambda \notin \spectr_{\mfA}(A) \) then \( (A-\lambda I)^{-1} \) exists in \( \mfA \). Hence
\begin{align}
    \big( (A- \lambda I)^{\star} \big)^{-1} = (A^{\star} - \bar{\lambda} I)^{-1} = \big( (A - \lambda I)^{-1} \big)^{\star},
\end{align}
and \( (A- \lambda I)^{\star} \big)^{-1} \) exists. Thus, \( \bar{\lambda} \notin \spectr_{\mfA}(A^{\star}) \). Interchanging \( A \) and \( A^{\star} \) in the above argument yields the desired conclusion.
\end{proof}

\begin{lem}\label{lem:spectrum of invertible}
For a Banach algebra \( \mfA \), let \( A \) be an invertible element in  \( \mfA \). Then the spectrum of \( A^{-1} \) is
\begin{align}
    \spectr_{\mfA}(A^{-1})
    = \{ \lambda^{-1} \colon \lambda \in \spectr_{\mfA}(A) \}
    \Big( = \spectr_{\mfA}(A)^{-1} \Big)
\end{align}
\end{lem}
\begin{proof}
If \( A \) is invertible then \( 0 \notin \spectr_{\mfA}(A) \). Thus, if \( \lambda \) is in \( \spectr_{\mfA}(A) \) then \( \lambda \neq 0 \) and
\begin{align}
    A^{-1} -  \lambda^{-1} \mathds{1} = (A \lambda)^{-1}(\lambda - A).
\end{align}
Hence, \( A^{-1} - \lambda^{-1} \mathds{1} \) is invertible if and only if \( A - \lambda \mathds{1} \) is invertible, meaning
\begin{align}
    \spectr_{\mfA}(A^{-1})
    =  \{ \lambda^{-1} \colon \ \lambda \in \spectr_{\mfA}(A) \}.
\end{align}
This shows the assertion.
\end{proof}
In order to study the spectrum of an element it is good to have better understanding of invertible elements. The next few lemmas provide us with this understanding.
\begin{lem}[\textbf{Neumann series}]
If \( A \) is an element of a Banach algebra \( \mfA \) and \( \| A \| < 1 \) then
\begin{align}
    \sum_{k=0}^{n} A^{k} \qquad (\text{with } A^{0} = I),
\end{align}
has a limit \( B \in \mfA \) as \( n \) goes to infinity. Moreover, it holds that
\begin{align}
    B(I - A) = (I - A) B = I
 \end{align}
 and \( B \) is the (two-sided) inverse of \( I - A \).
\end{lem}

\begin{proof}
For \( n < m \) we can estimate
\begin{align}
    \| \sum_{k=0}^{n} A^{k} - \sum_{k=0}^{m} A^{k} \|
    = \| \sum_{k=n+1}^{m} A^{k} \| \leq \sum_{k=n+1}^{m} \| A\|^{k}.
\end{align}
Since \( \| A \| <1 \) the right-hand side of the above estimate goes to zero. Thus, \( \sum_{k=0}^{n} A^{k} \) is a Cauchy sequence. Since \( \mfA \) is complete this sequence converges to an element \( B \) in \( \mfA \) as \( n\) goes to infinity. Then
\begin{align}
    \big( \sum_{k=0}^{n} A^{k} \big) (I -A) = \sum_{k=0}^{n}A^{k} - \sum_{k=0}^{n} A^{k+1}
    =(I-A) \big( \sum_{k=0}^{n} A^{k} \big) = I - A^{n+1}.
    \label{eq:proof neumann series}
\end{align}
By continuity of the product we have \( \| A^{n+1} \| \leq \| A \|^{n+1} \overset{n\to \infty}{\longrightarrow} 0\). Thus \( A^{n+1} \) goes to zero as \( n \) goes to infinity. Thence the right-hand side of \eqref{eq:proof neumann series} converges to \( I \) and we obtain (again by continuity of multiplication)
\begin{align}
    B(I-A) = (I - A) B = I.
\end{align}
This completes the proof.
\end{proof}

\begin{prop} \label{prop: continuity of inverse}
If \( \mfA \) is a Banach algebra, the set \( \pN \subset \mfA\) of invertible elements in \( \mfA \) is an open subset of \( \mfA \) and the operation
\begin{align}
    \inv \colon A \to A^{-1}
\end{align}
on \( \pN \) is continuous.
\end{prop}

The following definition will be convenient for the proof of the proposition.
\begin{defi}
Let \( A \) be an element of a Banach algebra \( \mfA \). We define the \textit{left multiplication by \( A \)} as the mapping \( L_{A} \colon \mfA \to \mfA \) such that \( L_{A}(B) = AB \). In the similar way we define the \textit{right multiplication by \( A \)} as the mapping \( R_{A} \colon \mfA \to \mfA \) such that \( R_{A}(B) = BA \).
\end{defi}

\begin{info}
From the (joint) continuity of the product in a Banach algebra it follows directly that \( L_{A} \) and \( R_{A} \) are continuous for every \( A  \) in \( \mfA \).
\end{info}

\begin{proof}[Proof of proposition \ref{prop: continuity of inverse}]
    Even though the proof is not complicated, it is still helpful to understand the general idea of the argument.
    \begin{info}[Idea:]
    To show that \( \pN \) is an open set we need to show that for every element in \( \pN \) there exists an open set containing this element and that is contained in \( \pN \). We know that an open ball of radius 1 around the identity \( I \) is in \( \pN \). The idea is to shift this ball appropriately to prove the proposition. The multiplication operators that we just defined provide the ``shifting''. We now formalize the idea.
\end{info}

\begin{figure}
    \centering
    \def\svgwidth{16cm}
	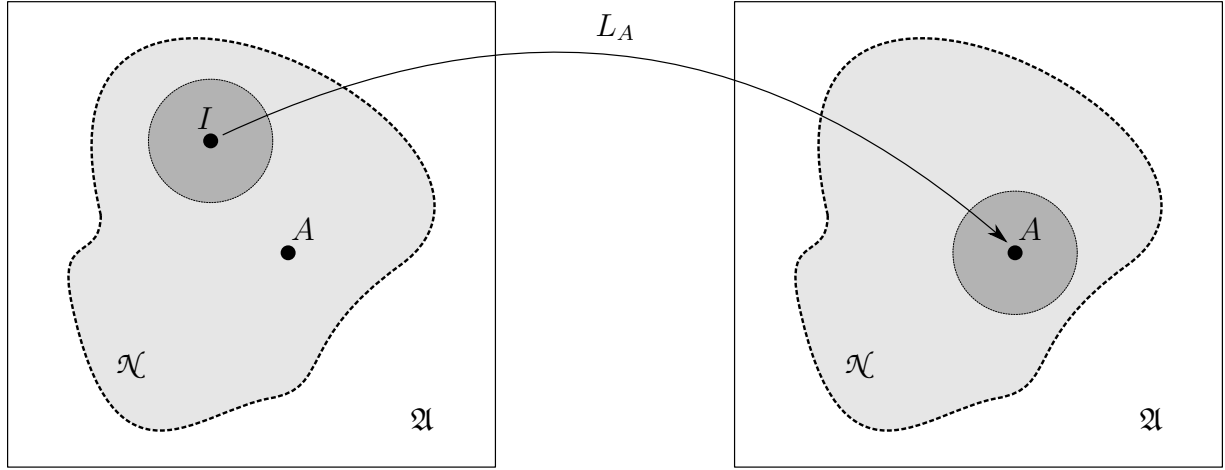
    \caption{Idea of proof.}
\end{figure}

For \( A \in \pN \) let \( L_{A} \) be the left multiplication by \( A \). Since \( A \) is invertible then so is \( L_{A} \) with the inverse \( L_{A^{-1}} \). As we mentioned above \( L_{A} \) and \( L_{A^{-1}} \) are both continuous (thus \( L_{A} \) is a homeomorphism).

From the Neumann series we know that \( B_{1}(I) \) --- an open ball of radius 1 around \( I \) --- is a subset of \( \pN \). Since \( L_{A} \) is a homeomorphism (in particular open) it maps \( B_{1}(I) \) to an open set \( \pU \) in \( \mfA \). The set \( \pU \) contains \( A \) since \( I \in B_{1}(I)\) and \( L_{A}(I)= A \); also \( \pU \) is contained in \( \pN \) since if \( A \) and \( B \) are invertible then so is \( AB \) (with the inverse \( B^{-1} A^{-1} \)).
Thus, for every \( A \in \pN \) there is an open set \( L_{A}(B_{1}(I)) \subset \pN \) containing \( A \) and \( \pN \) is open.

Now we prove that the mapping \( \inv \) is continuous on \( \pN \).
Assume that \( B \in \mfA \) is such that \( \| I - B \| < 1 \), then by the Neumann series
\begin{align}
    \sum_{k=0}^{\infty} (I - B)^{k} = \big(I - (I-B) \big)^{-1} = B^{-1}.
\end{align}
Choosing \( \| I - B \| < \epsilon < 2^{-1} \) and the fact that \( I \) is self-inverse (i.e.\ \( I I = I \)) we can estimate
\begin{align}
    \| \inv(B) - \inv(I) \|
    &= \| B^{-1} - I \| \leq \sum_{k=1}^{\infty} \| I - B \|^{k} \\
    &= \| I - B \| \Big(\sum_{k=0}^{\infty}  \| I - B \|^{k} \Big)
    = \frac{\| I - B \| }{1 - \| I - B \| }
    \leq 2 \| I - B \|
    \leq 2 \epsilon.
\end{align}
Thus \( \inv \) is continuous at the identity. To show that it is continuous at every element observe that for any \( B \in \pN \) the following identity holds
\begin{align}
    \inv = R_{B^{-1}} \circ \inv \circ L_{B^{-1}}.
\end{align}
Thus for each fixed element \( A \in \pN \) choosing the right-hand side of the above equality to be \( R_{A^{-1}} \circ \inv \circ L_{A^{-1}} \) shows that the map \( \inv \) is a combination of continuous mappings \( A \overset{L_{A^{-1}}}{\to} I \overset{\inv}{\to} I \overset{R_{A^{-1}}}{\to} A^{-1} \). Thus for each \( A \in \pN \) the mapping \( \inv \) is continuous.
\end{proof}

We will now show that every element in a Banach algebra has a spectrum. The proof uses some known, but not entirely trivial statements from the complex and functional analysis. Since, these topics are themselves very extensive,  we only recall the basic needed statements before continuing with the theorem.
\begin{info}[Reminder: Holomorphic function.]
    A complex valued function \( f \) is called holomorphic at \( z_{0} \in \bbC \) if on some neighborhood of \( z_{0} \) the limit
    \begin{align}
        \lim_{z\to z_{0}} \frac{f(z) - f(z_{0})}{z - z_{0}}
    \end{align}
    exists. A function that is holomorphic at every point of some open set \( \pO \subset \bbC \) is said to be holomorphic on \( \pO \). If \( f(z) \) is holomorphic on an entire complex plain we call it \textit{an entire} function. By Liouville's theorem, every bounded entire function must be constant. In particular, if an entire function vanishes at infinite it must be zero.
\end{info}

\begin{info}[Reminder: Hahn-Banach theorems.]
    The name ``Hahn-Banach theorem'' refers to several related results (consult e.g.\ \cite[chap. 1  2]{kadison1997fundamentalsI}). There are two types of such theorems: separation and extension results.
    One version of the Hahn-Banach extension theorem is the following.

    \begin{thm*}
    Let \( X \) be a vector space and \( X_{0} \) be a subspace of \( X \). Let \( \rho_{0} \colon X_{0} \to \bbC \) (or \( \bbR \)) be a linear functional, and \( p \colon X \to \bbR \) be a semi-norm such that
    \begin{align}
        | \rho_{0} (y) | \leq p(y) \qquad (y \in X_{0}).
    \end{align}
    Then there exists a linear functional \( \rho \colon X \to \bbC \) (or \( \bbR \)) such that
    \begin{align}
        | \rho(x) | \leq p(x) \qquad (x \in X), && \rho(y) = \rho_{0}(y) \qquad (y \in X_{0}).
    \end{align}
    \end{thm*}
The corollaries of this very general theorem are often more useful than the theorem itself. For example the next corollary assures that we can always extend a linear functional.
\begin{cor*}
Let \( X_{0} \) be a subspace of a normed space \( X \). Let \( \rho_{0} \) be a bounded linear functional on \( X_{0} \). There exists a bounded linear functional \( \rho \) on the whole \( X \) such that \( \| \rho \|  = \| \rho_{0} \| \) and \( \rho(x) = \rho_{0}(x) \) for \( x \in X_{0} \).
\end{cor*}
Another useful corollary is the one we use in the proof of Theorem \ref{thm:non-empty spectrum}.
    \begin{cor*}
    If \( x \) is a non-zero vector in a normed space \( X \), there is a bounded linear functional \( \rho \) on \( X \) such that \( \| \rho \| = 1 \) and \( \rho(x) = \| x \| \).
    \end{cor*}
\end{info}

Equipped with all that fundamental knowledge. We can now prove that the spectrum of an element in a Banach algebra is never empty.

\begin{thm}\label{thm:non-empty spectrum}
If \( A \) is an element of a Banach algebra \( \mfA \) then \( \spectr_{\mfA}(A) \) is a non-empty closed subset of the closed disc in \( \bbC \) with center 0 and radius \( \| A \|  \).
\end{thm}

\begin{proof}
Let \( \lambda \in \bbC \) be such that \( | \lambda | > \| A \|  \) then
\begin{align}
    A - \lambda I = \lambda \Big(\frac{1}{\lambda} A - I \Big).
\end{align}
Since \( \| A\lambda^{-1} \| < 1 \) the element \( \frac{1}{\lambda} A - I \) is invertible (by the Neumann series). Thus, \( A - \lambda I \) has the inverse \( \frac{1}{\lambda} ( \frac{1}{\lambda}A - I)^{-1} \). It follows that \( \spectr_{\mfA}(A) \subset D_{\| A\|}(0) \) (a disc with center 0 and radius \( \| A\| \)).

Now we need to show that \( \spectr_{\mfA}(A) \) is not empty. Let \( \lambda \notin \spectr_{\mfA}(A) \). Then by the proposition \ref{prop: continuity of inverse}, the element \( A - \lambda I \) is invertible in a small neighborhood around \( \lambda \). In particular the set \( \bbC \setminus \spectr_{\mfA}(A) \) is open.

\begin{info}[Reminder: Resolvent equation.]
Recall the first resolvent equation \cite[thm. VI.5]{reed2012methods}
\begin{align}
    (A - \lambda' I)^{-1} - (A - \lambda I )^{-1}
    &=(\lambda' - \lambda)(A-\lambda' I) (A - \lambda I)^{-1}.
\end{align}
\end{info}

Let \( \lambda' \) be close enough to \( \lambda \) so that \( A - \lambda' I \) is invertible.
Let \( \rho \colon \mfA \to \bbC \) be a continuous linear functional on \( \mfA \). Then by the resolvent equation
\begin{align*}
    \frac{\rho \Big( (A - \lambda' I)^{-1} \Big)  - \rho \Big( (A - \lambda I)^{-1}\Big)}{\lambda' - \lambda}
    =\rho \Big( (A - \lambda' I)^{-1} (A - \lambda I)^{-1}\Big)
    \overset{\lambda' \to \lambda}{\to} \rho \Big((A-\lambda I)^{-2} \Big).
\end{align*}
The last limit follows by the continuity of \( \rho \), the product, and the mapping \( \inv \) on \( \pN \).
Thus the function \( \lambda \mapsto \rho(A - \lambda I) \) is holomorphic on \( \bbC \setminus \spectr_{\mfA}(A) \).
Also since \( (\frac{A}{\lambda} - I) \) approaches \( -I \) as \( | \lambda | \) grows, we have
\begin{align}
    \rho \Big( (A - \lambda I)^{-1} \Big) = \frac{1}{\lambda} \rho \Big( (\frac{1}{\lambda}A - I)^{-1} \Big) \to 0, \qquad \text{for }| \lambda | \to \infty.
\end{align}
Now, if \( \spectr_{\mfA}(A) \) were empty, the function \( \lambda \mapsto \rho (A - \lambda I) \) would be an entire function that vanishes at infinity. By the Liouville's theorem we then had \( \rho(A - \lambda I)^{-1} = 0 \)  for all \( \lambda \in \bbC \) and in particular for \( \lambda = 0 \). But then \( \rho(A^{-1}) = 0 \) for every continuous linear functional \( \rho \), and by the Hahn-Banach theorem \( A^{-1} = 0 \). This is a contradiction. Hence \( \spectr_{\mfA}(A) \neq \emptyset\). This completes the proof.
\end{proof}

The above theorem states that the spectrum of an element in the Banach algebra is bounded by its norm. However, the norm is just \textit{an upper bound} not \textit{the smallest} upper bound. To study how different or similar is the norm and the smallest upper bound of the spectrum, we give the latter quantity a proper name first.

\begin{defi}[Spectral radius]
The \textit{spectral radius} \( r_{\mfA}(A) \) of an element \( A \) in a Banach algebra \( \mfA \) is
\begin{align}
    r_{\mfA}(A) = \sup \{|\lambda| \colon \lambda \in \spectr_{\mfA}(A) \}
\end{align}
\end{defi}

\begin{info}
The spectral radius is the radius of the smallest disc in \( \bbC \) with center 0 containing \( \spectr_{\mfA}(A) \). Also from the preceding theorem, \( \| A \|  \) is the radius of the largest disc in \( \bbC \) with center 0 that contains \( \spectr_{\mfA}(A) \). Hence we always have \( r_{\mfA}(A) \leq \| A \|  \).
\end{info}
\begin{figure}
    \centering
    \def\svgwidth{7cm}
\begingroup%
  \makeatletter%
  \providecommand\color[2][]{%
    \errmessage{(Inkscape) Color is used for the text in Inkscape, but the package 'color.sty' is not loaded}%
    \renewcommand\color[2][]{}%
  }%
  \providecommand\transparent[1]{%
    \errmessage{(Inkscape) Transparency is used (non-zero) for the text in Inkscape, but the package 'transparent.sty' is not loaded}%
    \renewcommand\transparent[1]{}%
  }%
  \providecommand\rotatebox[2]{#2}%
  \newcommand*\fsize{\dimexpr\f@size pt\relax}%
  \newcommand*\lineheight[1]{\fontsize{\fsize}{#1\fsize}\selectfont}%
  \ifx\svgwidth\undefined%
    \setlength{\unitlength}{403.25998486bp}%
    \ifx\svgscale\undefined%
      \relax%
    \else%
      \setlength{\unitlength}{\unitlength * \real{\svgscale}}%
    \fi%
  \else%
    \setlength{\unitlength}{\svgwidth}%
  \fi%
  \global\let\svgwidth\undefined%
  \global\let\svgscale\undefined%
  \makeatother%
  \begin{picture}(1,1)%
    \lineheight{1}%
    \setlength\tabcolsep{0pt}%
    \put(0,0){\includegraphics[width=\unitlength,page=1]{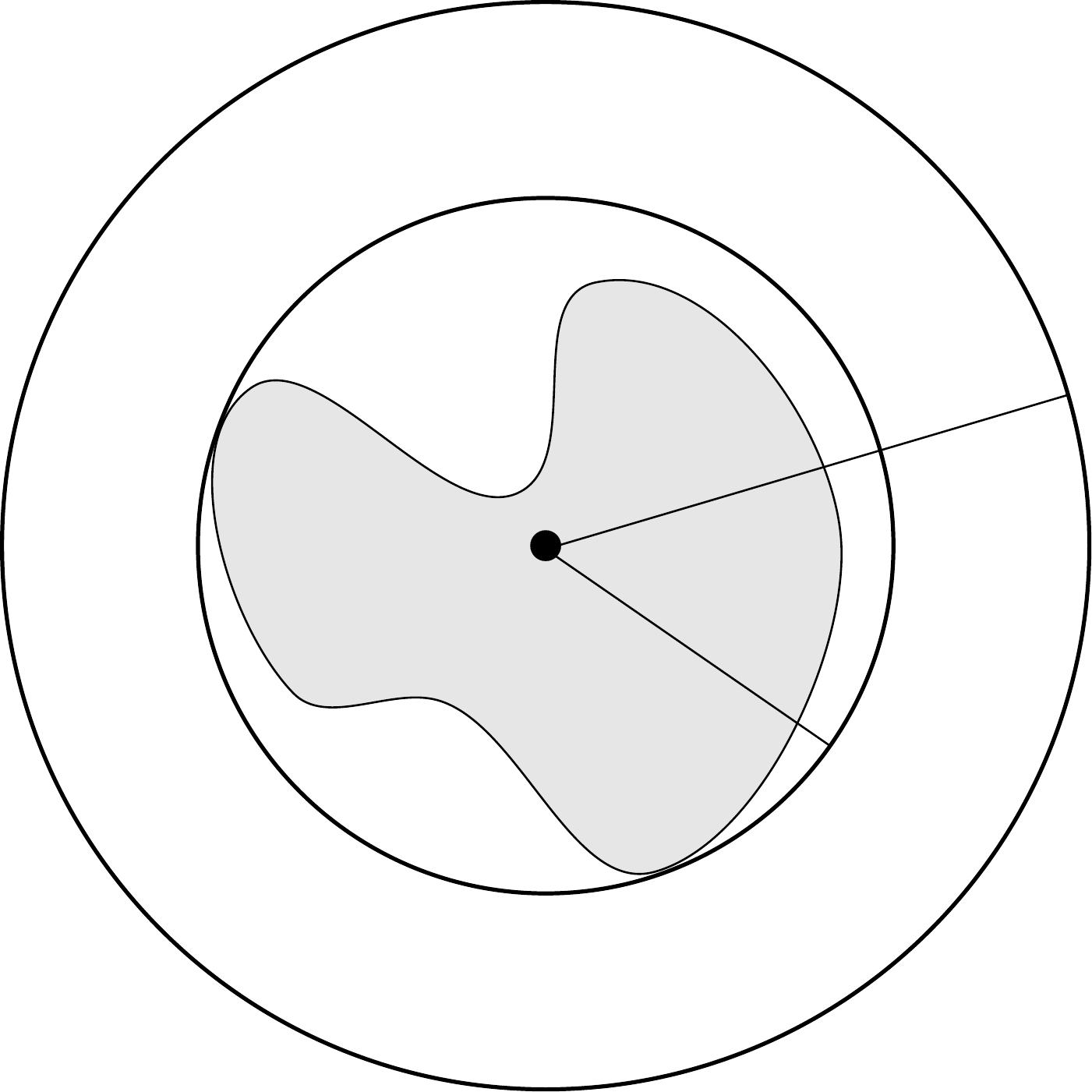}}%
    \put(0.03078416,0.93065022){\color[rgb]{0,0,0}\makebox(0,0)[lt]{\lineheight{1.25}\smash{\begin{tabular}[t]{l}$\bbC$\end{tabular}}}}%
    \put(0.99084058,0.65104895){\color[rgb]{0,0,0}\makebox(0,0)[lt]{\lineheight{1.25}\smash{\begin{tabular}[t]{l}$\| A \|$\end{tabular}}}}%
    \put(0.78364538,0.30395774){\color[rgb]{0,0,0}\makebox(0,0)[lt]{\lineheight{1.25}\smash{\begin{tabular}[t]{l}$r_{\mfA}(A)$\end{tabular}}}}%
    \put(0.26779479,0.41965485){\color[rgb]{0,0,0}\makebox(0,0)[lt]{\lineheight{1.25}\smash{\begin{tabular}[t]{l}$\spectr_{\mfA}(A)$\end{tabular}}}}%
  \end{picture}%
\endgroup%

    \caption{Spectrum of \( A \) relative to \( \mfA \) bounded by the spectral radius and by \(\| A  || \).}
\end{figure}
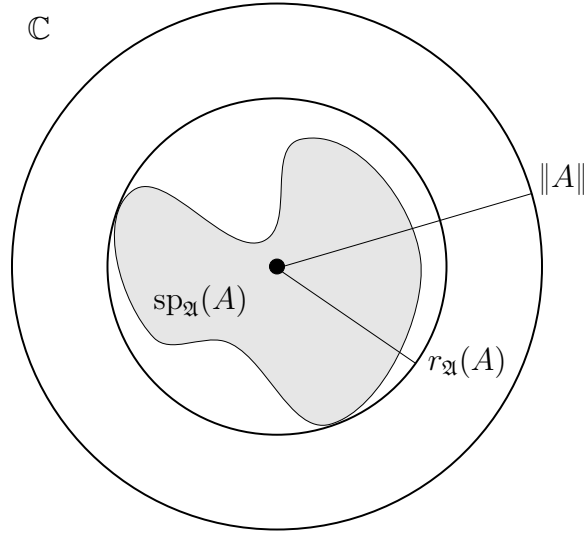

Surprisingly, the spectral radius --- rather an algebraic property --- can be stated in analytical terms, and thus it relates to topology. For the proof consult e.g.\ \cite[thm. 3.3.3]{kadison1997fundamentalsI}

\begin{thm}[\textbf{Spectral radius formula}]
Let \( A \) be an element of a Banach algebra \( \mfA \). The sequence \( ( \| A^{n} \|^{1/n} ) \) converges to \( r_{\mfA}(A) \) such that
\begin{align}
    r_{\mfA}(A) = \lim_{n \to \infty} \| A^{n} \|^{1/n}.
\end{align}
\end{thm}

This relation between the spectrum and topology is as surprising as it is powerful. We have already seen that the spectrum of an element depends on the ambient Banach algebra. Nevertheless, from the spectral radius formula if follows that if \( \mfB \) is a Banach subalgebra of \( \mfA \) then the spectral radii of \( A \in \mfB \) and \( A \in \mfA \) coincide, i.e. \( r_{\mfA}(A) = r_{\mfB}(A) \). In a \Cs-algebra the spectral radius formula provides even stronger results.

\begin{prop}\label{prop:spectral radius and norm in cstar}
Suppose that \( A \) is an element of a \Cs-algebra \( \mfA \). Then
\begin{enumerate}
    \item if \( A \) is normal then \( r_{\mfA(A)} = \| A \| \); \label{it:a}
    \item if \( A \) is self-adjoint then \( \spectr_{\mfA}(A) \) is a compact subset of the real line and it contains at least one of the two real numbers \( \pm \| A \| \);\label{it:b}
    \item if \( A \) is unitary then \( \| A \| = 1 \) and \( \spectr_{\mfA}(A) \) is a compact subset of the unite circle \( \{  a \in \bbC \colon |a | =1 \} \). \label{it:c}
\end{enumerate}
\end{prop}
\begin{proof}
\ref{it:a} Let \( H \) be a self-adjoint element in \( \mfA \) and \( n \in \bbN \). Then by the \Cs-property we have
\begin{align}
    \| H^{2n} \| = \| (H^{n})^{\star}H^{n} \| = \| H^{n} \|^{2}.
\end{align}
Hence, if \( q = 2^{m} \) for some \( m=1,2,\dots \) then by induction we get
\begin{align}
    \| H^{q} \| = \| H \|^{q}.
\end{align}
Using the spectral radius formula we get
\begin{align}
    r_{\mfA}(H) = \lim_{q \to \infty} \| H^{q} \|^{1/q} = \| H \|.
\end{align}
So the assertion holds for self-adjoint elements. Now, let \( A \) be normal and \( H = A^{\star}A \) then
\begin{align}
    \| A \|^{2} \overset{(1)}{=} \| A^{\star}A\| \overset{(2)}{=} r_{\mfA}(A^{\star}A) \overset{(3)}{\leq} r_{\mfA}(A^{\star}) r_{\mfA}(A) \overset{(4)}{=} r_{\mfA}(A)^{2} \overset{(5)}{\leq} \| A \|^{2}.
\end{align}
With the following justifications for the estimates:
(1) is the \Cs-property; (2) follows from what we have just shown above; (3) follows from the fact that for normal elements \( \| (A^{\star} A)^{n}\| \leq \|(A^{\star})^{n} \| \|A^{n} \| \) and thus \( \lim \| (A^{\star} A)^{n} \|^{1/n} \leq \lim \| (A^{\star})^{n} \|^{1/n} \| A^{n} \|^{1/n} \); (4) follows from the fact that \( \spectr_{\mfA}(A^{\star}) = \overline{\spectr_{\mfA}(A)}\); (5) follows since \( \| A \| \) defines the largest radius of a disc that contains the spectrum.

Since the left and the right hand side of the above estimate are equal all the "less or equal" signs in the estimate have to be actual equalities. Hence \( r_{\mfA}(A) = \| A \| \).

\vspace{0.3cm}
\ref{it:b} By Theorem \ref{thm:non-empty spectrum} the spectrum of \( A \in \mfA \) is a compact set of \( \bbC \). Thus the spectral radius attains its maximum and there is a \( \lambda \in \spectr_{\mfA}(A) \) such that \( |\lambda | = r_{\mfA}(A) \). If \( A \) is self-adjoint, then from the first part of the proposition we also know that \( r_{\mfA}(A) = \| A \| \) and thus there is an element \( \lambda \in \spectr_{\mfA}(A) \) with \( |\lambda | = \| A \| \). To prove the assertion we only need to show that the spectrum of a self-adjoint element is a subset of a real line.

Assume that for some \( a,b \in \bbR \) the element \( a + \imath b \) is in \( \spectr_{\mfA}(A) \).
For \( n \in \bbN \) consider the element \( B_{n} = A + \imath n b I \). Then, since \( a + \imath (n+1)b \in \spectr_{\mfA}(B_{n}) \), we have
\begin{align}\label{eq:square of spectrum2}
    | a + \imath (n+1) b |^{2}
    \leq \| B_{n} \|^{2}
    \leq \| B_{n}^{\star} B_{n} \| %
    = \| A^{2} + n^2 b^2 \|
    \leq \| A \|^{2} + n^2 b^2.
\end{align}
On the other hand,
\begin{align}\label{eq:square of spectrum1}
    | a + \imath (n+1) b |^{2} = (a+\imath (n+1) b)(a - \imath (n+1)b) = a^2 + n^2 b^2 + 2nb^2 + b^2.
\end{align}
So combining \eqref{eq:square of spectrum2} and \eqref{eq:square of spectrum1},
we get \( a^{2} + (2n+1)b^2  \leq \| A \|^{2} \) for every \( n \in \{ 0,1,2,\dots \} \). Since \( A \) is, however, bounded \( b \) must be zero.

\vspace{0.3cm}
\ref{it:c}
If \( A \) is unitary in \( \mfA \), then \( \| A \|^2 =  \| A^{\star} A \| = 1 \). Thus \( \| A \| = 1 \). Therefore, if \( \lambda \) is in \( \spectr_{\mfA}(A) \) then \( | \lambda | \leq \| A \| \leq 1 \). On the other hand, since the inverse of \( A \) is its adjoint \( A^{\star} \)  we have from Lemma \ref{lem:spectrum of invertible} that \( \lambda^{-1} \in \spectr_{\mfA}(A^{-1}) = \spectr_{\mfA}(A^{\star}) \) and \( | \lambda^{-1} | \leq \| A^{\star} \| = \| A \| \leq 1 \).  Thus, \( | \lambda | \leq 1 \) and \( | \lambda^{-1} | \leq 1 \) meaning that \( | \lambda | = 1 \).
\end{proof}

\begin{info}[Algebra vs. bounded operators on a Hilbert space.]
For bounded operators on a Hilbert space we can use the inner product to prove that a self-adjoint operator on a Hilbert space has a real spectrum. This makes the proof simpler (but also less general).
    Consider a self-adjoint bounded linear operator \( A \) on a Hilbert space \( H \). Let \( \lambda \in \bbC \) be such that \( | \Im(\lambda) | > 0 \). Then for every \( x \in H \) we have by the Cauchy-Schwarz inequality,
    \begin{align}
        \| x \| \| (A-\lambda \mathds{1}) x \|
        \geq   | \langle x, (A-\lambda \mathds{1}) x \rangle | \geq |\Im(\lambda)| \| x \|^{2} >0.
    \end{align}
In particular, \( \|(A- \lambda \mathds{1}) x \| \geq |\Im(\lambda)| \| x \| \)  for every \( x \in H \). Hence \( A - \lambda \mathds{1} \) is injective and has closed range. Also, since \( | \Im(\bar{\lambda}) | = | \Im(\lambda) | \), with the same argument we get that \( (A - \lambda \mathds{1})^{\dagger} \) is also injective, and \( A - \lambda \mathds{1} \) has dense range in \( H \). It follows that \( A - \lambda \mathds{1} \) is invertible and \( \lambda \notin \spectr_{B(H)}(A) \).
\end{info}

\subsection{Function calculus for \Cs-algebras}
Function calculus on \Cs-algebras generalizes the idea of diagonalizing matrices. Diagonalizing a matrix allows us to define powers, square roots, exponential, logarithms of matrices, and what not...
For bounded operators on a Hilbert space the generalization leads to  the spectral calculus.
The functional calculus in \Cs-algebras essentially bears the same meaning.
We can even consider it a generalization of the spectral calculus, as it is defined for abstract algebra elements.
Generally speaking the function calculus allows us to associate algebra elements to continuous functions, calculate with those functions, and then associate the result back to algebra elements. This process, is very powerful. It is hard to overestimate the use of the function calculus in the abstract theory and applications...

In the following we will use the following notation. If \( p \) is a polynomial, say \( p(t) = \sum_{k=0}^{n} a_{k} t^{k} \) with \( a_k \in \bbC \), then we will write \( p(A) \) to denote the algebra element \( \sum_{k=0}^{n} a_k A^{k}\).

\begin{prop}[\textbf{Spectral calculus for polynomials}]\label{prop:spectral for polys}
If \( A \) is an element in a Banach algebra \( \mfA \) and \( p \) is a polynomial in one variable, then
\begin{align}
    \spectr_{\mfA}(p(A)) = \{ p(\lambda) \colon \lambda \in \spectr_{\mfA}(A) \} \Big(= p(\spectr_{\mfA}(A))\Big).
\end{align}
\end{prop}

\begin{proof}
If \( \lambda \in \spectr_{\mfA}(A) \), then \( A - \lambda I \) has no two-sided inverse in \( \mfA \).
If \( p(x) = a_{n}x^{n} + \cdots + a_{0} \) then
\begin{align}
    p(A) - p(\lambda) I = a_{n} (A^{n} - \lambda^{n}I) + \cdots + a_{1}(A - \lambda I).
\end{align}
For \( k\geq 2 \) we can rewrite \( A^{k} - \lambda^{k}I \) as
\begin{align*}
    A^{k} - \lambda^{k}I
    &= (A - \lambda I)(A^{k-1} + \lambda A^{k-2} + \cdots + \lambda^{k-1}I)\\
    &= (A^{k-1} + \lambda A^{k-2} + \cdots + \lambda^{k-1}I)(A - \lambda I).
\end{align*}
Hence, we get
\begin{align*}
    p(A) - p(\lambda) I = (A - \lambda I)\big(a_{n}(A^{n-1} + \lambda A^{n-1} + \cdots + \lambda^{n-1}I) + \cdots + a_{1} \big)
    = \big( \cdots \big)(A - \lambda I),
\end{align*}
and \( p(A) - p(\lambda) I \) does not have a two-sided inverse (otherwise \( A - \lambda I \) would have a two-sided inverse). It follows that \( p(\lambda) \in \spectr_{\mfA}(p(A)) \), and thus \( p(\spectr_{\mfA}(A)) \subset \spectr_{\mfA}(p(A)) \).

For the opposite inclusion consider that if \( \gamma \) is in \( \spectr_{\mfA}(p(A)) \) and \( \lambda_{1},\dots , \lambda_{n} \) are the \( n \) roots of \( p(x) - \gamma \) then
\begin{align}
    p(A) - \gamma I = c (A - \lambda_{1} I ) \cdots (A - \lambda_{n} I),
\end{align}
and at least one of \( A - \lambda_{j} I \) is not invertible. Hence, \( \lambda_{j} \in \spectr_{\mfA}(A) \) for at least one of \( \lambda_{j} \)'s and by previous discussion \( \gamma = p(\lambda_{j}) = \spectr_{\mfA}(p(A)) \) (the first equality follows since \( p(\lambda_{j}) - \gamma = 0 \) by definition of roots). Hence, \( \spectr_{\mfA}(p(A)) \subset p(\spectr_{\mfA}(A)) \).
\end{proof}

To unclutter the notation in the proofs we will use the following common notations. We summarize them in the following paragraph.

\begin{info}[Common notation.]
By \( \pol \) we will denote the space of polynomials with complex coefficients.
If \( A \) is an element in a \Cs-algebra \( \mfA \) with the spectrum \( \spectr_{\mfA}(A) \) then,
\begin{itemize}
    \item
    by \( \cc{\spectr_{\mfA}(A)} \) we will denote the \Cs-algebra of continuous functions on \( \spectr_{\mfA}(A) \);

    \item
    by \( \id \colon \spectr_{\mfA}(A) \to \spectr_{\mfA}(A) \) we will denote the identity mapping \( z \mapsto z \).

     \item
     when we write \( \pol \subset \cc{\spectr_{\mfA}(A)} \) we consider elements of \( \pol \) as continuous functions on \( \spectr_{\mfA}(A) \).
\end{itemize}

\end{info}

It is time to introduce the function calculus for self-adjoint operators. Almost the same theorem, with minor modifications, holds also for normal elements in a \Cs-algebra (see the Remark at the end of the section). However, we will prove only the following result, as it does not require any additional technicalities that are needed for normal elements. Before the proof, we recall the Weierstrass approximation theorem for convenience.

\begin{info}[Reminder: Weierstrass approximation theorem.]
    Any continuous function on a bounded interval can be uniformly approximated by polynomial functions.
\end{info}

\begin{thm}[\textbf{Function calculus for self-adjoint elements}]\label{thm:spectral calculus}
Let \( A \) be a self-adjoint element of a \Cs-algebra \( \mfA \). There is a unique isometric \( \star \)-isomorphism \( \varphi \) from \( \cc{\spectr_{\mfA}(A)} \) into \( \mfA \) such that \( \varphi(\id) = A \). Moreover, for \( f \in \cc{\spectr_{\mfA}(A)} \),
\begin{enumerate}
    \item\label{it:sc normal}
    \( \varphi(f) \) is a normal element and it is self-adjoint if and only if \( f \) is real valued;

    \item\label{it:sc smallest comm}
    the set \( \{ \varphi(f) \colon f \in \cc{\spectr_{\mfA}(A)} \}\) is a commutative \Cs-algebra and it is the smallest \Cs-subalgebra of \( \mfA \) that contains \( A \);

    \item\label{it:sc lim of pol}
    \( \varphi(f) \) is the limit of a sequence of polynomials in \( I \) and \( A \);

    \item\label{it:sc comm with others}
    if \( B \) is in  \( \mfA \) that commutes with \( A \), then it also commutes with \( \varphi(f) \).
\end{enumerate}
\end{thm}

\begin{proof}
Assume a polynomial \( p \in \pol \) of the form,
\begin{align}
    p(t) = \sum_{k=0}^{n} a_{k} \ t^{k}.
\end{align}
By Proposition \ref{prop:spectral radius and norm in cstar}\,\ref{it:b}, the spectrum of \( A \) is a compact subset of the real line. Then we can write \( p \) restricted to this subset using the identity mapping in \( \cc{\spectr_{\mfA}(A)} \) such that
\begin{align}
    p = \sum_{k=0}^{n} a_{k} \ \id^{k} \qquad (\text{with } \id^{0} = 1).
\end{align}
Define \( \varphi _0 \colon \pol \to \mfA \) by
\begin{align}
    \varphi_0(p) = \sum_{k=0}^{n} a_{k} A^{k} = p(A) && \text{and} &&
    \varphi_0(p)^{\star}= \sum_{k=0}^{n} \bar{a}_{k} A^{k} = p(A)^{\star} .
\end{align}
Then, \( \varphi_{0} \) is linear and \( \varphi_0(\id^{k}) \mapsto A^{k} \) for \( k = \bbN \).
Also, \( \varphi_0(p) \) and \( \varphi_0(p)^{\star} \) commute, implying that \( \varphi_0(p) \) is  normal. Thus, Proposition \ref{prop:spectral radius and norm in cstar}\,\ref{it:a} applies and together with Proposition \ref{prop:spectral for polys} we have
\begin{align}
    \| \varphi_0(p) \| = r_{\mfA}(p(A))
        = \sup \{  |\lambda | \colon \lambda \in \spectr_{\mfA}(p(A))\}
        = \sup \{ |p(\lambda) | \colon \lambda \in \spectr_{\mfA}(A) \}
        = \| p \|,
\end{align}
where the norm on the left-hand side is the \Cs-norm on \( \mfA \) and the norm on the right-hand side is the supremum norm on \( \cc{\spectr_{\mfA}(A)} \). Thus, \( \varphi_0 \) is an isometry.
If \( p \) and \( q \) are two polynomials in \( \pol \), identically equal on \( \spectr_{\mfA}(A) \), then
\begin{align}
    \| \varphi_0(p) - \varphi_0(q) \| = \| p - q \| = 0,
\end{align}
and then \( \varphi_0(p) = \varphi_0(q) \). Thus,  the mapping \( \varphi_0 \) is well defined (unambiguous).

By the Weierstrass approximation theorem, \( \pol \subset \cc{\spectr_{\mfA}(A)} \) is everywhere dense in \( \cc{\spectr_{\mfA}(A)} \). Since \( \mfA \) is complete, \( \varphi_0 \) extends to the mapping \( \varphi \) defined on \( \cc{\spectr_{\mfA}(A)} \). This extension remains isometric; since  if \( f \in \cc{\spectr_{\mfA}(A)} \) and \( \{ p_n \} \) is a sequence of polynomials that approximates \( f \) then,
\begin{align}
    \| \varphi(f) \|
    = \lim_{n\to \infty} \| \varphi(p_n) \|
    = \lim_{n\to\infty} \| p_n \|
    = \| f \|.
\end{align}
(This proves existence of the asserted mapping.)

Thence, since  \( \cc{\spectr_{\mfA}(A)} \) is complete
its image, \( \{ \varphi(f) \colon f \in \cc{\spectr_{\mfA}(A)} \} \), is complete as well. Therefore, it is a commutative \Cs-subalgebra of \( \mfA \) that contains \( I \) and \( A \). We will denote this subalgebra \( \mfA(A) \). Since polynomials are everywhere dense in \( \cc{\spectr_{\mfA}(A)} \) then every element of \( \mfA(A) \) is a limit of a sequence of polynomials in \( A \). A closed \Cs-subalgebra \( \mfB \) of \( \mfA \) that contains \( I \) and \( A \) necessarily contains all polynomials in \( A \) and therefore it contains \( \mfA(A) \). Hence \( \mfA(A) \) is the smallest \Cs-subalgebra that contains \( I  \) and \( A \).
(This proves \ref{it:sc smallest comm} and \ref{it:sc lim of pol}.)

By definition, \( \varphi_0 \) is linear, multiplicative (i.e.\ \( \varphi(pq) = \varphi(p) \varphi(q) \) for \( p,q \in \pol \)), preserves the star  (i.e.\ \( \varphi(\bar{p}) = \varphi(p)^{\star} \)), and carries the unite of \( \pol \) onto that of \( \mfA \). These properties extend by continuity to \( \varphi \) and thus, \( \varphi \) is a \( \star \)-homomorphism. Since it is isometric it is one-to-one and thus, a \( \star \)-isomorphism. (This proves the isometry assertion.)

If \( \phi \) is another such \( \star \)-isomorphism that maps \( \id \) to \( A \) then \( \phi \) and \( \varphi \) coincide on the dense space \( \pol \); thus they are equal. (This proves uniqueness of the asserted mapping.)

Since we can approximate every element in \( \mfA(A) \) by polynomials in \( A \), it follows at once that \( \varphi(f) \) is a normal element, and that it is self-adjoint if and only if \( f \) is real valued. If \( B \in \mfA \)  commutes with \( A \) then it commutes with all polynomials in \( A \), and thus it commutes with every element in \( \mfA(A) \).
This concludes the proof.
\end{proof}

The mapping \( \varphi \) is called the function calculus for the self-adjoint element \( A \). To simplify the notation, it is often common to write \( f(A) \) instead of \( \varphi(f) \). Henceforth, we will use this convention.

We have already seen that a spectrum of a Banach algebra depends on the ambient Banach algebra. For \Cs-algebras this is not the case. Here, the spectrum of an element is independent of the \Cs-algebra in which we consider this element to be. Stated sloppy, \Cs-algebras are always large enough to contain an inverse of an element if it exists. With the aid of the function calculus we can now prove this statement.

\begin{prop}\label{prop:independence of spectrum}
Let \( \mfA \) be a \Cs-algebra, and \( \mfB \) be a \Cs-subalgebra of \( \mfA \). If \( B \) is in \( \mfB \), then
\begin{align}
\spectr_{\mfA}(B) = \spectr_{\mfB}(B).
\end{align}
\end{prop}

\begin{proof}
We have already discussed that \( \spectr_{\mfA}(B) \subseteq \spectr_{\mfB}(B) \). Thus we need only to prove the opposite. For this we will show that if \( A \in \mfB \) has an inverse in \( \mfA \) then this inverse is also in \( \mfB \).

We consider first a self-adjoint element \( A \). Since \( 0 \notin \spectr_{\mfA}(A) \), the function \( f(t) = t^{-1} \) is continuous on \( \spectr_{\mfA}(A) \). Thus, by the function calculus relative to the algebra \( \mfA \) we get an element \( f(A) \) in \( \mfA \). From Theorem \ref{thm:spectral calculus}\,\ref{it:sc smallest comm} this element is contained in \( \mfB \).
Since \( t f(t) = 1 \) for each \( t \) in \( \spectr_{\mfA}(A) \), it follows by multiplicativity of the function calculus that
\begin{align}
    I = \varphi(1) = \varphi(\id\ \cdot f) = \varphi(\id) \varphi(f) = A f(A).
\end{align}
Thus \( A^{-1} = f(A) \in \mfB \).

Now, consider a generic \( A \in \mfB \) that has an inverse \( A^{-1} \) in \( \mfA \). Then \( A^{\star} \) is also in \( \mfB \) and has an inverse \( (A^{-1})^{\star} \) in \( \mfA \). Now, \( A^{\star}A \) is a self-adjoint element in \( \mfB \) and it has an inverse \( A^{-1}(A^{-1})^{\star} \) in \( \mfA \). By the first part of the proof, \( A^{-1}(A^{-1})^{\star} \) is in \( \mfB \). Thus, \( A^{-1} = (A^{-1}(A^{-1})^{\star})A^{\star}  \in \mfB\).
\end{proof}

The above proposition shows that the spectrum of an element in a \Cs-algebra is independent of the \Cs-algebra this element is in. Thus, from now on, we can (and will) omit the subscript referring to the ambient \Cs-algebra and write \( \spectr(A) \) instead of \( \spectr_{\mfA}(A) \).

In the proof of Proposition \ref{prop:independence of spectrum} we use the phrase ``... by the function calculus relative to the algebra \( \mfA \) ...''. Implying, that the function calculus may depend on the original algebra. Nevertheless, using Proposition \ref{prop:independence of spectrum} we can see that this cannot happen.
Namely, if \( \mfB \) is a \Cs-subalgebra of \( \mfA \), and \( A \) is a self-adjoint element in \( \mfB \), then let \( \varphi \) be the spectral calculus \( \spectr(A) \to \mfA \) and \( \psi \) be the spectral calculus \( \spectr(A) \to \mfB \). However, since the spectrum of \( A \) relative to either algebra is the same, we can consider \( \psi \) as a spectral calculus into the larger algebra \( \mfA \). By uniqueness of the spectral calculus, \( \psi \)  must then coincide with \( \varphi \). Thus, we do not need to worry about the ambient \Cs-algebra when talking about the spectral calculus. In particular, we can always choose a sufficiently small \Cs-algebra, \( \mfA(A) = \{ f(A) \colon f \in \cc{\spectr(A)} \}\) such that the spectral calculus, \(\cc{\spectr(A)} \to \mfA(A) \), is onto. Then, we can state the following about the spectrum of \( f(A) \).

\begin{cor}\label{cor:spectr of function}
If \( A \) is a self-adjoint element of a \Cs-algebra \( \mfA \), and \( f \in \cc{\spectr(A)} \), then
\begin{align}
    \spectr(f(A) ) = \{ f(t)\colon \ t \in \spectr(A) \}.
\end{align}
\end{cor}
\begin{proof}
    The function calculus \( f \to f(A) \) is a \( \star \)-isomorphism from \(  \cc{\spectr(A)} \) onto the \Cs-subalgebra \( \mfA(A)\) of \( \mfA \), and \( \star \)-isomorphisms preserve the spectrum.
\end{proof}

Another useful consequence of the function calculus is the definition of a square root.

\begin{cor}\label{cor:square root}
Let \( A \) be a self-adjoint element of a \Cs-algebra \( \mfA \), with \( \spectr(A) \subseteq \bbR^{+} \). There exists a unique self-adjoint element \( H \) in the commutative \Cs-subalgebra \( \mfA(A) \) of \( \mfA \) such that \( A = H H \). Moreover,  \( \spectr(H) \subseteq \bbR^{+} \), and if \( B \in \mfA \) commutes with \( A \) then it also commutes with \( H \).
\end{cor}
\begin{proof}
Since \( \spectr(A) \) is a subset of positive real numbers, the function \( f(t) = t^{1/2} \big(= \sqrt{\id} \big) \) is real-valued, positive, and continuous on \( \spectr(A) \). Define \( H \) as \( f(A) \). If \( \varphi \) denotes the corresponding function calculus then
\begin{align}
    HH = f(A)f(A) = \varphi(f) \varphi(f) = \varphi(f^{2}) = \varphi(\id) = A.
\end{align}
Since, \( f \) is real valued, the element \( H \) is self-adjoint by \ref{thm:spectral calculus} \ref{it:sc normal}. Since \( f \) is positive the spectrum of \( H \) is positive by \ref{cor:spectr of function}. If \( B \in \mfA \) commutes with \( A \) then it commutes with \( H \) by \ref{thm:spectral calculus} \ref{it:sc comm with others}.

To show uniqueness, suppose that \( K \) is any element in \( \mfA \) with positive spectrum and such that \( K^{2} = A \). Let \( \{ p_{n} \} \) be a sequence of polynomials that converge to \( f = \sqrt{id} \). Set \( q_{n}(t) = p_{n}(t^{2})  \). Then we have,
\begin{align}
    \lim_{n\to \infty}  q_{n}(t) = \lim_{n\to \infty} p_{n}(t^{2}) = f(t^{2}) = t,
\end{align}
uniformly for \( t \in \spectr(K) \). Hence, by the function calculus for \( K \) we get
\begin{align}
    K = \lim_{n\to \infty} q_{n}(K) = \lim_{n \to \infty} p_{n}(K^{2})
        = \lim_{n \to \infty} p_{n}(A) = f(A) = H,
\end{align}
which proves uniqueness.
\end{proof}
There are more useful corollaries that we can derive from the function calculus. For now, however, we close this chapter and proceed to the next  ingredient that will be necessary to understand the construction of quantum mechanics: the GNS construction.

\begin{info}
    The function calculus can be extended to normal elements. The methods to prove it, however, are a bit different. We provide the full theorem without proof.

    \begin{thm}[\textbf{Function calculus for normal elements}]
    Let \( A \) be a normal element of a \Cs-algebra \( \mfA \). There is a unique isometric \( \star \)-isomorphism \( \varphi \) from \( \cc{\spectr_{\mfA}(A)} \) into \( \mfA \) such that \( \varphi(\id) = A \). Moreover, for \( f \in \cc{\spectr_{\mfA}(A)} \),
    \begin{enumerate}
        \item
        \( \varphi(f) \) is a normal element of \( \mfA \);

        \item
        the set \( \{ \varphi(f) \colon f \in \cc{\spectr_{\mfA}(A)} \}\) is a commutative \Cs-algebra and it is the smallest \Cs-subalgebra of \( \mfA \) that contains \( A \);

        \item
        \( \varphi(f) \) is the limit of a sequence of polynomials in \( I \), \( A \), and \( A^{\star} \);

        \item
        moreover, the element \( A \) is
            \vspace{-0.5cm}
    \begin{multicols}{2}
        \begin{itemize}
            \item
            self-adjoint iff \( \spectr(A) \subseteq \bbR\);
            \item
            positive iff \( \spectr(A) \subseteq \bbR^{+} \);
            \item
            unitary iff \( \spectr(A) \subseteq \bbC_1 \);
            \item
            projection iff \( \spectr(A) \subseteq \{ 0,1\} \).
        \end{itemize}
    \end{multicols}
\end{enumerate}
\end{thm}
\end{info}

\newpage
\subsection{The GNS construction}
In this section we will see that any abstract \Cs-algebra can be ``represented'' as an algebra of bounded linear operators on some appropriate Hilbert space. The construction used to show this was developed by Gelfand, Naimark, and Segal, hence the name GNS-construction. This construction is vital for our purpose: to use \Cs-algebras in applications to quantum mechanics. To discuss the GNS construction we need additional tools --- the order structure on the \Cs-algebra, the space of states, and a brief discussion of representations of a \Cs-algebras. We begin with the order structure.

\subsubsection{Order structure on the \Cs-algebra}
\begin{defi}
Let \( V \) be a vector space. A (positive) \textit{cone} \( \pC \) of \( V \) is a subset of \( V \) with the following properties:
\begin{enumerate}
    \item
    if \( A \) and \( -A \) are in \( \pC \), then \( A = 0 \);

    \item
    if \( a \) is a positive scalar and \( A  \) is in \( \pC \), then \( aA \) is in \( \pC \);

    \item
    if \( A, B \) are in \( \pC \) then \( A + B   \) is also in \( \pC \).
\end{enumerate}
\end{defi}

A cone on a vector space allows us to compare elements in that vector space. In particular, we can define an element \( A \in V \) to be ``smaller then'' \( B \in V \) whenever \( B - A \) is in the cone \( \pC \). This relation defines a \textit{partial ordering} on \( V \). An element \( I \in V \) is called an \textit{order unit} when, for a given \( A \in V \) it holds that \( -aI \leq A \leq aI \) for some positive scalar \( a \) that may depend on \( A \).

\begin{info}[Reminder: partial order.]
A partial order on a set \( \Omega \) is a relation over the set \( \Omega \) such that if \( a,b,c \in \Omega \):
\begin{enumerate}
\item
\( a \leq a \), that is the relation is reflexive;

\item
if \( a \leq b \) and \( b \leq a \) then \( a = b \), that is the relation is anti-symmetric;

\item
if \( a \leq b \) and \( b \leq c \) then \( a \leq c \), that is the relation is transitive.
\end{enumerate}
\end{info}

\begin{figure}[ht]
    \centering
    \subfloat[A cone in a vector space.]{%
        \def\svgwidth{7cm}
\begingroup%
  \makeatletter%
  \providecommand\color[2][]{%
    \errmessage{(Inkscape) Color is used for the text in Inkscape, but the package 'color.sty' is not loaded}%
    \renewcommand\color[2][]{}%
  }%
  \providecommand\transparent[1]{%
    \errmessage{(Inkscape) Transparency is used (non-zero) for the text in Inkscape, but the package 'transparent.sty' is not loaded}%
    \renewcommand\transparent[1]{}%
  }%
  \providecommand\rotatebox[2]{#2}%
  \newcommand*\fsize{\dimexpr\f@size pt\relax}%
  \newcommand*\lineheight[1]{\fontsize{\fsize}{#1\fsize}\selectfont}%
  \ifx\svgwidth\undefined%
    \setlength{\unitlength}{366.18125102bp}%
    \ifx\svgscale\undefined%
      \relax%
    \else%
      \setlength{\unitlength}{\unitlength * \real{\svgscale}}%
    \fi%
  \else%
    \setlength{\unitlength}{\svgwidth}%
  \fi%
  \global\let\svgwidth\undefined%
  \global\let\svgscale\undefined%
  \makeatother%
  \begin{picture}(1,0.94064793)%
    \lineheight{1}%
    \setlength\tabcolsep{0pt}%
    \put(0,0){\includegraphics[width=\unitlength,page=1]{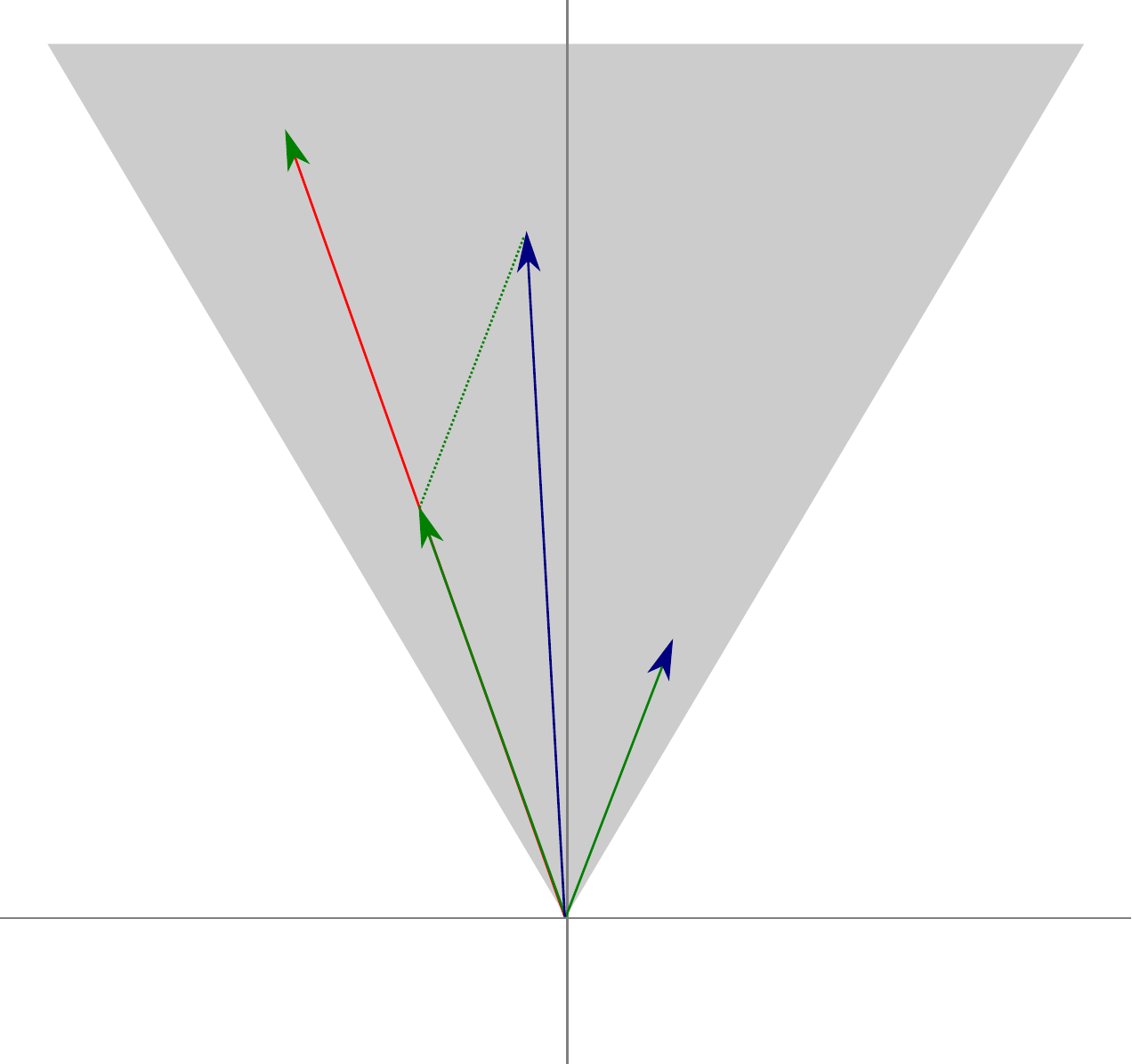}}%
    \put(0.39891702,0.46295132){\color[rgb]{0,0,0}\makebox(0,0)[lt]{\lineheight{1.25}\smash{\begin{tabular}[t]{l}$A$\end{tabular}}}}%
    \put(0.56905798,0.38704233){\color[rgb]{0,0,0}\makebox(0,0)[lt]{\lineheight{1.25}\smash{\begin{tabular}[t]{l}$B$\end{tabular}}}}%
    \put(0.39891702,0.74302963){\color[rgb]{0,0,0}\makebox(0,0)[lt]{\lineheight{1.25}\smash{\begin{tabular}[t]{l}$A + B$\end{tabular}}}}%
    \put(0.23924626,0.83726154){\color[rgb]{0,0,0}\makebox(0,0)[lt]{\lineheight{1.25}\smash{\begin{tabular}[t]{l}$aA$\end{tabular}}}}%
  \end{picture}%
\endgroup%

        \label{fig:cone}%
        }
        \hspace{1cm}
        \subfloat[Partial order induced by the cone.]{%
            \def\svgwidth{7cm}
\begingroup%
  \makeatletter%
  \providecommand\color[2][]{%
    \errmessage{(Inkscape) Color is used for the text in Inkscape, but the package 'color.sty' is not loaded}%
    \renewcommand\color[2][]{}%
  }%
  \providecommand\transparent[1]{%
    \errmessage{(Inkscape) Transparency is used (non-zero) for the text in Inkscape, but the package 'transparent.sty' is not loaded}%
    \renewcommand\transparent[1]{}%
  }%
  \providecommand\rotatebox[2]{#2}%
  \newcommand*\fsize{\dimexpr\f@size pt\relax}%
  \newcommand*\lineheight[1]{\fontsize{\fsize}{#1\fsize}\selectfont}%
  \ifx\svgwidth\undefined%
    \setlength{\unitlength}{366.18125102bp}%
    \ifx\svgscale\undefined%
      \relax%
    \else%
      \setlength{\unitlength}{\unitlength * \real{\svgscale}}%
    \fi%
  \else%
    \setlength{\unitlength}{\svgwidth}%
  \fi%
  \global\let\svgwidth\undefined%
  \global\let\svgscale\undefined%
  \makeatother%
  \begin{picture}(1,0.94064793)%
    \lineheight{1}%
    \setlength\tabcolsep{0pt}%
    \put(0,0){\includegraphics[width=\unitlength,page=1]{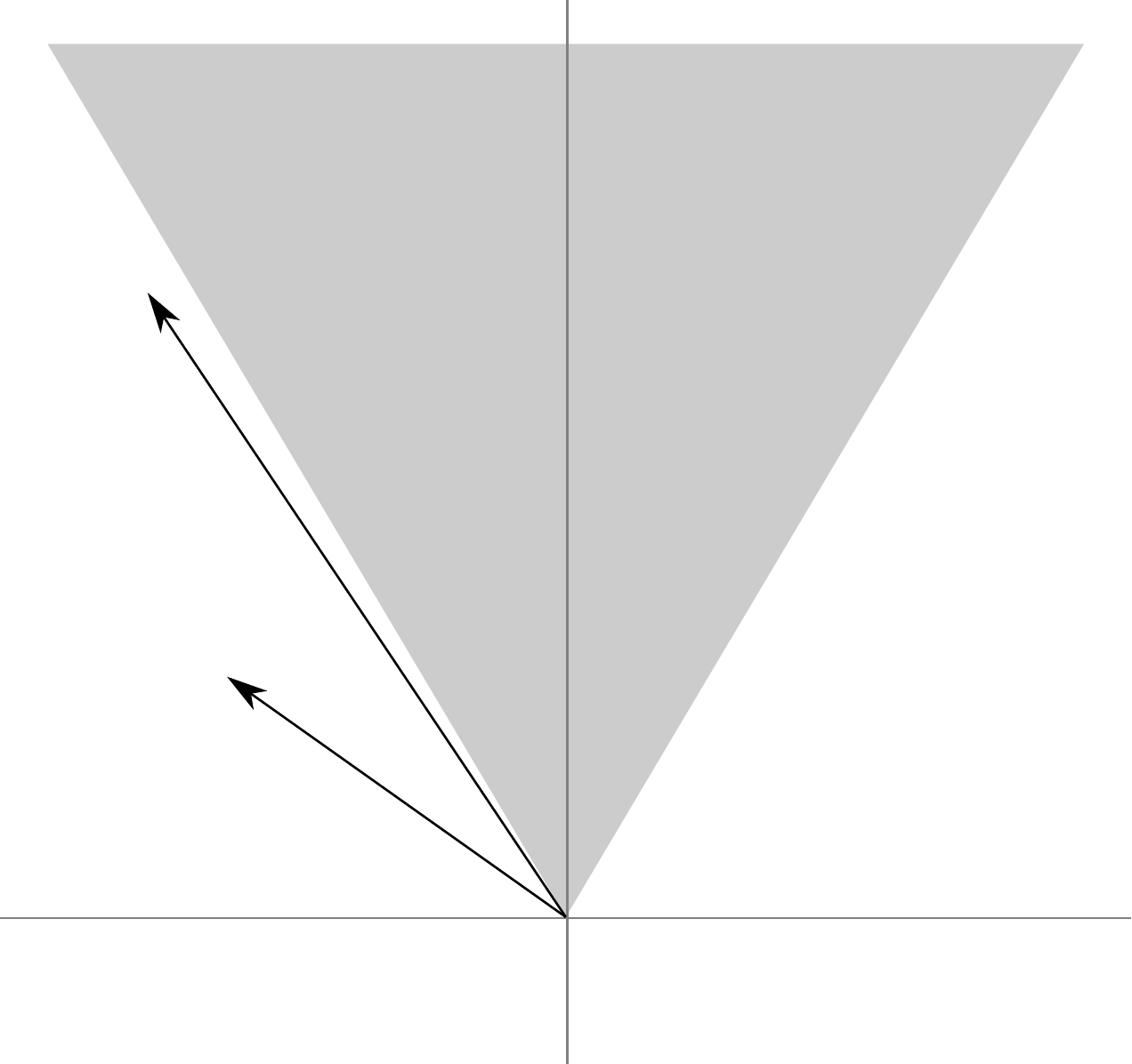}}%
    \put(0.17118982,0.35824915){\color[rgb]{0,0,0}\makebox(0,0)[lt]{\lineheight{1.25}\smash{\begin{tabular}[t]{l}$A$\end{tabular}}}}%
    \put(0.08742805,0.69329617){\color[rgb]{0,0,0}\makebox(0,0)[lt]{\lineheight{1.25}\smash{\begin{tabular}[t]{l}$B$\end{tabular}}}}%
    \put(0,0){\includegraphics[width=\unitlength,page=2]{./Pics/cone2.pdf}}%
    \put(0.41246005,0.49510855){\color[rgb]{0,0,0}\makebox(0,0)[lt]{\lineheight{1.25}\smash{\begin{tabular}[t]{l}$ B - A$\end{tabular}}}}%
  \end{picture}%
\endgroup%

            \label{fig:partial order}%
            }
    \caption{
    Cone structure and the partial order.}
    \label{fig:cone-order}
\end{figure}

\begin{expl}[Cone on \( \bbR^{2} \)]
Let \( V = \bbR^{2} \) a two dimensional Euclidean vector space. Define the cone \( \pC \) on \( V \) as the set of all those elements that in the canonical basis of \( \bbR^{2} \) have positive coefficients, that is
\begin{align}
    \pC = \{ v \in \bbR^{2} \colon v = (v_1,v_2), \text{ and } v_1 \geq 0, \ v_2 \geq 0 \}.
\end{align}
It is easy to check that \( \pC \) is indeed a cone, and thus this provides a partial ordering for the space \( \bbR^{2} \). However, this cone (and thus the partial ordering) depends on the basis and is not preserved under rotations.
\end{expl}

Since a (positive) cone defines a set of positive elements, it stands to reason that a set of positive elements --- defined in some meaningful way --- will  yield a cone. As is common, we can borrow the intuition of positive elements from bounded linear operators to define positive elements in a \Cs-algebra.

\begin{defi}[\textbf{Positive element}]
A \textit{positive element} of a \Cs-algebra is a self-adjoint element with a positive spectrum.
By \( \mfA^{+} \) we denote the set of all positive elements of \( \mfA \).
\end{defi}

\begin{expl}
Consider a \Cs-algebra \( \cc{X} \) on a compact Hausdorff space \( X \). Then \( f \in \cc{X} \) is self-adjoint if and only if it is real-valued. Since, the spectrum of \( f \) is its image, it follows that \( f \) is positive (in the \Cs-algebra sense) if and only if \( f(x) \geq 0 \) for each \( x \in X \). Thus the notion of positivity for elements in the \Cs-algebra \( \cc{X} \) is consistent with the notion of positivity for functions.
\end{expl}

\begin{expl}
If \( B \) is a bounded linear operator on some Hilbert space \( H \) then we call it positive, if \( \langle x , B x \rangle \geq 0 \) for every \( x \in H \). We have seen in our tutorial class, that \( B \) is positive if and only if it is self-adjoint and its spectrum is positive. Hence, the notion of positivity for elements in the \Cs-algebra \( B(H) \) is consistent with the notion of positivity for bounded linear operators.
\end{expl}

Even though, positivity is a statement about ``signs'', it has a very useful characterization in terms of the norm. We state this characterization in the following lemma.

\begin{lem}\label{lem:positivity}
    A self-adjoint element \( A \) in a \Cs-algebra is positive if and only if
    \begin{align}
         \| A - a I \| \leq a,
    \end{align}
     for some \( a \in \bbR \) with \( a \geq \| A \| \).
\end{lem}
\begin{proof}
Fix \( a \geq \| A \| \), then \( \spectr(A) \subseteq [-a, a] \), and
\begin{align}
    \| A - a I \| = r(A - aI) = \sup \{ | t - a |\colon t \in \spectr(A) \} = \sup \{ (a -t )\colon t \in \spectr(A) \}.
\end{align}
Where the first equality holds by \ref{prop:spectral radius and norm in cstar} \ref{it:a} and the second equality holds by \ref{prop:spectral for polys}.
Hence \( \| A - a I \| \leq a \) if and only if \( \spectr(A) \subseteq \bbR^{+} \).
\end{proof}

Now, we can show that the set of positive elements in a \Cs-algebra defines a closed positive cone.

\begin{thm}\label{thm:positive cone}
Let \( \mfA \) be a \Cs-algebra. Then,
\begin{enumerate}
    \item\label{it:positive cone a}
    \( \mfA^{+} \) is closed in \( \mfA \);

    \item\label{it:positive cone e}
    for \( A \in \mfA^{+} \), if \( -A \) is also in \( \mfA^{+} \) then \( A = 0 \).

    \item\label{it:positive cone b}
    for \( A \in \mfA^{+} \) and \(  a \in \bbR^{+} \) we have \( a A \in \mfA^{+} \);

    \item\label{it:positive cone c}
    for \( A, B \in \mfA^{+} \) we have \( A + B \in \mfA^{+} \);

    \item\label{it:positive cone d}
    for \( A,B \in \mfA^{+} \) and \( AB = BA \) we have \( AB \in \mfA^{+} \);

\end{enumerate}
\end{thm}

\begin{proof}
\ref{it:positive cone a}
From Lemma \ref{lem:positivity} we have
\begin{align}
    \mfA^{+} = \{ A \in \mfA \colon A = A^{\star} \text{ and } \big\| A - \|A \| I \big\|\leq \| A \| \}.
\end{align}
By continuity of the norm, \( \mfA^{+} \) is closed.

\ref{it:positive cone e}
If \( A \) and \( -A \) are both in \( \mfA^{+} \), then \( \spectr(A) \subseteq \bbR^{+} \cap (-\bbR^{+}) = \{ 0 \} \); so \( \| A \| = r(A) = 0 \).

\ref{it:positive cone b}
For \( A \in \mfA^{+} \) and \( a \geq 0 \), the element \( a A \) is self-adjoint and
\begin{align}
    \spectr(aA) = \{ a \lambda \colon \lambda \in \spectr(A) \} \subseteq \bbR^{+}.
\end{align}

\ref{it:positive cone c}
For \( A,B \in \mfA^{+} \) we have by Lemma \ref{lem:positivity} that
\begin{align}
    \big\| A - \| A \| I \big\| \leq \| A \| &&
    \big\| B - \| B \| I \big\| \leq \| B\|.
\end{align}
Therefore,
\begin{align}
    \big\| A + B - (\|A \| + \| B \|) I \big\| \leq \| A \| + \| B \|.
\end{align}
With \( a = \| A \| + \| B \| \geq \| A + B \| \) it follows again from Lemma \ref{lem:positivity} that \( A + B \in \mfA^{+} \).

\ref{it:positive cone d}
For \( A,B \in \mfA^{+} \), let \( H \) be the square root of \( A \) defined as in Corollary \ref{cor:square root}, and \( K \) be the square root of \( B \).
By Corollary \ref{cor:square root}, if \( A \) commutes with \( B \) then also \( H \) commutes with \( B \), and by the same corollary, \( K \) commutes with \( H \).
Thus, \( AB =HH KK =  HK HK = (HK)^{2}\). Since \( H \) and \( K \) are both self-adjoint and since they commute, \( HK \) is self-adjoint, and thus its spectrum is a subset of a real line. By Corollary \ref{cor:spectr of function} we have \( \spectr(AB) = \spectr \big( (HK)^{2} \big) = \big( \spectr(HK) \big)^{2} \subseteq \bbR^{+}\).
\end{proof}

Condition, \ref{it:positive cone e}, \ref{it:positive cone b}, and \ref{it:positive cone c} state that \( \mfA^{+} \) is a (positive) cone. Hence, we can use it to define partial ordering on the \Cs-algebra \( \mfA \). Using this ordering, we can say that an element \( A \) is positive if \( A \geq 0 \). Also notice, that in one of our exercises we have shown that \( \| A \|I - A \) is always positive, if \( A \) is self-adjoint. Hence, for a unital \Cs-algebra the unit of the algebra is also the order unit for this partial ordering,
\begin{align}
    -\| A \| I \leq A \leq \| A \|I.
\end{align}

\subsubsection{States of the \Cs-algebra}
So far, we have discussed the elements of the algebra. However, a great deal of structure on the \Cs-algebra comes from its dual space: the space of continuous linear functionals on the algebra. We have encountered the power of linear functionals in the Theorem \ref{thm:non-empty spectrum}. Now we approach the topic of linear functionals more systematically. With the aid of the positive cone on the \Cs-algebra we can define a cone on its dual space.

\begin{info}[Reminder.]
    We briefly recall some basic definitions in topology.
\begin{defi}[Convex set.]
A set \( Y \) in a vector space \( V \) is called \textit{convex} if \( x \) and \(y \) in \( Y \) imply that the line \(\{ b y_1 + (1-b) y_2 \colon 0\leq  b  \leq1  \} \)  is also in \( Y \).
\end{defi}

\begin{defi}[Locally convex space.]
A \textit{locally convex space} is a topological vector space, in which the topology has a base consisting of convex sets. That means:
Let \( V\) be a topological vector space with topology \( \tau \).
The basis of \( \tau \) is the family \( \beta \subset \tau \) such that every open set \( O \in \tau \) is of the form \( O = \cup_{\alpha} B_{\alpha} \) for some family \( \{ B_{\alpha} \}\subset \beta \).
If each set in \( \beta \) is convex, then \( V \) is a locally convex space.
\end{defi}
\begin{defi}[Extreme points.]
Let \( V \) be a locally convex space and \( Y \subset V \) a convex subset of \( V \). Then a point \( y_0 \in Y\)  is called an \textit{extreme point} of \( Y \) if  \( y_0 = (1-a) y_1 + a y_2 \), with \( 0< a <1 \) and \( y_1, y_2 \in Y \) implies \( y_1 = y_2 = y_0 \).
\end{defi}
\begin{thm}[Krein-Milman]
If \( X \) is a non-empty compact convex set in a locally convex space \( V \), then \( X \) has an extreme point. Moreover, \( X \) is the closure of the convex hull of all extreme points of \( X \).
\end{thm}
\end{info}

After this short reminder we are equipped with everything to dive into the discussion of linear functional on a \Cs-algebra. We begin with some general definitions.

\begin{defi}\label{defi:states}
Let \( V \) be a partially ordered vector space with order unit \( I \). Then,
\begin{enumerate}
\item\label{defi: states it1}
a linear functional \( \rho \) on \( V \) is called \textit{positive} if \( \rho(v) \geq 0  \) for every \( v \geq 0 \) with \( v \in V \);

\item\label{defi: states it2}
a positive linear functional \( \rho \) on \( V \) is called a \textit{state} if \( \rho(I) = 1 \);

\item
the \textit{state space} \( \mcS(V) \) is the set of all states on \( V \).
\end{enumerate}
\end{defi}

If \( \rho \) is a linear functional on a \Cs-algebra \( \mfA \), we can define another linear functional \( \rho^{\star} \), by the equation \( \rho^{\star} (A) = \overline{\rho(A^{\star})} \) for every \( A \in \mfA \).

\begin{defi}
Let \( \mfA \) be a \Cs-algebra and \( \rho \) a linear functional on it.
If  \( \rho(A) = \overline{\rho(A^{\star})} \) for every \( A \in \mfA \),
 we call \( \rho \) \textit{hermitian}.
\end{defi}

Recall that every element of a \Cs-algebra uniquely decomposes into its real and imaginary part. Using this decomposition we can rephrase the condition for a state to be hermitian only in terms of self-adjoint operators:
let \( A \) be an element in a \Cs-algebra \( \mfA \) with the real part \( H \) and imaginary part \( K \), then for a linear functional \( \rho \) on \( \mfA \) we have,
\begin{align}
    \rho(A) = \rho(H + \imath K) =  \rho(H) + \imath \rho(K), &&
    \overline{\rho(A^{\star})} = \overline{\rho(H - \imath K)} = \overline{\rho(H)} + \imath \overline{\rho(K)}.
\end{align}
It follows that \( \rho \) is hermitian if and only if \( \rho(H) = \overline{\rho(H)} \) for every self-adjoint \( H \in \mfA \).

Using this characterization of hermitian functionals we can deduce that a positive linear functional is always hermitian: Let \( A \) be self-adjoint in \( \mfA \) and \( \rho \) a positive linear functional on \( \mfA \). Since \( \|A \| I \pm A \geq 0 \) (as we already have seen) it follows from positivity of \( \rho \) that \( \rho(\|A \| I \pm A) \geq 0 \). In particular, \( \rho(\|A \| \pm A) \) is real, and from
\begin{align}\label{hermitian functional}
    \rho(A) = \frac{1}{2} \Big( \rho\big(\|A \| I + A\big) - \rho\big(\| A \| I - A\big) \Big),
\end{align}
it follows that \( \rho(A) \) itself is real.

Additionally, we can characterize positive linear functionals by their norm as follows.
\begin{thm}\label{thm:positive functionals}
A linear functional \( \rho \) on a (unital) \Cs-algebra \( \mfA \) is positive if and only if \( \rho \) is bounded and \( \| \rho \| = \rho(I) \).
\end{thm}

\begin{proof}
    For the ``only if'' direction, assume that \( \rho \) is positive. Let \( A \) be in \( \mfA \), choose \( a \in \bbC \) such that \( | a | =1 \) and \( a \rho(A) \geq 0 \), and let \( H \) be the real part of \( aA \). Since \( \| H \| \leq \| A \|  \) we have \( H \leq \| H \| I \leq \| A \| I \) and thus
\begin{align}
    \| A \| \rho(I) - \rho(H) = \rho(\| A \| I - H) \geq 0.
\end{align}
Since \( \rho \) is hermitian we have \( \rho(\bar{a}A^{\star}) =  \overline{\rho(aA)} = \rho(aA)  \) and consequently,
\begin{align}
    | \rho(A) |
        = \frac{1}{2} \rho(aA + \bar{a}A^{\star}) = \rho(H) \leq \rho(I) \| A \|.
\end{align}
Thus, \( \rho \) is bounded, with \( \| \rho \| \leq \rho(I) \). Simultaneously, \( |\rho(I)| \leq \| \rho \| \); implying \( \| \rho \| = \rho(I) \).

For the ``if'' direction, suppose that \( \rho \) is bounded and \( \| \rho \| = \rho(I) \). After suitable rescaling, it is enough to consider the case in which \( \| \rho \| = \rho(I) = 1 \). With \( A \in \mfA^{+}\) and some \( a,b \in \bbR \) assume that \( \rho(A) = a + \imath b \). We want to show that \( a \geq 0 \) and \( b = 0 \). For this we observe that since \( \spectr(A) \subset \bbR^{+} \), there is a small positive number \( s \) such that,
\begin{align}
    \spectr(I -sA) = \{ 1 - st \colon t \in \spectr(A) \} \subseteq [0,1].
\end{align}
So by the spectral radius formula we have \( \| I - sA \| = r(I - sA) \leq 1 \). From the continuity of \( \rho \) we then have
\begin{align}
    1 -sa \leq | 1 - s(a+\imath b)| = | \rho(I-sA) | \leq \| I - sA \| \leq 1,
\end{align}
and \( a \geq 0 \). To show that \( b = 0 \) we proceed as in the proof of Proposition \ref{prop:spectral radius and norm in cstar} \ref{it:b}. For \( n \in \bbN \) consider the element  \( B_n =  A + \imath n bI \) in \( \mfA \). Then,
\begin{align}
    \| B_{n} \|^{2} = \| B^{\star}_{n} B_{n} \| = \| A^{2} + n^{2} b^{2} I \| \leq \| A \|^{2} + n^{2} b^{2},
    \label{eq:bound on bb}
 \end{align}
and we get for every \( n = 1,2,\dots \)
 \begin{align}
     a^{2} + (n^{2} + 2n + 1)b^{2}
     =|a + \imath b + \imath n b|^{2}
     = | \rho(B_{n})|^{2} \leq \| A \|^{2} + n^{2}b^{2},
 \end{align}
where the last inequality follows from the continuity of \( \rho \) and \eqref{eq:bound on bb}. Since \( A \) is bounded, \( b \) has to vanish.
\end{proof}

With this understanding we can summarize the similarities and differences between positive elements in the algebra and its dual.
\begin{center}
\begin{minipage}[l]{0,45\textwidth}
\textbf{An element \( A \) in \( \mfA \) is positive if}
\begin{enumerate}
    \item it is self-adjoint,
    \item \( \spectr(A) \subset \bbR^{+} \),
    \item \( \| A - a I \| \leq a  \).
\end{enumerate}
\end{minipage}\hfill
\begin{minipage}[r]{0,45\textwidth}
\textbf{An element \( \rho \) in \( \mfA' \) is positive if}
\begin{enumerate}
    \item it is hermitian,
    \item \( \rho(A) \geq 0 \) whenever \( A \in \mfA^{+} \),
    \item \( \rho(I) = \| \rho \| < \infty \).
\end{enumerate}
\end{minipage}
\end{center}

Before we proceed with the analysis of states, we recall some ideas regarding the topologies on vector spaces.

\begin{info}[Reminder: \ws-topology.]
Let \( V \) be a linear space and let \( \mcF \) be a family of linear functionals on \( V \). Assume that \( \mcF \) separates points of \( V \) in the sense that if \( x \in V \) is non-zero then there is a \( \rho \in \mcF \) such that \( \rho(x) \neq 0 \).
We define the topology \( \sigma(V, \mcF) \) as the coarsest (weakest, one with the fewer open sets) topology on \( V \) relative to which each element of \( \mcF \) is a continuous mapping from \( V \) into \( \bbC \).

If \( V \) is a locally convex topological vector space, we denote its continuous (topological dual) as \( V' \). If \( x \) is in \( V \) then the equation
\begin{align}
    \hat{x}(\rho) = \rho(x) \qquad (\rho \in V'),
\end{align}
defines a linear functional \( \hat{x} \) on the dual \( V' \). The set \( \hat{V} = \{ \hat{x} \colon x \in V \) \}
is a linear subspace of \( V'' \) --- the algebraic dual space of \( V' \)  (the space of not necessarily continuous functionals on \( V' \)). The set \( \hat{V} \) separates points of \( V' \) since, if \( \rho \in V' \) and \( \hat{x}(\rho) = 0 \) for all \( x \in V \) then \( \rho(x) = 0 \) for all \( x \in V \) and thus \( \rho = 0 \).
The \ws-topology on \( V' \) is the topology \( \sigma(V', \hat{V}) \). We often denote it simply by \( \sigma(V', V) \).
Hence, the \ws-topology is the coarsest topology on \( V' \) for which each functional \( \hat{x} \) is continuous.
\end{info}

A bounded closed set in an infinite dimensional space is not always compact. However, in the \ws-topology this is indeed true (see the useful Lemma below). This makes the \ws-topology an important tool in analysis.

\begin{info}[Useful lemma.]
    Let \( X \) be a normed space and \( X' \) be the topological dual of \( X \). If \( S \) is a bounded \ws-closed subset of \( X' \), then \( S \) is \ws-compact.
\end{info}

From Theorem \ref{thm:positive functionals} we see that a state on \( \mfA \) is a bounded linear functional on \( \mfA \), with \( \| \rho \| = 1 \). Hence, \( \mcS(\mfA) \) is contained in the unit sphere of the Banach dual space \( \mfA' \) and is therefore bounded; it is also convex, since for every \( \rho_{1}, \rho_{2} \in \mcS(\mfA) \) and \( 0 \leq a \leq 1 \) the functional \( \rho = a \rho_{1} + (1-a) \rho_{2} \) is again positive (since positive elements form a cone) and \( \rho(I) = a + (1 - a) = 1 \). Additionally, \( \mcS(\mfA) \) is \ws-closed. To see this, observe that
\begin{align}\label{eq:state space 1}
    \mcS (\mfA) = \{ \rho \in \mfA' \colon \rho(I) = 1, \ \rho(A) \geq 0 \text{ when } A \in \mfA^{+} \}.
\end{align}
For \( A \in \mfA \) let \( \hat{A} \) be the dual element of \( \mfA' \) such that \( \rho(A) = \hat{A}(\rho) \) for every \( \rho \in \mfA' \).
Since each \( \hat{A} \) is by definition continuous in the \ws-topology, the set
 \( C_{A} =  \hat{A}^{-1}\big( [0,\infty) \big) \) --- the pre-image of \( [0,\infty) \) under \( \hat{A} \) --- is closed, and so is the set \( C_{0} =  \hat{I}^{-1}(\{ 1\} ) \). Then \( \mcS(\mfA) = \big( \cap_{A \in \mfA^{+}} C_{A} \big) \cap C_{0} \) is an intersection of \ws-closed sets and thus is \ws-closed.
In summary: \( \mcS(\mfA) \) is a bounded, \ws-closed, convex subset of \( \mfA' \), and thus from the ``useful lemma'' above, it is \ws-compact. By the Krein-Milman theorem \( \mcS(\mfA) \) is the closed convex hull of its extreme points.
\begin{defi}\label{defi: states it3}
An extreme point of \( \mcS(\mfA) \) is called a \textit{pure state} of \( \mfA \). The set of all pure states of \( \mfA \) will be \( \mcP(\mfA) \). The \ws-closure of \( \mcP(\mfA) \) is called the \textit{pure state space} of \( \mfA \).
\end{defi}
\begin{figure}[ht]
\centering
\def\svgwidth{7cm}
\begingroup%
  \makeatletter%
  \providecommand\color[2][]{%
    \errmessage{(Inkscape) Color is used for the text in Inkscape, but the package 'color.sty' is not loaded}%
    \renewcommand\color[2][]{}%
  }%
  \providecommand\transparent[1]{%
    \errmessage{(Inkscape) Transparency is used (non-zero) for the text in Inkscape, but the package 'transparent.sty' is not loaded}%
    \renewcommand\transparent[1]{}%
  }%
  \providecommand\rotatebox[2]{#2}%
  \newcommand*\fsize{\dimexpr\f@size pt\relax}%
  \newcommand*\lineheight[1]{\fontsize{\fsize}{#1\fsize}\selectfont}%
  \ifx\svgwidth\undefined%
    \setlength{\unitlength}{490.18952095bp}%
    \ifx\svgscale\undefined%
      \relax%
    \else%
      \setlength{\unitlength}{\unitlength * \real{\svgscale}}%
    \fi%
  \else%
    \setlength{\unitlength}{\svgwidth}%
  \fi%
  \global\let\svgwidth\undefined%
  \global\let\svgscale\undefined%
  \makeatother%
  \begin{picture}(1,0.56141447)%
    \lineheight{1}%
    \setlength\tabcolsep{0pt}%
    \put(0,0){\includegraphics[width=\unitlength,page=1]{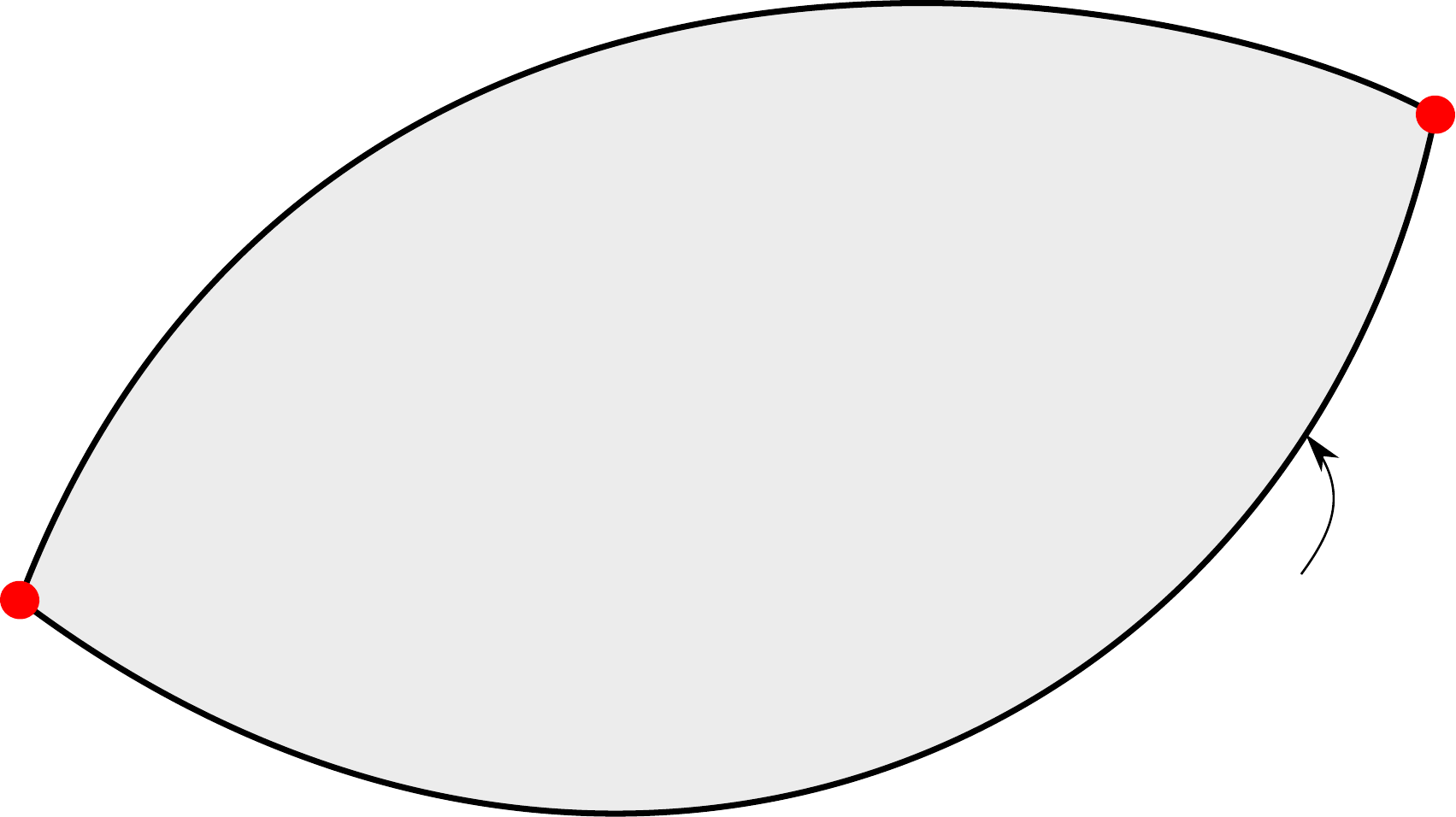}}%
    \put(0.84199091,0.10558595){\color[rgb]{0,0,0}\makebox(0,0)[lt]{\lineheight{1.25}\smash{\begin{tabular}[t]{l}$\mcP(\mfA)$\end{tabular}}}}%
    \put(0.03531666,0.44630428){\color[rgb]{0,0,0}\makebox(0,0)[lt]{\lineheight{1.25}\smash{\begin{tabular}[t]{l}$\mcS(\mfA)$\end{tabular}}}}%
    \put(0.33254362,0.24797689){\color[rgb]{0,0,0}\makebox(0,0)[lt]{\lineheight{1.25}\smash{\begin{tabular}[t]{l}mixed states\end{tabular}}}}%
  \end{picture}%
\endgroup%

\caption{Intuition for the space of states \( \mcS(\mfA) \). The pure states lie on the black solid boundary line. The line, however, is not closed; it may miss some points, which are indicated in red. The \ws-closure of that line is the whole boundary and it is the space of pure states \( \mcP(\mfA) \). In particular, the space of pure states may contain elements which are not pure states (the red points). The interior of the set is sometimes called the space of mixed states.}
\label{fig:space of states}
\end{figure}
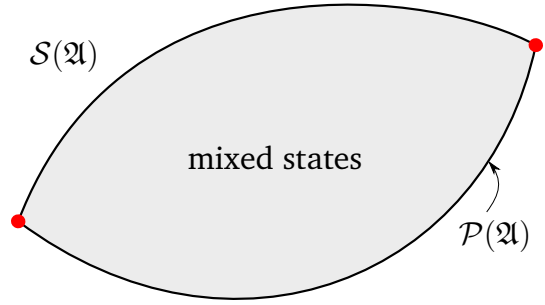

Visually we can picture the state space as in Figure \ref{fig:space of states}.
In general, the set of pure states is not \ws-closed, and \( \mcP(\mfA) \) can contain states which are not pure. For commutative \Cs-algebras, however, this does not happen. So if \( \mfA \) is commutative then the space of all pure states is already \ws-closed. This observation relates to a  deep result attributed to Gelfand --- sometimes called the Gelfand representation of a \Cs-algebra.
We state the result without proof.

\begin{thm}[\textbf{Gelfand representation.}]\label{thm:Gelfand rep}
Let \( \mfA \) be a commutative \Cs-algebra, and \( \mcP(\mfA) \) be the set of all pure states of \( \mfA \). Then \( \mcP(\mfA) \) is a compact Hausdorff space, relative to the \ws-topology. For \( A \in \mfA \) let \( \hat{A} \) be a complex-valued function on \( \mcP(\mfA) \) defined by \( \hat{A}(\rho) = \rho(A) \).  The mapping \( A \to \hat{A} \) is a \( \star \)-isomorphism from \( \mfA \) onto the \Cs-algebra \( \cc{\mcP(\mfA)} \).
\end{thm}

\begin{info}[Remark.]
    By this construction we can associate a compact Hausdorff space to every commutative \Cs-algebra. This can be used to show that every Hausdorff space admits a compactification --- the Stone-\v{C}eck compactification.
\end{info}

\begin{info}[Remark.]
If \( \mfA \) is a commutative \Cs-algebra then  by the Gelfand representation it is \( \star \)-isomorphic to the algebra \( \cc{\mcP(\mfA) } \) of continuous functions on the pure state space \( \mcP (\mfA) \) of \( \mfA \). Thus, if \( \omega \) is a pure state of \( \mfA \) and \( f \) is in \( \cc{\mcP(\mfA)} \) then there is a functional \( \delta_{\omega}(f) = f(\omega) \). We call such functionals \textit{point evaluations}. Therefore a pure state of a commutative \Cs-algebra \( \mfA \) corresponds to a point evaluation on the algebra \( \cc{\mcP(\mfA)} \). Indeed it is a general result required to prove the Gelfand representation, namely: if \( \cc{X} \)  is a \Cs-algebra of continuous functions on a compact Hausdorff space \( X \) then a state \( \omega \in \mcS (\mfA) \) is pure if and only if it is a point evaluation, i.e. there is an \( x \in X \) such that \(\omega(f) = \delta_{x} (f) = f(x) \) for every \( f \in \cc{X} \).
\end{info}

\subsubsection{Representations of a \Cs-algebra}

Up to now we discussed abstract \Cs-algebras. Now, we will see that every \Cs-algebra can be represented as a (sub)algebra of bounded linear operators on some Hilbert space. From the physicist point of view, it is this realization that allows us to formulate quantum and classical mechanics in terms of algebras. After all, as a physicist we understand operators in terms of how they act, not in terms of to what space they belong.

\begin{defi}
By a \textit{representation} of \Cs-algebra \( \mfA \) on a Hilbert space \( H \), we mean a \( \star \)-homomorphism \( \varphi \) from \( \mfA \) onto \( B(H) \). If in addition, \( \varphi \) is a \( \star \)-isomorphism (that is one-to-one) then it is called a \textit{faithful} representation.
\end{defi}

The Gelfand-Naimark theorem assures that every \Cs-algebra has a faithful representation on a Hilbert space. However, we only will discus the construction that provides a (possibly not faithful) representation.

The following terminology is common when discussing representations.
Let \( \varphi \) be a representation of a \Cs-algebra \( \mfA \) on a Hilbert space \( H \). Recall from our assignment 4 i) that every \( \star \)-homomorphism between \Cs-algebras is continuous, and so \( \varphi \) is continuous, i.e.  \( \| \varphi(A) \|_{H} \leq \| A \|_{\mfA} \) for every \( A \in \mfA \). Also recall, that a \( \star \)-isomorphism (an injective \( \star \)-homomorphism) is isometric, and so if \( \varphi \) is faithful \( \| \varphi(A) \|_{H} = \| A \|_{\mfA} \). If you recall the definition of ideals from assignment sheet 5, you realize that the set \( \big\{ A \in \mfA \colon \varphi(A) = 0 \big\} \) is a closed two-sided ideal in \( \mfA \). We call this set the \textit{kernel of} \( \varphi \). If there is a vector \( x \in H \) for which the linear subspace
\begin{align}
    \varphi(\mfA)x = \big\{ \varphi(A)x \colon A \in \mfA \big\}
\end{align}
is everywhere dense in \( H \), then \( \varphi \) is said to be a \textit{cyclic representation}, and \( x \) is said to be the \textit{cyclic vector} for \( \varphi \).

\begin{expl}
Recall, that we defines a \( \star \)-homomorphism as a mapping that among other things maps preserves the unit of the algebra. If for a second, we do not consider this requirement then the map \( \varphi_{0} (A) = 0 \) for every \( A \in \mfA \) is defines a trivial homomorphism and we can include to say that \( \varphi_{0} \) is a trivial representation.
\end{expl}

\begin{expl}
Let \( \mfA \) be a subalgebra of bounded linear operators \( B(H) \) on the Hilbert space \( H \). Then the inclusion mapping \( \varphi(A) = A \) for every \( A \in \mfA \) is a representation. This representation is faithful, since if \( A \in \mfA \) is not zero, then it is also non-zero in \( B(H) \).

Let \( H \) be a Hilbert space \( K \) be a closed subspace of \( H \), that is invariant under each operator from \( \mfA \), i.e.  the set \( \big\{ Ax \colon A \in \mfA, \ x \in K \big\} \) is a subset of \( K \). Let \( A\vert_{K} \) be the restriction of the operator \( A \) to the subspace \( K \). It can be checked that this restriction is a \( \star \)-homomorphism and thus \( A \to A\vert_{K} \colon \mfA  \to B(K) \) is a representation of \( \mfA \) on the Hilbert space \( K \).
In the same way one can see that the orthogonal complement \( K^{\perp} \) is also invariant under \( \mfA \) and thus the restriction \( A \to A\vert_{K^{\perp}} \) is a representation of \( \mfA \) on \( K^{\perp} \).

If the inclusion representation \( A \in \mfA \to A \in B(H) \) is cyclic then also the restriction representations are cyclic as well. To see this let \( P \) be the orthogonal projection from \( H \) onto \( K \) and \( x \in H \) be the cyclic vector \( x \in H \) for the inclusion representation. If \( Px = 0 \) then \( x \) is in \( K^{\perp} \). This, violates the fact that the inclusion representation is cyclic since \( K^{\perp} \) is invariant under the action of \( \mfA \). Thus \( Px \neq 0 \) and \( Px \in K\) is a cyclic vector for the representation \( A \to A\vert_{K} \), and \( (I - P)x \) is a cyclic vector for the representation \( A \to A\vert_{K^{\perp}} \).
\end{expl}

For the following discussion it will be convenient to define an inner product as a sesquilinear form which is not necessarily definite. In the following reminder we quickly recall the properties of the inner product. We also refresh that the well known Cauchy-Schwarz inequality does not depend on the definiteness of the inner product.

\begin{info}[Reminder: Inner product and the Cauchy-Schwarz inequality.]
    In the following we use the convention that an inner product on a complex vector space \( V \) may not be definite. That is, we say that a mapping \( (v,w) \to \scalp{v,w} \) from  \( V \times V \) to \( \bbC \) is an \textit{inner product} on \( V \) when it is
\begin{enumerate}
\item positive: \( \scalp{v,v} \geq 0 \);
\item hermitian: \( \scalp{v,w} = \overline{\scalp{w,v}} \);
\item sesquilinear: \( \scalp{v,\alpha w + \beta z} = \alpha \scalp{A,B} + \beta \scalp{A,C} \);
\end{enumerate}
with  \( v,w,z \in V \) and \( \alpha,\beta \in \bbC \).
If in addition it is
\begin{enumerate}[start=4]
\item definite: \(\scalp{v,v} = 0\) only when \( v = 0 \),
\end{enumerate}
then we call \( \scalp{\ , \ } \) a definite inner product.
\begin{prop}\label{prop:Cauch Schwarz}
    Let \( V \) be a complex vector space.
Suppose that \( \scalp{ \ , \ } \) is an inner product on \( V \). Then,
\begin{enumerate}
\item\label{prop:inprod it:1}
the Cauchy-Schwarz inequality,
\( | \scalp{v,w} |^{2} \leq \scalp{v,v} \scalp{w,w} \), holds for \( v,w \in V \);

\item\label{prop:inprod it:2}
the set \( \mfL = \{ v \in V \colon \ \scalp{v,v} = 0 \} \) is a linear subspace of \( V \);

\item\label{prop:inprod it:3}
the equation
\begin{align}
    \scalp{v + \mfL, w + \mfL}_{q} = \scalp{v, w}
    \qquad \big( v,w \in V \big)
\end{align}
defines a definite inner product \( \scalp{\ , \ }_{q} \) on the quotient space \( V \slash \mfL \).
\end{enumerate}
\end{prop}
\begin{proof}
For the proof of the above Proposition consult e.g.\ \cite[prop. 2.1.1]{kadison1997fundamentalsI}
\end{proof}
\end{info}

\begin{prop}
    Let \( \mfA \) be a \Cs-algebra, \( \rho \) a state on \( \mfA \), and \(
    \mfL_{\rho} = \big\{ A \in \mfA \colon \rho (A^{\star} A) = 0 \big\}
    \).
The state \( \rho \) provides a definite inner product on the quotient space \( \mfA \slash \mfL_{\rho} \) by the relation
\begin{align}
\scalp{A + \mfL_{\rho}, B + \mfL_{\rho}} = \rho(A^{\star}B)
\qquad \big( A,B \in \mfA \big),
\end{align}
Moreover, \( \mfL_{\rho} \) is a closed left ideal in \( \mfA \), and \( \rho (A^{\star} B) =0 \) whenever \( A \in \mfL_{\rho} \) and \( B \in \mfA \).
\end{prop}

\begin{proof}
For \( A,B \in \mfA \) define \( \scalp{ A, B}_{\rho} = \rho(A^{\star}B) \). Then \( \scalp{ \ ,\  }_{\rho} \) is positive since \( \rho \) is positive, hermitian since \(
\rho \) is hermitian, and sesquilinear by definition. Thus, \( \scalp{\ , \ } \) is an inner product on \( \mfA \)  with \( \mfL_{\rho} = \big\{ A \in \mfA \colon \scalp{A,A}_{\rho} = 0 \big\} \). By Proposition \ref{prop:Cauch Schwarz}, the set \( \mfL_{\rho} \) is a linear subspace of \( \mfA \), and
\(
\scalp{A + \mfL_{\rho}, B + \mfL_{\rho}}
    = \scalp{A,B}_{\rho}
    = \rho(A^{\star} B)
\)
is a definite inner product on \( \mfA \slash \mfL_{\rho} \).

With \( A \in \mfL_{\rho} \) and \( B \in \mfA \) we have by the Cauchy-Schwarz inequality
\begin{align}
    | \rho(A^{\star} B) |^{2}  \leq \rho(A^{\star}A) \rho(B^{\star}B) = 0,
\end{align}
and so \( \rho(A^{\star}B) = 0 \). It follows that
\begin{align}
    \rho\big((BA)^{\star} BA\big) = \rho \big( A^{\star} (B^{\star}BA) \big) = 0,
\end{align}
and thus \( BA \) is in \( \mfL_{\rho} \) whenever \( A \in \mfL_{\rho}\)  and \( B \in \mfA \). Hence, \( \mfL_{\rho} \) is a left ideal of \( \mfA \), which is closed since \( \rho \) and the mapping \( A \to A^{\star}A \) are continuous.
\end{proof}

We will call the space \( \mfL_{\rho} \) the \textit{left kernel} of the state \( \rho \).
We can now summarize the idea of the GNS-construction in a loosely way before proving it rigorously. The preceding Proposition provides us with means to define a (definite) inner product on \( \mfA \). With the norm induced by this inner product, the vector space \( \mfA \) may not be closed, but after closing it we obtain a Hilbert space. Moreover, since \( \mfA \) is an algebra we can define an action of \( A\in \mfA \) on the vector space \( \mfA \) by setting \(B \to AB \colon \mfA \to \mfA\). After proving that this mapping is continuous we can extend it to the whole \( H \). This will give us the desired representation. As always, however, the devil lies in the complement of a dense set...

\begin{thm}[\textbf{GNS representation}]\label{thm:GNS rep}
For every state \( \rho \) on a \Cs-algebra \( \mfA \) there is a cyclic representation \( \pi_{\rho} \) of \( \mfA \) on a Hilbert space \( H_{\rho} \) and a unit cyclic vector \( x_{\rho} \) for \( \pi_{\rho} \) such that
\begin{align}
    \rho(A) = \scalp{x_{\rho} , \pi_{\rho}(A) x_{\rho}}
    \qquad \big( A \in \mfA \big).
\end{align}
If \( \varphi \) is a cyclic representation of \( \mfA \) on a Hilbert space \( H \) with the unit cyclic vector \( x \) for \( \pi \) such that
\begin{align}
    \rho(A) = \scalp{x, \varphi(A) x}
    \qquad \big( A \in \mfA \big),
\end{align}
then there is an isomorphism \( U\colon H_{\rho} \to H \) from \( H_{\rho} \) onto \( H \)  such that
\begin{align}
    x = Ux_{\rho} && \text{and} && \varphi(A) = U \pi_{\rho}(A) U^{\ast}
    \qquad \big( A \in \mfA \big).
\end{align}
\end{thm}

\begin{proof}
Let \( \mfL_{\rho} \) be the left kernel of \( \rho \). The quotient linear space \( \mfA \slash \mfL_{\rho} \) is a pre-Hilbert space relative to the definite inner product defined by
\begin{align}
    \scalp{A + \mfL_{\rho}, B + \mfL_{\rho}} = \rho(A^{\star}B)
    \qquad \big( A,B \in \mfA \big).
\end{align}
Define \( H_{\rho} \) as the completion of \( \mfA \slash \mfL_{\rho} \) in the norm induced by the above scalar product.

To show that for \( A \in \mfA \) the mapping \( B + \mfL_{\rho} \to AB + \mfL_{\rho} \) defines an unambiguous operator on \( H_{\rho} \), we need to show that it does not depend on the ``representative'' \( B \) of the set \( B + \mfL_{\rho} \). For this assume \( B, C \in \mfA \) such that \( B + \mfL_{\rho} = C + \mfL_{\rho} \). Then \( B - C \) is in \( \mfL_{\rho} \) and since \( \mfL_{\rho} \) is a left ideal of \( \mfA \) we have for every \( A \in \mfA \) that \( AB - AC \in \mfL_{\rho} \) and consequently \( AB + \mfL_{\rho} = AC + \mfL_{\rho} \). Hence, the equation \( \pi(A)(B+\mfL_{\rho}) = AB + \mfL_{\rho} \) defines a linear operator \( \pi(A) \) acting on the pre-Hilbert space \( \mfA \slash \mfL_{\rho} \).

Next, we want to show that \( \pi \) is continuous.
For this recall that
\begin{align}
    \| A \|^{2} I - A^{\star}A = \| A^{\star} A \| I - A^{\star}A \in \mfA^{+},
\end{align}
implies that the positive square root \( K = \big( \| A \|^{2} I - A^{\star}A  \big)^{1/2}  \) exists, and \( B^{\star} (\| A \|^{2}I - A^{\star}A)B  = (KB)^{\star} (KB)  \in \mfA^{+} \). Therefore for each \( A,B \in \mfA \) we have
\begin{align*}
    \| A \|^{2} \| B + \mfL_{\rho} \|^{2} - \| \pi(A) (B + \mfL_{\rho}) \|^{2}
    &=\| A \|^{2} \| B + \mfL_{\rho} \|^{2 } - \| AB + \mfL_{\rho} \|^{2} \\
    &= \| A \|^{2} \scalp{B + \mfL_{\rho} , B + \mfL_{\rho}} - \scalp{AB + \mfL_{\rho}, AB + \mfL_{\rho}} \\
    &= \| A \|^{2} \rho(B^{\star}B) - \rho(B^{\star}A^{\star}A B) \\
    &= \rho \big(B^{\star} (\|A \|^{2}I - A^{\star}A) B \big) \geq 0.
\end{align*}
Thus \( \pi(A) \) is bounded, with \( \| \pi(A) \| \leq \| A \| \); and \( \pi(A) \) extends by continuity to a bounded linear operator \( \pi_{\rho}(A) \) acting on \( H_{\rho} \). It remains to show that \( \pi_{\rho} \) is a \( \star \)-homomorphism and that it is cyclic.

With \( A,B,C \in \mfA \) and \( \alpha, \beta \in \bbC \) we get:
\newline
linearity,
\begin{align*}
    \pi_{\rho} (\alpha A + \beta B) (C + \mfL_{\rho})
    = (\alpha A + \beta B)C + \mfL_{\rho}
    &= \alpha (AC + \mfL_{\rho}) + \beta (BC + \mfL_{\rho}) \\
    &= \big( \alpha \pi_{\rho}(A) + \beta \pi_{\rho}(B) \big)(C + \mfL_{\rho});
\end{align*}
multiplicativity,
\begin{align*}
    \pi_{\rho}(AB) \big(C + \mfL_{\rho} \big)
    = ABC + \mfL_{\rho}
    = \pi_{\rho}(A) \big(BC + \mfL_{\rho} \big)
    = \pi_{\rho}(A)\pi_{\rho}(B) \big(C + \mfL_{\rho} \big);
\end{align*}
star preservation,
\begin{align*}
    \scalp{B + \mfL_{\rho}, \pi_{\rho}(A)(C + \mfL_{\rho})}
    &= \scalp{B + \mfL_{\rho}, AC + \mfL_{\rho}} \\
    &= \rho \big( B^{\star}(AC) \big)
    = \rho \big( (A^{\star}B)^{\star} C \big)\\
    &= \scalp{ A^{\star}B + \mfL_{\rho}, C + \mfL_{\rho}}
    =   \scalp{\pi_{\rho}(A^{\star})(B+\mfL_{\rho}), C + \mfL_{\rho}}.
\end{align*}
Since \( \mfA \slash \mfL_{\rho} \) is everywhere dense in \( H_{\rho} \) these properties extend to the whole \( H_{\rho} \) and thus \( \pi_{\rho} \) is a representation of \( \mfA \) on \( H_{\rho} \).

To see that \( \pi_{\rho} \) is cyclic let \( x_\rho \) be the vector \( I + \mfL_{\rho} \) in \( \mfA \slash \mfL_{\rho} \), then
\begin{align}
    \pi_{\rho}(A) x_{\rho} = \pi_{\rho}(A) (I + \mfL_{\rho}) = A + \mfL_{\rho}
    \qquad \big( A \in \mfA \big),
\end{align}
and \( \pi_{\rho}(\mfA) x_{\rho} \) is \( \mfA \slash \mfL_{\rho} \) --- an everywhere-dense subset of \( H_{\rho} \). Hence \( x_{\rho} \) is a cyclic vector for \( \pi_{\rho} \). Moreover,
\(
    \scalp{ x_{\rho}, \pi_{\rho}(A) x_{\rho} } = \scalp{ I + \mfL_{\rho}, A + \mfL_{\rho} } = \rho(A)
     \) for every \( A \in \mfA \).
In particular, \( \| x_{\rho}\|^{2} = \rho(I) = 1 \). This proves the first part of the theorem.

For the second part, let \( \varphi \) be another representation as stated in the theorem. For each \( A \in \mfA\),
\begin{align*}
    \| \varphi(A) x \|^{2}
    = \scalp{\varphi(A)x, \varphi(A)x}
    = \scalp{x, \varphi(A^{\star}A)x}
    = \rho(A^{\star}A)
    = \scalp{x_{\rho}, \pi_{\rho}(A^{\star}A) x_{\rho}}
    = \| \pi_{\rho}(A) x_{\rho} \|^{2}.
\end{align*}
Hence, if \( A,B \in \mfA \) are such that \( \pi_{\rho}(A) x_{\rho} = \pi_{\rho}(B)x_{\rho} \) then inserting \( A - B \) instead of \( A \) in the above relations, yields \( \varphi(A)x = \varphi(B)x \). Therefore, we can define a map \( U_{0} \colon \pi_{\rho}(\mfA)x_{\rho} \to \varphi(\mfA)x \) by the equation \( U_{0}\pi_{\rho}(A) x_{\rho} = \varphi(A)x \) for each \(A \in \mfA\). It follows that \( U_{0} \) is a norm-preserving linear operator from \( \pi_{\rho}(\mfA)x_{\rho} \) onto \( \varphi(\mfA)x \).
Since the closure of \( \pi_{\rho}(\mfA)x_{\rho} \) is the entire \( H_{\rho} \) and the closure of \( \varphi(\mfA)x \) is the entire \( H \), it follows that \( U_{0} \) extends by continuity to an isomorphism \( U \) from \( H_{\rho} \) onto \( H \), and
\begin{align}
    Ux_{\rho} = U_{0} \pi_{\rho}(I) x_{\rho} = \varphi(I) x = x.
\end{align}
With \( A, B \in \mfA \),
\begin{align}
    U \pi_{\rho}(A) \pi_{\rho}(B)x_{\rho}
    = U \pi_{\rho}(AB) x_{\rho}
    = \varphi(AB)x
    = \varphi(A) \varphi(B) x
    = \varphi(A) U \pi_{\rho}(B)x_{\rho}.
\end{align}
Since vectors of the form \( \pi_{\rho}(B)x_{\rho} \) with \( B \in \mfA \) form an everywhere-dense subset of \( H_{\rho} \), it follows that \( U \pi_{\rho}(A) = \varphi(A) U \), and thus \( \varphi(A) =  U \pi_{\rho}(A)U^{\star} \).
\end{proof}

The method used in the proof of the above Theorem is called the GNS (Gelfand-Naimark-Segal) construction. To avoid bloated notation, we will often write  \( (\pi, H, x) \) to refer to a cyclic representation \( \pi \) of a \Cs-algebra \( \mfA \) on the Hilbert space \( H \) with a cyclic vector \( x \) for \( \pi \). When its convenient we will also call the triple \( (\pi, H, x) \) itself a cyclic representation of \( \mfA \).
If \( \rho \) is a state on \( \mfA \), and \( (\pi_{\rho}, H_{\rho}, x_{\rho}) \) is the representation set out in the Theorem \ref{thm:GNS rep}, we will call it the GNS representation.

For a given \Cs-algebra \( \mfA \), let \( \pi \) be a representation of  \( \mfA \) on the Hilbert space \( H \) and \( \varphi \) be another representation of \( \mfA \) on the Hilbert space \( K \). We will call the two representations \textit{(unitarily) equivalent} if there exists an isomorphism \( U\colon H \to K \) such that \( \varphi(A) = U \pi(A) U^{\star} \) for each \( A \in \mfA \). In particular, the GNS Theorem states that a GNS representation \( \pi_{\omega} \) of \( \mfA \) and the cyclic state \( x_{\omega} \) is uniquely determined (up to unitary equivalence) by the condition \( \rho(\cdot) = \scalp{ x_{\omega}, \pi_{\omega}(\cdot ) x_{\omega} } \).

We already mentioned that if \( \pi \) is a representation of a \Cs-algebra on a Hilbert space \( H \) then \( H \) may have invariant subspaces \( K \subset H \) such that \( \pi(\mfA)K \subseteq K \). If this happens, we call \( \pi \) a \textit{reducible representation}. If on the other hand, \( H \) has no such invariant subspace, and the only invariant subspace is \( H \) itself, then we call \( \pi \) an \textit{irreducible representation}. By restricting a reducible representation to the invariant subspaces we can get an irreducible representation from a reducible one. Conversely, by ``piling up'' (direct sum) irreducible representations we can construct new representations, which will be reducible.
In this sense, irreducible representations are the building blocks of representations. Since \Cs-algebra representations play a vital role in physics, it is important to characterize them. From the above it follows that it is enough to characterize only the irreducible representations. This is a more tractable undertaking.

Using the GNS construction one can draw a connection between irreducible representations of a \Cs-algebra and its pure states. However, we are not going to do this. Instead we will now transition to the mathematical  description of physical systems to find out how \Cs-algebras enter physics.
We close this section with an example that will be important for the next section.

\begin{expl}
Consider a \Cs-algebra \( \mfA \) and a state \( \rho \) on \( \mfA \).
\begin{enumerate}
\item
The GNS representation for the state \( \rho \) is the \( \star \)-homomorphism \( \pi_{\rho} \colon \mfA \to B(H_{\rho}) \), with the Hilbert space \( H_{\rho} \) and a cyclic (GNS) vector \( x_{\rho} \);

\item
Let \( y \) be a unit vector in \( H_{\rho} \). Then we can use \( y \) instead of the cyclic vector for \( \pi_{\rho} \). If \( y \) is cyclic for \( \pi_{\rho} \) then \( (\pi_{\rho},H_{\rho},y) \) is a new cyclic representation which is equivalent to the GNS representation \( \pi_{\rho} \). In fact, the \( \star \)-homomorphism and the Hilbert spaces are both the same, and thus \( U = I \) is the desired isomorphism. (Notice that we do not need that the cyclic vectors transform from one to the other).

\item
For \( y \) from the previous example, set \( \omega_{y}(A) = \scalp{y, \pi_{\omega}(A) y} \) for every \( A \in \mfA \). Then \( \omega_{y} \) is a positive linear functional on \( \mfA \) with \( \omega_{y}(I) =1 \); thus \( \omega_{y} \) is a state on \( \mfA \). The GNS representation \( (\pi_{\omega_{y}}, H_{\omega_{y}}, x_{\omega_{y}}) \) is thus unitarily equivalent to the representation from the previous example.
\end{enumerate}
From the above example we see that every cyclic state of the GNS Hilbert space \( H_{\omega} \) defines a new state \( \omega_{y} \) which leads to a new GNS representation that is equivalent to the original one. In fact, if the original state \( \omega \) is pure, than its GNS representation is irreducible and every vector \( y\in H_{\omega} \) is cyclic. Thus, every vector from the GNS representation of a pure state, defines an equivalent representation.
\end{expl}
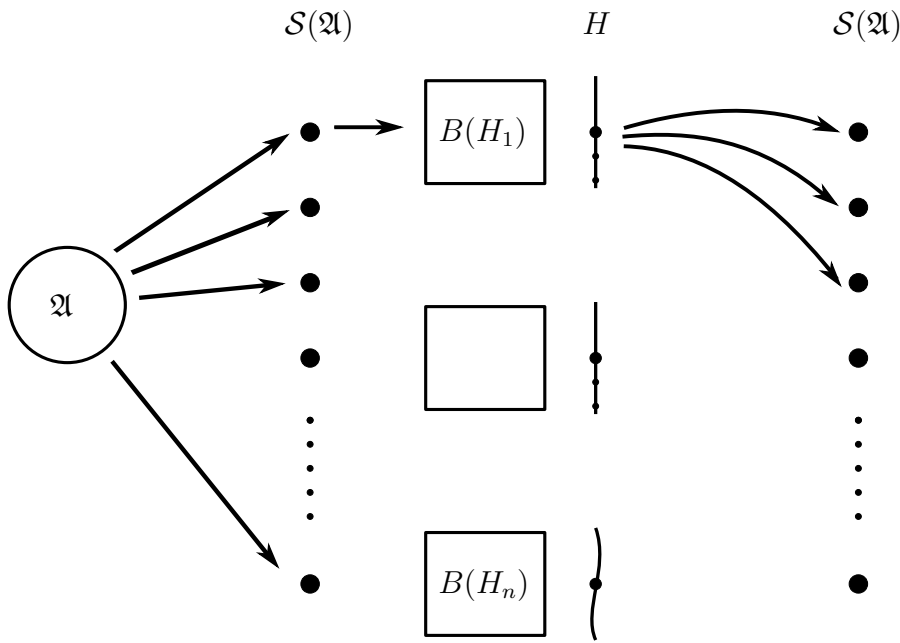
\begin{figure}[h]
\centering
\def\svgwidth{14cm}
\begingroup%
  \makeatletter%
  \providecommand\color[2][]{%
    \errmessage{(Inkscape) Color is used for the text in Inkscape, but the package 'color.sty' is not loaded}%
    \renewcommand\color[2][]{}%
  }%
  \providecommand\transparent[1]{%
    \errmessage{(Inkscape) Transparency is used (non-zero) for the text in Inkscape, but the package 'transparent.sty' is not loaded}%
    \renewcommand\transparent[1]{}%
  }%
  \providecommand\rotatebox[2]{#2}%
  \newcommand*\fsize{\dimexpr\f@size pt\relax}%
  \newcommand*\lineheight[1]{\fontsize{\fsize}{#1\fsize}\selectfont}%
  \ifx\svgwidth\undefined%
    \setlength{\unitlength}{1018.85702175bp}%
    \ifx\svgscale\undefined%
      \relax%
    \else%
      \setlength{\unitlength}{\unitlength * \real{\svgscale}}%
    \fi%
  \else%
    \setlength{\unitlength}{\svgwidth}%
  \fi%
  \global\let\svgwidth\undefined%
  \global\let\svgscale\undefined%
  \makeatother%
  \begin{picture}(1,0.59741396)%
    \lineheight{1}%
    \setlength\tabcolsep{0pt}%
    \put(0,0){\includegraphics[width=\unitlength,page=1]{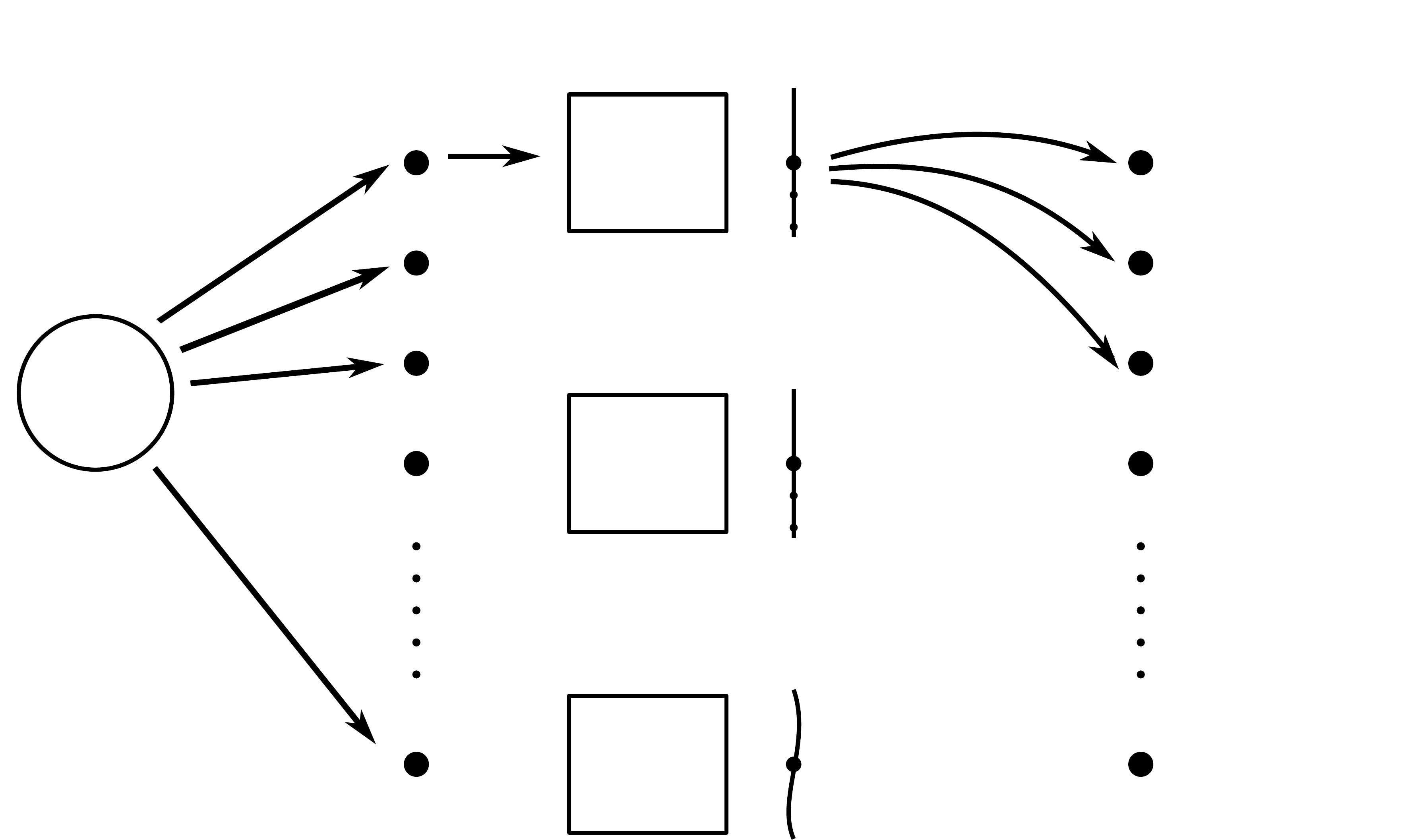}}%
    \put(0.05200805,0.3071041){\color[rgb]{0,0,0}\makebox(0,0)[lt]{\lineheight{1.25}\smash{\begin{tabular}[t]{l}$\mfA$\end{tabular}}}}%
    \put(0.55211255,0.57461159){\color[rgb]{0,0,0}\makebox(0,0)[lt]{\lineheight{1.25}\smash{\begin{tabular}[t]{l}$H$\end{tabular}}}}%
    \put(0.41676752,0.47324087){\color[rgb]{0,0,0}\makebox(0,0)[lt]{\lineheight{1.25}\smash{\begin{tabular}[t]{l}$B(H_{1})$\end{tabular}}}}%
    \put(0.41676752,0.04532693){\color[rgb]{0,0,0}\makebox(0,0)[lt]{\lineheight{1.25}\smash{\begin{tabular}[t]{l}$B(H_{n})$\end{tabular}}}}%
    \put(0.27224192,0.57475537){\color[rgb]{0,0,0}\makebox(0,0)[lt]{\lineheight{1.25}\smash{\begin{tabular}[t]{l}$\mcS(\mfA)$\end{tabular}}}}%
    \put(0.78773385,0.57475537){\color[rgb]{0,0,0}\makebox(0,0)[lt]{\lineheight{1.25}\smash{\begin{tabular}[t]{l}$\mcS(\mfA)$\end{tabular}}}}%
  \end{picture}%
\endgroup%

\caption{The state space \( \mcS(\mfA) \) yields different GNS representations. Some are equivalent some are not. In particular, each (cyclic) vector from the GNS Hilbert space defines a state. The GNS representations of these states are all equivalent to the original representation. Also other states (not created from vectors from the Hilbert space) can yield equivalent representations.}
\end{figure}

\newpage
\section{Operator algebras in physics}
\begin{sectionmeta}
This section is based on two books. The part on classical mechanics as well as the formulation of the Dirac-von Neumann axioms is adopted from \cite{strocchi2008introduction}. The axioms for a general physical system are from \cite{emch2009algebraic}. Nevertheless, axiom 4 is shortened, because otherwise we would need too much additional notation. Also the final axiom from \cite{emch2009algebraic} is missing entirely, because by now we have a classification of Jordan-Banach algebras that make the axiom (more or less) obsolete. The operational motivation (viewing observables as measurement devices) is again from \cite{strocchi2008introduction}. Further reading: the book by Strocchi \cite{strocchi2008introduction} is very commendable for a simple introduction to the physical side of the discussion. Despite being a book for mathematicians, it focuses on the general picture to phrase quantum mechanics in the algebraic language. It is easy to read and is not too heavy on the mathematical details. As a complementary read the book by Emch \cite{emch2009algebraic} is mathematically much more precise and helps to understand how careful and rigorous the algebraic formulation is constructed. For the algebraic formulation of statistical physics Ruelle \cite{Ruelle:1974ur} is a great reference.
Alternatively, the two volumes by Bratteli and Robinson \cite{bratteli2012operator, bratteli2012operatorII} give a self-consistent introduction to the subject. The first volume covers the \Cs-algebra theory, the second volume discusses the algebraic formulation of quantum mechanics and non-relativistic quantum field theory. Contrary to Strocchi and Emch, however, these books focus almost exclusively on the mathematics. The deep physical ideas that historically lied underneath the subject are not a matter of discussion in these two books.
\end{sectionmeta}
\subsection{Classical mechanics}

We separate the description of a physical system into two parts: \textit{kinematics} and \textit{dynamics}.
Kinematics describes the possible configurations of the system; dynamics describes their evolution in time.

\subsubsection{Kinematics}
Kinematics is concerned with describing configurations of a physical system --- objects that uniquely determine all experimental measurements. Such configurations are called \textit{states}. To describe the states we usually choose (independent) variables that determine the state uniquely.  Each such independent physical variable is sometimes termed the \textit{(physical) degree of freedom}.

One way to define a system and to identify its degrees of freedom is to use the Lagrangian formalism. Another, is to use the Hamiltonian description. In the latter, the state is uniquely described in terms of \textit{canonical variables}, which one can think of as (generalized) positions and (generalized) momenta of a particle. The state of a mechanical system is described by a set of canonical variables, \( \{ q, p\} \) where \( q = (q_{1},\dots q_{n}) \in \bbR^{n} \) and \( p = (p_{1}, \dots, p_{n}) \in \bbR^{n}\). The collection of all possible canonical variables for the system is called the phase space \( \Gamma \). For simplicity, we will assume in the following that the phase space is compact, that is the range of possible positions and momenta of a particle are bounded. Since the kinetic energy of a particle with the mass \( m \) is given by
\begin{align}
    E_{kin}(q,p) = \frac{p^{2}}{2m},
\end{align}
this assumption means that the particle under consideration is confined to a bounded region and has finite energy. Thus this assumption is reasonable.

\paragraph{Observables} of the system are the physical quantities that result from the physical degrees of freedom. Thus in particular the position \( q \) and the momentum \( p \) are themselves observables. But all its powers and sums are observables too, and so the observables include the polynomials of the canonical variables \( q \) and \( p \). Without loss of generality, we can take the closure of this space in the uniform topology (supremum norm), to obtain that all observables form the space of real-valued continuous function \( \cc{\Gamma}_{\bbR} \).

Equipped with a pointwise multiplication \( \cc{\Gamma}_{\bbR} \) forms a real algebra. At the same time, \( \cc{\Gamma}_{\bbR} \) is the subspace of self-adjoint elements of the \Cs-algebra \( \cc{\Gamma} \) equipped with the conjugation as the star operation and topologized with the supremum norm.
Extending \( \cc{\Gamma}_{\bbR} \) to \( \cc{\Gamma} \) is quite natural from the algebraic point of view. It is also not harmful from the physics point of view since a \Cs-algebra is uniquely determined by its  self-adjoint elements as we have seen.
Thus, we see that the space of observables of a classical system defines a \Cs-algebra with identity \( I \) --- the constant unit function on \( \Gamma \).

\paragraph{Physical states} define the physical system at hand.
From our description, every state is a point of the phase space \( \{ q, p\} = P \in \Gamma \), and it uniquely determines the value of all the observables. Conversely, the value of all the observables uniquely determines the state by the Urysohn's and the Stone-Weierstrass theorem. This is sometimes called the duality relation between the states and observables.

Identifying states of a classical system with points of a phase space is, however, an idealization. It requires infinitely precise measurements to resolve single points of \( \Gamma \). For realistic measurements, we always have measurement errors and noise that corrupt the measurement. The general procedure (and the one that is in fact used in physics) takes into account these errors. To overcome the noise problem we typically repeat the experiment several times and take the average of the measurement results. More precisely, let us refer to the state of a classical system by \( \omega \). Then the value of an observable \( f \in \cc{\Gamma} \) in the state \( \omega \) is given by the following procedure: we take \( n \) replicas of \( \omega \) and perform measurements of \( f \) on each of the replicas. With this we obtain the measurement results \( m^{(\omega)}_{1}(f), m^{(\omega)}_{2}(f),\dots, m^{(\omega)}_{n}(f) \); subsequently we compute the average
\begin{align}
    \scalp{f}_{n}^{(\omega)} = \frac{1}{n} \big(  m^{(\omega)}_{1}(f) + m^{(\omega)}_{2}(f) +\cdots + m^{(\omega)}_{n}(f) \big)
\end{align}
It is one of the basic assumptions of experimental physics that for \( n \to \infty \) the above quantity converges. This limit value is the expectation value of the observable \( f \) in the state \( \omega \). It is this value that we assign to a measurement
\begin{align}
    \omega(f) = \lim_{n \to \infty} \scalp{f}_{n}^{(\omega)}.
\end{align}

Since the expectation values of all the observables is all there is that we can measure, we can identify a state of a physical system with the set of expectation values of its observables. Since \( \omega(f) \) is operationally defined by the averaging process it follows that such expectations are linear. That is, for \( f,g \in \cc{\Gamma} \) and  \( \alpha, \beta \in \bbC \) we have
\begin{align}
    \omega(\alpha f + \beta g) = \alpha \omega(f) + \beta \omega(g).
\end{align}
Also such expectations are positive
\begin{align}
    \omega(f^{\star}f) \geq 0
    \qquad \big( f \in \cc{\Gamma} \big),
\end{align}
since each measurement \( m^{(\omega)}_{k}(f^{\star}f) \) yields a positive number for every \( k \in \bbN  \).
Thus \( \omega \) defines a positive linear functional on the \Cs-algebra \( \cc{\Gamma} \).

As we have seen such functionals induce an inner product on the \Cs-algebra by
\begin{align}
    \scalp{f,g} = \omega(f^{\star}g)
    \qquad \big( f,g \in \cc{\Gamma} \big),
\end{align}
for which the Cauchy-Schwarz inequality holds
\begin{align}
    |\scalp{f,g}|^{2} \leq \scalp{f,f} \scalp{g,g}
    \qquad \big( f,g \in \cc{\Gamma} \big).
\end{align}
In particular, this implies that for every \( f \in \cc{\Gamma} \) we have
\begin{align}
    |\omega(f)| \leq \omega(I) \omega(f^{\star}f).
\end{align}
If \( \omega(I) = 0 \) then \( \omega(f) = 0 \) for every \( f \in \cc{\Gamma} \), meaning that \( \omega \) is a trivial functional. Conversely, if \( \omega \) is non-trivial then \( \omega(f) \neq 0 \) for at leas one \( f \in \cc{\Gamma} \) and thus \( \omega(I) \neq 0 \). In this case, we can normalize \( \omega \) to satisfy \( \omega(I) = 1 \).
Thus we can identify physical states of a classical system with (algebraic) states on the algebra of observables \( \cc{\Gamma} \).

\begin{info}[Reminder: Riesz-Markov representation theorem.]
Let \( X \) be a locally compact Hausdorff space. For any positive linear functional \( \omega \) on \( \cc{X} \), there is a unique Radon measure \( \mu \) on \( X \) such that
\begin{align}
    \omega(f) = \int_{X} f(x) \ d \mu(x)
    \qquad \big( f \in \cc{X} \big).
\end{align}
\begin{proof}
For the proof consult e.g.\ \cite[thm. IV. 14]{reed2012methods}.
\end{proof}
\end{info}

From the Riesz-Markov theorem it follows that a state \( \omega \) on \( \cc{\Gamma} \) defines a Radon measure (regular Borel measure) \( \mu_{\omega} \) on \( \Gamma \) such that
\begin{align}
    \omega(f) = \int_{\Gamma} f(q,p) \ d \mu_{\omega}(q,p)
    \qquad \big( f \in \cc{\Gamma} \big).
\end{align}
Since \(1 =  \omega(I) = \int_{\Gamma} d\mu_{\omega} = \mu(\Gamma) \) we see that \( \mu \) defines a probability measure on \( \Gamma \). Thus such an operational description of a physical system leads to a probabilistic description of a state, and observables then get the meaning of random variables.

As a special case, when \( \mu_{\delta_{P}} \) is a Dirac measure such that, for \( O \) an open subset of \( \Gamma \)
\begin{align}
    \mu_{\delta_{P}}(O) =
    \begin{cases}
     1 \text{ when } P \in O; \\
     0 \text{  when } P \notin O
\end{cases}
\end{align}
If \( P = \{ q_{0}, p_{0} \} \) we have
\begin{align}
    \delta_{P}(f) = \int_{\Gamma} f(q,p) \ d \mu_{\delta_{P}}(q,p) = f(P) = f(q_{0}, p_{0}).
\end{align}
Thus \( \delta_{P} \) is an evaluation state, and from the Gelfand representation Theorem \ref{thm:Gelfand rep} and the remark that follows it, we saw that such evaluation states define pure state of the \Cs-algebra \( \cc{\Gamma} \). Thus the standard description of classical mechanics, where points of the phase space define the states, corresponds to pure states in the algebraic formalism.

To quantify the precision of measurements for an observable \( f \in \cc{\Gamma} \) in a state \( \omega \), we use the variance
\begin{align}
    (\Delta_{\omega}f)^{2} = \omega \big((f - \omega(f)^{2}) \big).
\end{align}
The smaller the variance, the sharper is the measurement --- the more ``reliable'' is the mean value \( \omega(f) \). Pure states such as \( \delta_{P} \) (\( P \in \Gamma \)) minimize the variance, as for them we have \( \delta_{P}(f^{2}) = \delta_{P}(f)^{2} = f(P)^{2}  \) and thus \( \Delta_{\delta_{P}}f = 0 \). Thus such states are called \textit{dispersion-free}, as they represent measurements of a classical system without any uncertainty.

Since algebraic states capture both, the pure and the mixed states, they provide a unified framework to define sharp measurements as well as statistical fluctuations. For that reason we can use the same formalism to describe statistical mechanics.
In fact, in statistical mechanics the description of the system with pure states is unfeasible because we would need to know the exact positions and momenta for \( 10^{23} \) particles. Therefore, the only reasonable description of such systems is using probability distributions for states and random variables to describe observables. This is what we usually encounter in statistical approach to thermodynamics: the partition function there defines a probability distribution, and observables are random variables sampled from that distribution.

\subsubsection{Dynamics}
Dynamics of a physical system describes how the system evolves in time. When we use canonical variables \( \{ q, p \} \) to describe the system, then dynamics is defined by the time evolution of these variables
\begin{align}
    q \to q(t) &= q(t, q_{0}, p_{0})
    &
    p \to p(t) &= p(t, q_{0}, p_{0}).
\end{align}
In general, the evolution of canonical variables is given by the solution of a system of differential equation called \textit{the Hamiltonian equations of motion},
\begin{align}
    \frac{\partial q}{\partial t} &= \frac{\partial H(q,p)}{\partial p}
    &
    \frac{\partial p }{\partial t} &= - \frac{\partial H(q,p)}{\partial p},\\
    q(t_{0}) &= q_{0}
    &
    p(t_{0}) &=p_{0}.
\end{align}
where \( H(q,p) \) is called \textit{the Hamiltonian} of the system. If \( grad( H )\) is continuous, then the solution to the Hamiltonian equations of motions exists; and if \( grad (H) \) is continuously differentiable then the solution is unique. In this case, it can be extended to all times \( t \in \bbR \) and the evolution \( \{ q, p\} \to \{ q(t), p(t) \} \) is continuous in time \( t \in \bbR \), and depends continuously on the initial data \( \{ q(t_{0}), p(t_{0}) \} = \{ q_{0}, p_{0} \} \).

The evolution of the canonical variables implies that the results of measurements evolve in time. Since all observables are constructed from canonical variables, all observables evolve in time. In particular, this implies a one-parameter family of continuous invertible mappings \( \alpha_{t_{0}, t} \colon \cc{\Gamma} \to \cc{\Gamma} \) from the algebra of observables into itself so that
\begin{align}
    \alpha_{t_{0}, t} \colon f(q_{0},p_{0}) \to  f(q(t), p(t)).
\end{align}
If \( f,g \) are two polynomials of canonical variables
then it is easy to see that for \( \beta, \gamma \in \bbC \) and every \( t \in \bbR \) we have,
\begin{align}
    \alpha_{t_{0},t}(\beta f + \gamma g) &=
    \beta \alpha_{t_{0},t}( f) + \gamma \alpha_{t_{0},t}(g)
    = \beta f(t,q_{0},p_{0}) + \gamma g(t, q_{0}, p_{0}),\\
    \alpha_{t_{0},t} (f^{\star}) &=
    \alpha_{t_{0},t}(f)^{\star}
    = f(t, q_{0}, p_{0})^{\star}, \\
    \alpha_{t_{0},t}(fg) &=
    \alpha_{t_{0},t}(f)\alpha_{t_{0},t}(g)
    = f(t,q_{0}, p_{0} ) g(t, q_{0}, p_{0}).
\end{align}
By continuity, these properties extend to the whole algebra \( \cc{\Gamma} \), and thus \( \alpha_{t_{0},t} \) defines a one-parameter family of \( \star \)-automorphisms (\( \star \)-isomorphisms from the algebra into itself).
When the dynamics depends only on the difference \( t-t_{0} \), i.e.\ it is \textit{time-translation invariant}, then the one-parameter family is actually a one-parameter group of automorphisms \( \{ \alpha_{t} \colon t \in \bbR  \} \) on \( \cc{\Gamma} \). By duality it defines a one-parameter group of transformations \( \alpha_{t}^{\star} \) on the state space as
\begin{align}
    \omega_{t}(f) = (\alpha_{t}^{\star}(\omega))(f) = \omega( \alpha_{t}(f))
    \qquad \big( f \in \cc{\Gamma} \big).
\end{align}

\vspace{0,5cm}
In summary we see, that observables of a classical system define the commutative \Cs-algebra of observables \( \cc{\Gamma} \), and physical states define (algebraic) state on that algebra; dynamics of the canonical system is given by a one-parameter group of \( \star \)-automorphisms on \( \cc{\Gamma} \). In particular, all physically observable features of a physical system are captured by the algebraic formulation.
However, still in the standard approach we start with the phase space that describes our system, and then derive the \Cs-formulation.

The Gelfand representation Theorem \ref{thm:Gelfand rep}, however, allows us to reverse the logic. Instead of starting with the phase space, and deriving the algebra from it, we can start with a commutative \Cs-algebra \( \mfA \) that captures an abstract characterization of observables of our system; define the physical states as the state space \( \mcS(\mfA) \); and define the dynamics as a \( \star \)-automorphism on \( \mfA \). By the Gelfand representation Theorem the abstract algebra \( \mfA \) can be represented as an algebra of continuous functions \( \cc{X} \) on a compact Hausdorff space \( X \). Every dispersion-free state of \( \mfA \) is associated with the point in \( X \) and so the space \( X \) gets the meaning of the phase space. In this sense, the algebra comes first, from which the phase space is then derived.
This is the starting point for the algebraic construction. Using this motivation we aim in the next section to define a general algebraic framework for physical theories.

\subsection{An algebraic construction for physical theories}

The discovery of quantum mechanics soon proved to be pivotal for physics. This inspired a rapid development in the field. This development was of physical nature, as more and more quantum mechanical models described more and more experiments with extraordinary precision. However, on the mathematical side the theory was shaky, and the concepts were hardly justified on physical grounds. To formalize the theory, Dirac and von Neumann provided 5 axioms to summarize a common mathematical structure of quantum mechanics as developed by Heisenberg and Schr\"odinger --- sometimes referred to as the Dirac-von Neumann axioms of quantum mechanics. Before we delve into the algebraic framework we summarize a simplified version of these axioms that define the standard approach to quantum mechanics.
\paragraph{States:}
The states of a physical system are described by rays (essentially unit vectors) in a separable Hilbert space \( H \).
\paragraph{Observables:}
The observables of a physical system are described by the set of bounded self-adjoint operators in \( H \).
\paragraph{Expectations:}
If a state \( \omega \) is described by the vector \( \psi \in H \), for any observable \( A \) the expectation \( \omega(A) \), defined by the average of the outcomes of replicated measurements of \( A \) performed on the system in the state \( \omega \), is given by the Hilbert space matrix element
\begin{align}
    \omega(A) = \scalp{\psi, A \psi}
\end{align}
\paragraph{Dirac canonical quantization:}
Canonical coordinates and momenta satisfy the following commutation relations
\begin{align}
    [q_{i}, q_{j}] = 0 = [p_{i}, p_{j}],
    &&
    [q_{i}, p_{j}] = \imath \hbar \delta_{i,j},
    \qquad \big( i,j = 1,2,\dots, n \big).
\end{align}
\paragraph{Schr\"odinger representation:}
The canonical commutation relations are represented by the following operators in \( H = L^{2}(\bbR^{n}, \lambda) \):
\begin{align}
    q_{i} \psi(x) = x_{i} \psi(x),
    &&
    p_{j} \psi(x) = - \imath \hbar \frac{\partial \psi}{\partial x_{j}}(x)
    \qquad \big( i,j = 1,2,\dots, n \big).
\end{align}

Even though the physical framework proved to be very successful, the physical basis for many of its assumptions was unclear and unjustifiable.
For that reason, a group of mathematicians/physicists pioneered by Jordan, von Neumann, and Wigner, aimed to provide a conceptually clear, and mathematically rigorous formalism for the above postulates; to justify some of the tools in use of quantum mechanics; and to disregard some of the assumptions not supported by mathematics or physical intuition. In this section we discuss the summary of that approach that took a long time to be developed.

The idea of the algebraic approach is to use reasonable physical assumptions to define an algebraic structure on the set of observables of a theory. Then to show that this algebraic structure is rich enough to provide all the necessary tools used in quantum mechanics as formulated in the above axioms.
Following the literature, we formulate the necessary assumptions on the set of observables as axioms.

\begin{axiom}
To each physical system \( \Sigma \) we can associate the triple \( (\pA, \pS, \scalp{;}) \) formed by the set \( \pA \) of all observables, the set \( \pS \) of all its states, and a mapping \( \scalp{;} \colon (\pA,\pS) \to \bbR\) that assigns  to each pair \( (A,\phi) \in (\pA,\pS) \) a real number \( \scalp{\phi;A} \) that we interpret as the expectation value of the observable \( A \) in the state \( \phi \).
\end{axiom}

With a state of a physical system, we identify the method used for its preparation. That is a (reproducible) way to prepare a configuration of the system.
We assume that a state is fully determined by all its expectation values, that is if \( \phi, \varphi \in \pS \) are such that
\begin{align}
    \scalp{\phi; A} = \scalp{\varphi; A}
    \qquad \big( A \in \pA \big),
\end{align}
then this must imply that \( \phi = \varphi \). That is, we assume that there are no hidden properties that we can never measure. We say \textit{the observables separate the states}.

With an observable we identify a measurement device. Thus for any \( A \in \pA \) we can imagine a device with a pointer scale. The expectation value, is then the average over \( n \) measurements of \( A \) in identically prepared states \( \omega \). Ideally, \( n \) shall go to infinity, practically it has to be large enough to assure that our average is meaningful.
If any two devices produce the same expectation values for all states \( \phi \in \pS \), then we identify these devices. That is we assume that each measurement device produces a unique set of expectation values, when evaluated in all states of \( \pS \). We say, that \( \pS \) \textit{separates the observables} or simply that the set \( \pS \) is \textit{full with respect to} \( \pA \) (see Figure \ref{fig:observabels and states} for intuition).

\begin{figure}[h]
    \centering
    \subfloat[Observables: measurement devices.]{%
        \def\svgwidth{7cm}
\begingroup%
  \makeatletter%
  \providecommand\color[2][]{%
    \errmessage{(Inkscape) Color is used for the text in Inkscape, but the package 'color.sty' is not loaded}%
    \renewcommand\color[2][]{}%
  }%
  \providecommand\transparent[1]{%
    \errmessage{(Inkscape) Transparency is used (non-zero) for the text in Inkscape, but the package 'transparent.sty' is not loaded}%
    \renewcommand\transparent[1]{}%
  }%
  \providecommand\rotatebox[2]{#2}%
  \newcommand*\fsize{\dimexpr\f@size pt\relax}%
  \newcommand*\lineheight[1]{\fontsize{\fsize}{#1\fsize}\selectfont}%
  \ifx\svgwidth\undefined%
    \setlength{\unitlength}{424.95218185bp}%
    \ifx\svgscale\undefined%
      \relax%
    \else%
      \setlength{\unitlength}{\unitlength * \real{\svgscale}}%
    \fi%
  \else%
    \setlength{\unitlength}{\svgwidth}%
  \fi%
  \global\let\svgwidth\undefined%
  \global\let\svgscale\undefined%
  \makeatother%
  \begin{picture}(1,0.9411285)%
    \lineheight{1}%
    \setlength\tabcolsep{0pt}%
    \put(0,0){\includegraphics[width=\unitlength,page=1]{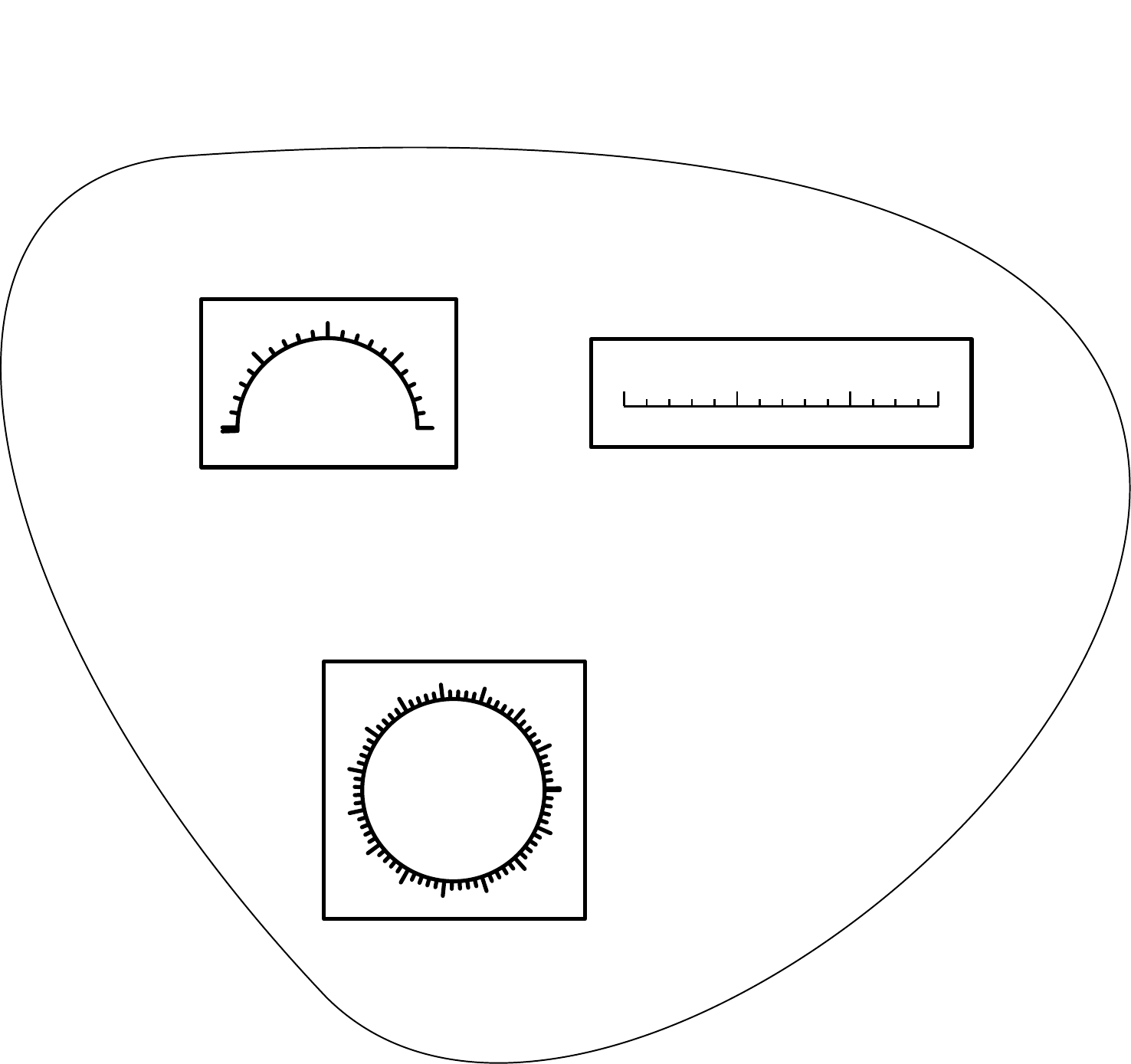}}%
    \put(0.61788895,0.26464022){\color[rgb]{0,0,0}\makebox(0,0)[lt]{\lineheight{1.25}\smash{\begin{tabular}[t]{l}$\ddots$\end{tabular}}}}%
    \put(0.43641782,0.88680271){\color[rgb]{0,0,0}\makebox(0,0)[lt]{\lineheight{1.25}\smash{\begin{tabular}[t]{l}$\pA$\end{tabular}}}}%
    \put(0.2738094,0.70768989){\color[rgb]{0,0,0}\makebox(0,0)[lt]{\lineheight{1.25}\smash{\begin{tabular}[t]{l}$A$\end{tabular}}}}%
    \put(0.67456653,0.6682102){\color[rgb]{0,0,0}\makebox(0,0)[lt]{\lineheight{1.25}\smash{\begin{tabular}[t]{l}$B$\end{tabular}}}}%
    \put(0.39577849,0.3918525){\color[rgb]{0,0,0}\makebox(0,0)[lt]{\lineheight{1.25}\smash{\begin{tabular}[t]{l}$C$\end{tabular}}}}%
  \end{picture}%
\endgroup%

        \label{fig:pN}%
        }
        \hspace{1cm}
        \subfloat[States: methods of state preparation.]{%
            \def\svgwidth{7cm}
\begingroup%
  \makeatletter%
  \providecommand\color[2][]{%
    \errmessage{(Inkscape) Color is used for the text in Inkscape, but the package 'color.sty' is not loaded}%
    \renewcommand\color[2][]{}%
  }%
  \providecommand\transparent[1]{%
    \errmessage{(Inkscape) Transparency is used (non-zero) for the text in Inkscape, but the package 'transparent.sty' is not loaded}%
    \renewcommand\transparent[1]{}%
  }%
  \providecommand\rotatebox[2]{#2}%
  \newcommand*\fsize{\dimexpr\f@size pt\relax}%
  \newcommand*\lineheight[1]{\fontsize{\fsize}{#1\fsize}\selectfont}%
  \ifx\svgwidth\undefined%
    \setlength{\unitlength}{665.24626133bp}%
    \ifx\svgscale\undefined%
      \relax%
    \else%
      \setlength{\unitlength}{\unitlength * \real{\svgscale}}%
    \fi%
  \else%
    \setlength{\unitlength}{\svgwidth}%
  \fi%
  \global\let\svgwidth\undefined%
  \global\let\svgscale\undefined%
  \makeatother%
  \begin{picture}(1,0.92449421)%
    \lineheight{1}%
    \setlength\tabcolsep{0pt}%
    \put(0,0){\includegraphics[width=\unitlength,page=1]{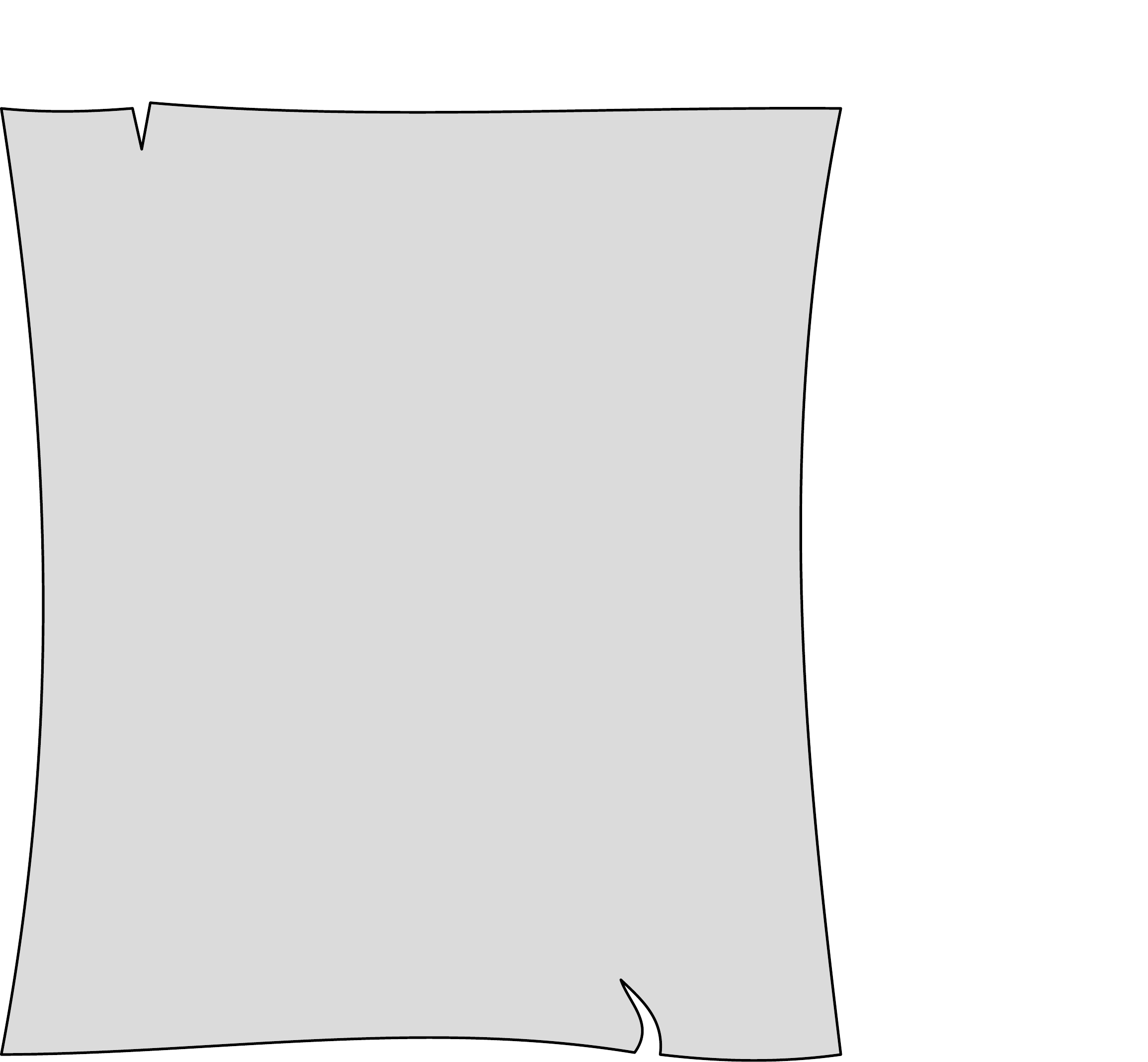}}%
    \put(0.08580071,0.65733115){\color[rgb]{0,0,0}\makebox(0,0)[lt]{\lineheight{1.25}\smash{\begin{tabular}[t]{l}1) switch on the stove\end{tabular}}}}%
    \put(0.088377,0.56089379){\color[rgb]{0,0,0}\makebox(0,0)[lt]{\lineheight{1.25}\smash{\begin{tabular}[t]{l}2) take 100g sugar\end{tabular}}}}%
    \put(0.08824485,0.47443258){\color[rgb]{0,0,0}\makebox(0,0)[lt]{\lineheight{1.25}\smash{\begin{tabular}[t]{l}3) eat it\end{tabular}}}}%
    \put(0.09022661,0.37965784){\color[rgb]{0,0,0}\makebox(0,0)[lt]{\lineheight{1.25}\smash{\begin{tabular}[t]{l}4) hope for the best\end{tabular}}}}%
    \put(0.32476355,0.21532281){\color[rgb]{0,0,0}\makebox(0,0)[lt]{\lineheight{1.25}\smash{\begin{tabular}[t]{l}$\vdots$\end{tabular}}}}%
    \put(0.80490722,0.4340289){\color[rgb]{0,0,0}\makebox(0,0)[lt]{\lineheight{1.25}\smash{\begin{tabular}[t]{l}$\cdots$\end{tabular}}}}%
    \put(0.48528161,0.88979147){\color[rgb]{0,0,0}\makebox(0,0)[lt]{\lineheight{1.25}\smash{\begin{tabular}[t]{l}$\pS$\end{tabular}}}}%
  \end{picture}%
\endgroup%

            \label{fig:pS}%
            }
    \caption{
    The set of observables and the set of states of a physical system. After preparing a state according to a given method we can measure it with all observables and identify it with a list of all expectation values. If any to methods lead to the same list of expectation values we identify these methods. On the other hand, if any two observables provide two identical lists of expectation values when measuring all states, we identify them either.}
    \label{fig:observabels and states}
\end{figure}
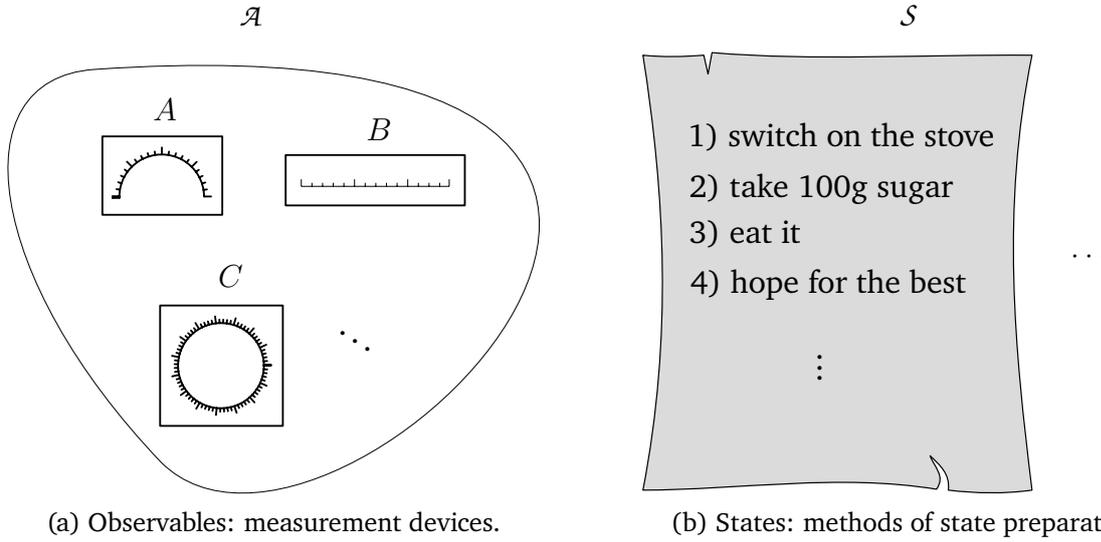

Next we realize that we can use the states to compare observables. For \( A, B \in \pA \) we say that \( A \leq B \) if \( \scalp{\phi; A} \leq \scalp{\phi; B} \) for all state \( \phi \in \pS \). In particular, we call \( A \geq 0  \) positive if \( \scalp{\phi; A} \geq 0 \) for every state \( \phi \in \pS \). Especially, this implies that the states are positive functionals on \( \pA \).

\begin{axiom}
The relation \( \leq \) defined above is a partial ordering on \( \pA \); in particular \( A \leq B \) and \( B \leq A \) imply \( A = B \).
\end{axiom}
With our interpretation of observables as measurement devices we see that the condition \( \scalp{\phi; A} \geq 0 \) for every \( \phi \in \pS \) implies that every measurement \( A \) is positive. Since the set of state is full with respect to the set of observables this meas that the scale of the measurement device represented by \( A \) is positive. Thus we can identify the positive element with measurement devices whose pointer scale ranges in positive real numbers.

Next we notice that for any two observables \( A,B \in \pA \) the set \( \{ \scalp{\phi; A} + \scalp{\phi; B} \colon \phi \in \pS \} \) can be interpreted as a set of expectation values, thus a measurable quantity on \( \Sigma \). We strengthen this observation to the next axiom.

\begin{axiom}\label{ax:3}
\begin{enumerate}
\item[]

\item
For each observable \( A \in \pA \)  and any \( \lambda \in \bbR \) there exists an element \( (\lambda A) \in \pA \) such that
\begin{align}
    \scalp{\phi; \lambda A} = \lambda \scalp{\phi; A}
    \qquad \big(  \phi \in \pS \big).
\end{align}

\item
For any pair of observables \( A,B \in \pA \) there exists an element \( (A + B) \in \pA \) such that
\begin{align}
    \scalp{\phi; A + B} = \scalp{\phi; A} + \scalp{\phi; B}
    \qquad \big( \phi \in \pS \big).
\end{align}
\end{enumerate}
\end{axiom}
If \( A \in \pA \) represents a measurement and \( \lambda \) is a real number then we can think of \( \lambda A \) as a new device, whose pointer scale is rescaled by \( \lambda \).

This axiom, even though reasonable, is non-trivial when we interpret elements in \( \pA \) as measurement devices. Using purely operational arguments, it is not clear that for any two measurement devices represented by \( A \) and \( B \) there is always one device \( C \) that produces the same expectation values as the sum of expectation values of \( A \) and \( B \). Strictly speaking, with this assumption we cannot interpret the set \( \pA \) solely as a set of measurement devices, but we extend this set to all hypothetical measurements that result from adding expectation values of measurements together.
 In particular, this defines a linear structure on \( \pA \) making it into a vector space. Nevertheless, this assumption is still quite reasonable. Being pragmatic, we can agree that the set \( \pA \) contains all theoretically measurable devices. Hence also such devices that correspond to measurements of \( A + B \).

Recall from our analysis of a classical system, that dispersion-free states corresponded to idealized measurements with infinite precision, i.e.\ to points of the phase space. Because of this, dispersion-free states uniquely defined all classical observables. It seems reasonable to expect that this should hold in a general physical system. We use this intuition to state the next axiom.

\begin{axiom}\label{ax:4}
For every observable \( A \) the set \( \pS \) contains  the set of dispersion-free states \( \pS_{A} \) for \( A \) that uniquely determines the one-dimensional subspace of \( \pA \) generated by \( A \); for any two observables \( A, B \) the dispersion-free states satisfy \( \pS_{A} \cap \pS_{B} \subseteq \pS_{A+B} \) and \( \pS_{I} = \pS \).
\end{axiom}
In particular, this axiom implies that for two observables \( A \) and \( B \) if \( \pS_{A} \subseteq \pS_{B} \) and \( \scalp{\phi;A} = \scalp{\phi;B} \) for every \( \phi \in \pS_{A} \) then \( A \) is equal to \( B \).
Thus dispersion-free states in a certain sense, uniquely determine the observables. This reflects our original motivation to introduce this axiom.
Nevertheless, this axiom requires an assumption of dispersion-free states, which is an idealization of a physical process of measurement, as we discussed above. To argue that such states are reasonable, we can take the standpoint that even though every realistic measurement is noisy, the noise can be arbitrary reduced. Therefore there is no fundamental (nature given) threshold which we cannot surpass by making a measurement more precise/less noisy. (In fact, this may be violated if one takes General Relativity into account. But we are not going into this discussion here.)

If \( A \) is an observable that represents a measurement device then we  can interpret \( A^{2} \) as a measurement device whose scale is the square of the pointer scale of \( A \). In particular, \( A^{2} \) is therefore always positive. Similarly, we can define any power \( A^{n} \) for a non-negative integer \( n \), and set \( I = A^{0} \). We put the existence of such powers as an additional axiom.

\begin{axiom}\label{ax:5}
For any observable \( A \) and any non-negative integer \( n \) there exists at least one element \( A^{n} \in \pA \) such that
\begin{enumerate}
\item
\( \pS_{A} \subseteq \pS_{A^{n}}  \),

\item
\( \scalp{\phi; A^{n}} = \scalp{\phi; A}^{n} \) for all \( \phi \in \pS_{A} \).
\end{enumerate}
\end{axiom}

From the second point of this Axiom and Axiom \ref{ax:4} we see, that the element \( A^{n} \) is unique in \( \pA \) and we can unambiguously call it the \( n \)th power of \( A \).

With this axiom the vector space \( \pA \) obtains a power structure. Moreover, it follows that \( \pA \) contains a unit \( I = A^{0} \) for every \( A \in \pA \) and an element \( O \) such that
\begin{align}
    \scalp{\phi; I} = 1
    && \text{and} &&
    \scalp{\phi; O} = 0
\qquad \big( \phi \in \pS \big).
\end{align}
From the second point of this axiom, it follows that \( \pS_{I} = \pS_{A^{0}} =  \pS \), and that the usual rules apply: \( A^{0} = I, \ A^{1} = A, \ I^{n}= I \) and \( O^{n} = O \) for \( n > 1 \).

Furthermore, since \( \scalp{\phi; A^{2}} \geq 0  \) for all \( \phi \in \pS_{A} \) and since \( \pS_{A} \) fully determines \( A^{2} \) it follows that \( A^{2} \geq 0 \). In particular this implies that \( \scalp{\phi; A^{2}} \geq 0 \) for every state \( \phi \in \mcS \). Since the set \( \mcS \) is full with respect to \( \mcA \), we can conclude that \( A^{2} + B^{2} + C^{2} + \cdots =0 \) implies \( A = B =C = \cdots 0 \). This conditions is sometimes called \textit{reality}, in analogy to the fact that this implication holds for real numbers (as opposed to complex numbers).

Even though we have argued for a lot of structure on the space \( \pA \), one important property is still missing: the distributivity of powers over addition of any two elements \( A,B \). In general, for \( A,B \in \pA \) we have that
\begin{align}
    (A+B)^{n} = (A+B)(A+B) \neq A (A+B) + B(A+B).
\end{align}
Stated more precisely, we never defined the product between \( A \) and \( B \) that again defines an observable \( AB \). Notice, that in the usual approach to quantum mechanics, where we identify observables with self-adjoint operators on a Hilbert space this is in particular false: a product of self-adjoint elements is in general not self-adjoint.

On the other hand, for any observable \( A \) the product  \( AA \) is defined  by the unique element \( A^{2} \) from the Axiom \ref{ax:5}. We use this observation to equip \( \pA \) with a product.
\begin{defi}
With each pair of observables \( A,B \) we associate an element \( A \circ B \) in \( \pA \), called the symmetric product of \( A \) and \( B \), defined by
\begin{align}
    A\circ B = \frac{1}{2}\big( (A+B)^{2} - A^{2} - B^{2} \big).
\end{align}
\end{defi}
This product depends on powers of the same observable (which we defined) and on sums between observables (which we also defined).
Moreover, this product is commutative, \( A \circ B = B \circ A \), and from Axiom \ref{ax:4} and \ref{ax:5} it satisfies
\begin{align}
    \scalp{\phi: A \circ B} = \scalp{\phi; A} \scalp{\phi; B}
    \qquad \big( \phi \in \pS_{A} \cap \pS_{B} \big).
\end{align}
Thus, since dispersion-free states uniquely determine the observables it follows that \( A^{n} \circ A^{m} = A^{n+m} \). In particular, we have \( A \circ I = A \). In other words, the power structure that we had derived after the previous Axiom, can be generated by the symmetric product on \( \pA \). Therefore \( \circ \) naturally extends the structure that we already have. This makes \( \pA \) almost an algebra with identity, however, not quite yet.

The symmetric product, even though commutative, is neither associative, \( (A\circ B) \circ C \neq A \circ (B \circ C) \), nor does it distribute over addition. The non-associativity is not a big loss, and in general there is little algebraic reason to want to restore associativity. Distributivity, on the other hand, is crucial for a definition of an algebra.
At the same time, it seems hard to find a physically compelling argument to make the product distributive.
However, one can show (even though we do not do it here) that distributivity follows from the following mild looking assumption.

\begin{axiom}\label{ax:6}
For any pair of observables \( A,B \) the symmetric product is homogeneous in the fist entry, i.e.\ for every real number \( \lambda \) it satisfies \( \lambda (A \circ B) = (\lambda A) \circ B  \).
\end{axiom}
Since the product is commutative, homogeneity in the first entry implies homogeneity in both entries.
With this assumption the symmetric product becomes distributive. In fact, it also becomes weakly associative meaning that for \( A,B \in \pA \) the product satisfies \( A\circ(B \circ A^{2}) = (A \circ B) \circ A^{2} \). A commutative algebra with a symmetric product that satisfies the weak associativity is called a \textit{Jordan algebra}; and a Jordan algebra that satisfies the reality condition is called a \textit{real Jordan algebra}. Therefore we can say that with the first 6 Axioms the set of observables \( \pA \) becomes a real Jordan algebra. Notice, that the Axioms 1-5 are very natural from the physicists perspective. The Axiom \ref{ax:6} is the only questionable one. However, even though there is no direct evidence that the product must be distributive, there is no direct physical reason why it should not be. Mathematically, this assumption makes a significant difference.

At this point, our set \( \pA \) is an algebra, but it remains to be shown that this algebra is indeed useful: that it provides a computational or conceptual advantage over the standard approach to classical or quantum theory. In particular, it is important to classify such algebras, and to understand their representations. To approach this question we need to introduce topology on \( \pA \).
\begin{defi}
We define a mapping \( \| \cdot \| \colon \mcA \to \bbR \) such that for each \( A \in \pA \),
\begin{align}\label{eq:norm joradn}
    \| A \| = \sup \{| \scalp{\phi; A}| \colon \phi \in \pS \}.
\end{align}
\end{defi}
Identifying observables with measurement devices, relation \eqref{eq:norm joradn} assigns the largest possible expectation value to each measurement. This mapping has almost all the properties of a norm: if \( \|A \| = 0 \) then \( \scalp{\phi;A} = 0 \) for every \( \phi \in \pS \), and since \( \pS \) separates the observables we get \( A = 0 \); by Axiom \ref{ax:3} each state is linear and thus, for \( A \in \pA \) and \( \lambda \in \bbR \) we have \(
    \| \lambda A \| = | \lambda | \| A \|\); with \( B \in \pA \) it is clear that \(
    \| A + B \| \leq \| A \| + \| B \|\). Moreover, from the reality condition it also holds that \( \| A^{2} \| \leq \| A^{2} + B^{2} \|  \) for any \( A,B \in \pA \). The only missing requirement such that \eqref{eq:norm joradn} defines a norm, is that it is finite for every element of \( \pA \). We put this as an axiom.

\begin{axiom}\label{ax:7}
Each element \( A \) in \( \pA \) has finite norm and \( \pA \) is complete in the norm topology, and we identify elements in \( \pS \) with positive continuous linear functional with \( \scalp{\phi;I} = 1 \).
\end{axiom}

The reason why we require the norm of observables to be bounded can be again understood with our identification of observables with measurement devices (even though, by now the space \( \pA \) is significantly larger than it was in the beginning of our discussion). Each practical measurement device has a finite (bounded) scale. Therefore, if \( A \) represents a measurement device then even the largest expectation value of \( A \) is going to be bounded by the scale of that device. Measurements that contain infinity as a valid measurement outcome are even theoretically not feasible. Therefore this axiom is reasonable on physical grounds. Assuming that \( \pA \) is complete, on the other hand, is merely a practical requirement. Even though it does enlarge the space \( \pA \), it does not significantly alter the physical interpretation.

With this norm each state in \( \pS \) is by definition continuous. Thus a state \( \phi \in \pS \) is a positive, linear, continuous functional on the algebra \( \pA \) with \( \phi(I) = 1 \). Thus, \( \pS \) is the space of states in the sense of our algebraic Definition \ref{defi:states} \ref{defi: states it2}. For that reason, henceforth we will write \( \mcS(\pA) \) instead of \( \pS \) and use our standard notation \( \omega(A) \) instead of \( \scalp{\omega; A} \).

A non-trivial property of the introduced norm is stated in the following Lemma.
\begin{lem}\label{lem:norm on pA}
For every observable \( A \in \pA \) it holds
\begin{align}
    \| A^{2} \| = \| A \|^{2}.
\end{align}
\end{lem}
\begin{proof}
By definition of the norm it follow that \( \omega(A) \leq \| A \| \) for every state \( \omega \in \mcS(\pA) \). Thus by linearity of \( \omega \) we have
\( \omega( \| A \| I \pm A) \geq 0 \), so both \( \|A \|I \pm A \) are positive. Thus \( (\|A \|I - A)(\| A \| I + A) \) is also positive (e.g. because it is uniquely determined by the dispersion-free states/ two measurement devices both with positive pointer scale).
Then for every state \( \omega \in \mcS(\pA) \),
\begin{align}
    \| A \|^{2} - \omega(A^{2}) = \omega \big((\| A \| I - A)(\| A \| I + A)\big) \geq 0,
\end{align}
which implies \( \|A \|^{2} \geq \| A^{2} \| \).

On the other hand, the positivity of \( (\| A \| I \pm A)^{2}  = \| A \|^{2} + A^{2} \pm 2 \| A \| A\) implies that for any state \( \omega \) we have
\begin{align}
    2 \|A \| \omega(A) \leq \| A \|^{2} + \omega(A^{2}) \leq \| A \|^{2} + \| A^{2} \|,
\end{align}
and therefore \( \|A\|^{2} \leq \|A^{2}\| \).
\end{proof}
From the Proposition \ref{prop:Cauch Schwarz} we know that states induce an inner product on the vector space on which they are defined; in this case on \( \pA \) by,
\begin{align}
    \omega(A\circ B) = \scalp{A, B}
    \qquad \big( A,B \in \pA \big).
\end{align}
Moreover, we have seen that this product satisfies the Cauchy-Schwarz inequality. This inequality, in turn, implies that the product is continuous since for \( A,B \in \pA \) we have
\begin{align}\label{eq:CS for jordan}
    |\omega( A \circ B)|^{2} \leq \omega(A^{2}) \omega(B^{2}),
\end{align}
for every \( \omega \in \mcS(\pA) \) and every pair \( A,B \in \pA \).
Taking the supremum on both sides of \eqref{eq:CS for jordan} and using Lemma \ref{lem:norm on pA} implies
\( \| A\circ B \|^{2} \leq \| A^{2} \| \|B^{2}\| = \| A \|^{2} \| B \|^{2}\) for every \( A,B \in \pA \).

A real Jordan algebra, which is a Banach space in the norm in which the product is continuous is called a \textit{Jordan-Banach algebra}.
However, we have more structure. The norm of our Jordan-Banach algebra also satisfies \( \|A^{2} \| = \|A \|^{2} \) and the reality condition \( \| A^{2} \| \leq \| A^{2} + B^{2} \| \). Such Jordan-Banach algebras are called \textit{JB-algebras}.

In summary, \( \pA \) is a JB-algebra which is commutative, distributive, but not-necessarily associative with a continuous product.
From the mathematical point of view, however, it is easier to treat an algebra which is associative but not-necessarily commutative.
Because of this, it is convenient to ask the question when the JB-algebra can be embedded into a larger algebra such that the symmetric product is given by the (associated but not necessarily commutative) product \( AB \) such that
\begin{align}
    A \circ B = \frac{1}{2} \big( (A+B)^{2} - A^{2} - B^{2} \big) = \frac{AB + BA}{2}.
\end{align}
To achieve this, we can first complexify the algebra \( \pA \) by considering the set
\begin{align}
    \{ A + \imath B \colon A,B \in \pA \},
\end{align}
and then introduce the involution by \( \star \colon A + \imath B \to A - \imath B \). One can show that most of the properties of the \Cs-algebra are satisfied. Especially, one can extend the above norm such that it satisfies the \Cs-property \( \|A^{\star} A \| = \|A \|^{2} \) and that \( \|A^{\star} \| = \|A \| \). However, the product on this algebra may still be non-associative. This is the reason why the embedding may sometimes fail.

 A  JB-algebra is called \textit{special} if its embedding into a \Cs-algebra exists; otherwise it is called \textit{exceptional}. Even though there is always a homomorphism from a JB-algebra into a \Cs-algebra \cite[thm. 7.1.8]{hanche1984jordan}, there may not exist an isometric isomorphism that embeds a JB algebra into a \Cs-algebra. A JB-algebra that is isometrically isomorphic to a Jordan subalgebra of a \Cs-algebra is called a JC-algebra.
The following theorem \cite[thm. 7.2.3 and lem. 7.2.2]{hanche1984jordan} characterizes all exceptional JB algebras.
\begin{thm}
Any exceptional JB-algebra \( \pA \) contains a unique purely exceptional ideal \( \mfI \) such that \( \pA \slash \mfI \) is a JC-algebra. The ideal \( \mfI \) is itself a Jordan algebra and each of its factor representations is onto \( H_{3}(\bbO) \).
\end{thm}
In the above Theorem \( H_{3}(\bbO) \) denotes the algebra of \( 3\times 3 \) self-adjoint matrices over octonions. This algebra seems too poor for physical applications, and so an embedding of a JB-algebra into a \Cs-algebra is quite natural. We state this as a finial axiom.

\begin{axiom}\label{ax:8}
    The algebra of observables \( \pA \) can be identified with the set of all self-adjoint elements of a not necessarily commutative \Cs-algebra \( \mfA \).
\end{axiom}
If such an embedding is possible then the JB-algebra (treated as a subset of self-adjoint elements in \( \mfA \)), generates the \Cs-algebra \( \mfA \). In particular, completeness of \( \pA \) in the JB-algebra norm already implies completeness of the ambient \Cs-algebra, thus no further closure assumption is needed.

We can summarize our above discussion as follows:
\begin{enumerate}
\item
A physical system is defined by its \Cs-algebra of observables with identity.

\item
The physical states are identified with the algebraic states on that \Cs-algebra.

\item
The dynamics of the system is a \( \star \)-automorphism on the \Cs-algebra of observables.
\end{enumerate}

\paragraph{To conclude:}
We started with a set of observables \( \pA \) and states \( \pS \); under mild physical assumptions we derived an algebraic structure on \( \pA \). We further realized, that the resulting algebra can be identified with a set of self-adjoint elements of a \Cs-algebra \( \mfA \).

Now we have the entire algebraic machinery for \Cs-algebras. In particular, by means of the GNS construction we know that each state \( \omega \) on \( \mfA \) yields a cyclic representation \( (\pi_{\omega}, H_{\omega}, x_{\omega}) \).
In particular, by the \( \star \)-homomorphism property self-adjoint elements of  \( \mfA \) are mapped to self-adjoint operators on the Hilbert space \( H_{\omega} \).
By the ``Expectation'' axiom of Dirac-von Neumann and the unitarily equivalence clause of the GNS Theorem, we see that the framework of quantum mechanics is unitarily equivalent to the so derived GNS representation.
Thus, we recovered the first three Dirac-von Neumann axioms.

The first three axioms of Dirac-von Neumann are of general nature. They do not refer to a quantum system, but describe a physical system in general. The fourth and the fifth axioms, on the other hand, refer specifically to a quantum system. It is for this reason that with our general construction we recovered the general axioms. As we will see in the next section, the fourth Dirac-von Neumann axiom will guide us to \textit{define} an explicit \Cs-algebra for quantum mechanics. The fifth axiom drops entirely by the result of the Stone-von Neumann theorem.

Additionally we notice that in the algebraic formalism, the structure of quantum mechanics becomes derived, just like the structure of the phase space of classical mechanics was derived from the commutative \Cs-algebra. In particular, we see that the Hilbert space of quantum mechanics is not a fixed construction of the framework, but a property that arises from a choice of a physical state, i.e.\  from the method of preparation of a physical system. Different methods of preparation may lead to different Hilbert spaces. Indeed, in the standard approach to quantum mechanics, such transitions between states that lead to different Hilbert spaces  typically result in exploding physical quantities such as the magnetization, susceptibility, conductivity, and the like. If this happens, we say that a system undergoes a \textit{phase transition}.
Such an explosion of physical quantities results from the fact, that the scalar product on the originally defined Hilbert space becomes inappropriate for the new physical system, since for the new system a new Hilbert space is required.
 In the algebraic approach, such s transition corresponds to the change between states \( \omega_{1} \) and \( \omega_{2} \) that lead to two not unitarily equivalent GNS representations of the \Cs-algebra. In this way, one avoids divergent quantities and can describe phase transitions in a mathematically careful way.

\newpage
 \section{A quantum mechanical particle}
 \begin{sectionmeta}
     This section is based on several books and papers. The semi-historical note is motivated by the discussion in the Heisenberg's book \cite{heisenberg1949physical} and his papers \cite{Heisenberg:1925vh, Heisenberg:1927vy}. The uniqueness of the Weyl algebra is from \cite{bratteli2012operatorII}. The proof of the Stone-von Neumann theorem and the Schr\"odinger representation of the Weyl algebra is adopted from \cite{strocchi2008introduction} and the original paper of von Neumann \cite{Neumann:1931ve}.

 \end{sectionmeta}
 After the last section we know that we need a  \Cs-algebra of observables to define a physical system. For classical physics the algebra has to be commutative. In this case, we can choose the algebra of continuous functions on the phase space.
 The question is, what \Cs-algebra do we choose to describe quantum mechanics.
To answer this question we use the canonical commutation relations and the Heisenberg uncertainty principle. Before we proceed, we mention a semi-historical note that shall motivate our approach.

\subsection{Semi-historical note}
Heisenberg realized that the physical quantities like position, velocity, acceleration, and the like, may not be defined for sub-atomic systems (systems of small size).
He argued that all physical properties of a particle are \textit{defined} by  measurements. A position of a particle is nothing else than an outcome of a measurement of position. If no experiment can be devised to measure a certain quantity then this quantity does not ``exist''. The concept of measurement is unavoidably connected to the property itself.

As an example he considered an electron. One way to measure the electrons position is to look at it under the microscope. The image we see under the microscope appears from the lightrays that scattered from the electron. The shorter the wave length of used light, the sharper will be the image and the better will be the resolution of the electron's position. From the photoelectric effect, described by Einstein several years prior, we know that momentum \( p \) of light particles --- the photons --- is indirectly proportional to the wavelength \( \lambda \) such that
\begin{align}
    p = \frac{\hbar}{2\pi \lambda},
\end{align}
where \( \hbar \) is the Planck's constant (up to factors of \( 2 \pi \)).
In other words, the sharper the image, the higher the momentum of photons. As the light scatters from the electron, the high momentum of photons will transfer onto the momentum of the electron --- a process described and measured by the Compton scattering. Moreover, this transfer of momentum will happen discontinuously, so after measuring the position the momentum of the electron becomes unknown.
Heisenberg estimated the accuracy of the measurements to be
\begin{align}
    \Delta x \Delta p \geq \frac{\hbar}{2 \pi},
\end{align}
where \( \Delta x \) and \( \Delta p \) denote the variance of position and momentum measurements.
He suggested that this uncertainty is a fundamental principle of nature; it is impossible to measure position and momentum of an electron with better simultaneous precision. Today we call it \textit{the Heisenberg's uncertainty principle}.

To describe this phenomenon quantitatively, he suggested to use the wave picture of the electron: the electron is a wave packet, electron's position is the localization of the packet, and electrons velocity is the group velocity of the packet. The sharper the wave packet is localized in space the more spread it is in momentum, thereby causing the packet to disperse in time. Thus, a sharply peaked wave packet quickly spreads and its group velocity is not clearly defined. The analysis of the problem requires care, because the aforementioned wave packet describes the \textit{probability distribution} for the electron and not the electron itself. Moreover, the dispersion of the packet depends on the implemented dynamics, as for non-linear processes the wave package may keep its shape for an extended time period. Nevertheless,  the above heuristic argument already captures the main idea of the rigorous calculation.

Heisenberg's idea was groundbreaking and is considered as the initial formulation of quantum mechanics. Nevertheless, it was Born and Jordan who realized that Heisenberg defined observables as infinitely large matrices, i.e.\ linear operators on an infinitely dimensional Hilbert space \cite{born1925quantenmechanik}. In particular, being operators, observables may not always commute. Using Heisenberg's uncertainty principle Born and Jordan derived that the position and the momentum of an quantum particle must satisfy what we nowadays know as the \textit{canonical commutation relations} or CCRs,
\begin{align}\label{eq:ccr}
    [x_{i},x_{j}] = 0 = [p_{i},p_{j}],
    &&
    [x_{i},p_{j}] = x_{i}p_{j} - p_{j}x_{i} = \imath \hbar \mathds{1},
    \qquad \big( i,j = 1,\dots ,d \big),
\end{align}
where \( d \) is the number of degrees of freedom for the electron (e.g.\ a free particle moving in a \( d \)-dimensional space; no spin is considered); and \( x_{j} \) (resp. \( p_{j} \)) denotes the position (resp. momentum) observable in the \( j \)th direction. Often we define
\( x = (x_{1},\dots, x_{d}) \) and \( p = (p_1,\dots ,p_{d}) \) as operators on \( d \) copies of the Hilbert space \( H \), and phrase \eqref{eq:ccr} as a single CCR relation
\begin{align}\label{eq:CCR}
    [x,p] = xp - px = \imath \hbar \mathds{1}.
\end{align}
In summery, the Heisenberg's uncertainty principle implies that the position and the momentum observables of an electron do not commute.
Moreover, the CCRs \eqref{eq:ccr} imply that either of the operators, \( x_i \) or \( p_i \) (or both), must be unbounded for every \( i = 1,\dots,d \).
To see this, we write the CCRs (dropping the index for clarity) as \( xp = \imath \hbar + px\). Then,
\begin{align*}
    x^{2}p - px^{2}
    &= x (xp) - px x\\
    &= x (\imath \hbar + px) - px x \\
    &= \imath \hbar x + (xp) x - px x \\
    &=\imath \hbar x + (\imath \hbar + px) x - px x
    =  \imath \hbar \ 2x.
\end{align*}
By induction we then get
\begin{align}
    [x^{n} , p ] = x^{n} p - p x^{n}= \imath \hbar \ n x^{n-1}.
\end{align}
The above relation implies that the norms of the operators \( x \) and \( p \) satisfy
\begin{align}\label{eq:unbounded x or p}
    \hbar \ n \|x^{n-1}\|  = \| [x^{n}, p ] \| \leq \| px^{n} \| + \| x^{n} p\| \leq 2 \| p \| \| x \| \| x^{n-1} \|.
\end{align}
Now, \( \| x^{n-1} \| \) cannot be zero, otherwise \( 0 = \|x ^{n-1} \| = \| x \|^{n-1} \) would imply \( x =0 \) which violates the CCR.
Thus, \eqref{eq:unbounded x or p} reduces to
\begin{align}
    \hbar \ \frac{n}{2} \leq \|q \| \| p \|
    \qquad \big( n \in \bbN \big).
\end{align}
Because \( n \) can be arbitrary large, at least one of the operators has to have an unbounded norm. This, in turn, implies that we have to be careful with the multiplication of operators \( x_{i} \) and \( p_{i} \). In particular, this multiplication can be defined only on an empty set --- implying that it is not meaningful. Moreover, even if the multiplication can be defined on the dense subset of the Hilbert space, it cannot be defined on the whole space. Because of this, the whole concept of ``commutation'' for unbounded operators is non-trivial.
Many technical complications arise when we want to formulate CCRs using clean mathematics.

\vspace{0,5cm}
We now return to our algebraic approach. In the algebraic approach, the mathematics is simple because all elements have finite norm. We only need to define what algebra is suitable to define observables of a particle. Ideally, we would like to use \( x \) and \( p \) to motivate the algebra, but since these operators are unbounded, they cannot generate a \Cs-algebra. Despite this, we can use \( x \) and \( p \) in a heuristic way (not rigorous calculations), to get an intuition for the right \Cs-algebra. One way to do this was suggested by Weyl and it leads us to the definition of the Weyl algebra.

\subsection{The Weyl algebra}

To define the Weyl algebra we use the operators \( x \) and \( p \) as if they were bounded, or at least as if there were no problems with unbounded operators. The resulting calculation will serve as a ``hint'' towards an appropriate \Cs-algebra of observables.

Since the concept of ``commutation'' is complicated for unbounded operators, Weyl suggested to look at the operators
\begin{align}
    U(\alpha) = e^{\imath \alpha x}
    &&
    V(\beta) = e^{\imath \beta p}
    \qquad \big( \alpha, \beta \in \bbR^{d} \big),
\end{align}
where we use the abbreviation \( \alpha x = \sum_{i=1}^{d}  \alpha_{i}x_{i} \) and \( \beta p = \sum_{i=1}^{d} \beta_{i} p_{i} \). The operators \( U(\alpha) \) and \( V(\beta) \) are called the \textit{Weyl operators}, and  the vectors \( \alpha \) and \( \beta \) \textit{label} these operators.
Since \( x \) and \( p \) suppose to represent observables we assume them to be self-adjoint. In this case, the Weyl operators are bounded and can be defined via the spectral calculus.
We want to use the Weyl operators to generate our algebra of observables. For this we need to define a product between them.
To this end we use the CCRs and the Baker-Campbell-Hausdorff formula.
\begin{info}[Reminder: the Baker-Campbell-Hausdorff formulat.] A special case of the Baker-Campbell-Hausdorff formula states that if the commutator \( [X,Y] \) of two operators \( X \) and \( Y \) commutes with both, \( X \) and \( Y \), then
\begin{align}
    e^{X}e^{Y} = e^{X + Y + \tfrac{1}{2}[X,Y]}.
\end{align}
\begin{proof}
For the proof and a more general statement consult \cite[p. 25 prop. 2]{rossmann2006lie}.
\end{proof}
\end{info}
Formally,  the CCRs say that \( [x,p]\) is proportional to the identity; thus it commutes with every operator, and we can use the Baker-Campbell-Hausdorff formula to define the following replacement of the CCR for the Weyl operators,
\begin{align}\label{eq:weyl rel}
    U(\alpha) V(\beta) = V(\beta) U(\alpha) e^{-\imath \hbar \alpha \beta}, &&
    U(\alpha)U(\beta) = U(\alpha + \beta), &&
    V(\alpha)V(\beta) = V(\alpha + \beta),
\end{align}
where \( \alpha \beta = \sum_{i=1}^{d} \alpha_{i}\beta_{i} \) is the Euclidean scalar product.
These relations are called the \textit{commutation relations in the Weyl form} or simply the \textit{Weyl relations}. Even though \( U(\alpha) \) and \( V(\beta) \) together with \eqref{eq:weyl rel} is enough to proceed, it will be more convenient in the following to work with a single operator instead of two. To this end, we use vectors in \( \bbR^{d} \times \bbR^{d} \), such that the first \( d \)-components represent the vector \( \alpha \), and the second \( d \)-components represent the vector \( \beta \). Then for \( v = (\alpha, \beta)\) and \(w= (\gamma,\delta)\) in \( \bbR^{d} \times \bbR^{d} (=\bbR^{2d})\) we define the following operation,
\begin{align}\label{eq:symp}
    \sigma(v,w)
= \alpha \delta - \gamma \beta
=
    \left(
    \begin{array}{c}
        \alpha \\
        \beta
    \end{array}
    \right)
    \left(
    \begin{array}{cc}
        0 & \mathds{1}  \\
        -\mathds{1} & 0
    \end{array}
    \right)
        \left(
        \begin{array}{c}
            \gamma \\
            \delta
        \end{array}
        \right).
\end{align}
Especially, \( \sigma(v,v) = 0 \) for every \( v \in \bbR^{2d} \). The mapping \( \sigma \colon \bbR^{2d} \times \bbR^{2d}  \to \bbR \) is a particular case of what is called a \textit{non-degenerate symplectic bilinear form}. It has a profound meaning in classical mechanics. Formulating our \Cs-algebra using \( \sigma \) is thus preferable as it can give us a general procedure how to define a quantum algebra from the classical case. In other words: how to \textit{quantize} a classical system.

Indeed, it turns out that if we use \( v = (\alpha,\beta)  \in \bbR^{2d} \) and set
\begin{align}
    W(v) = e^{-\imath \hbar \frac{\alpha \beta}{2}} V(\beta) U(\alpha) = e^{\imath \hbar \frac{\alpha\beta}{2}} U(\alpha) V(\beta).
\end{align}
then the Weyl relations can be derived from the following relations
\begin{align}\label{eq:weyl rel2}
    W(v)^{\star} = W(-v),
    &&
    W(v)W(w) = W(v+ w) e^{-\imath \hbar \frac{\sigma(v,w)}{2}}.
\end{align}
In the following, we will not use \( U(\alpha) \) and \( V(\beta) \) but instead only the \( W(v) \)'s labeled by vectors in \(\bbR^{2d}  \). For this reason we will call \( W(v) \) the Weyl elements and refer to \eqref{eq:weyl rel2} as to the Weyl relations.
Notice, however, that we can recover the operators \( U(\alpha) \) and \( V(\beta) \) by using vectors of the form \( v=(\alpha, 0) \in \bbR^{2d} \) or \( v = (0,\beta) \in \bbR^{2} \), respectively.
Even though, the above calculations were formal, we can  now use \eqref{eq:weyl rel2} as a starting point to define a clean \Cs-algebra.

We start with the set of \textit{abstract} elements \(\{ W(v) \colon v \in \bbR^{2d}  \} \), that transform under a \( \star \)-operation and satisfy the multiplication relations according to \eqref{eq:weyl rel2}. Then we consider an abstract algebra \( \mcA_{W} \) generated by this set.
That is \( \mcA_{W} \) comprises elements of the form
\begin{align}\label{eq:elements of Weyl algebra}
    a_{1}W(v_{1}) + a_{2} W(v_{2}) + \cdots + a_{n} W(v_{n})
\end{align}
for \( a_{1},\dots,a_{n} \in \bbC \) and \( v_{1},\dots, v_{n} \in \bbR^{2d} \). Notice that due to the multiplication relations, all monomials of Weyl operators can be rewritten as a single Weyl operator multiplied by a complex number.
To form a \Cs-algebra, we need to equip \( \mcA_{W} \) with a norm that satisfies the \Cs-property. But what norm shall we choose?
The following theorem provides an answer to this question. For the proof see \cite[thm. 5.2.8]{bratteli2012operatorII}.

\begin{thm}\label{thm:unique weyl}
Let \( V \) be a real linear space equipped with a non-degenerate symplectic bilinear form \( \sigma \). Let \( \mfA_{1} \) and \( \mfA_{2} \) be two \Cs-algebras generated by nonzero elements \( W_{i}(v) \), \( v \in V \), satisfying \eqref{eq:weyl rel2}. Then there exists a unique \( \star \)-isomorphism \( \alpha\colon \mfA_{1} \to \mfA_{2} \) such that
\begin{align}
    \alpha \big(W_{1}(v) \big) = W_{2}(v),
    \qquad (v \in V).
\end{align}
\end{thm}
Since \( \star \)-isomorphisms are isometric, there is essentially a unique \Cs-norm which makes \( \mcA_{W} \) a \Cs-algebra. In other words, the Weyl relations \eqref{eq:weyl rel2} generate a unique \Cs-algebra, up to \( \star \)-isomorphisms --- the \textit{Weyl algebra} \( \mfW \).

From the Weyl relations \eqref{eq:weyl rel2} we can state some basic, nonetheless important properties of the Weyl algebra and the Weyl operators.

\begin{prop}\label{prop: weyl}
Let \( \mfW \) be the Weyl algebra generated by elements \( \{ W(v)\colon v \in \bbR^{2d} \} \) that satisfies the Weyl relations. Then
\begin{enumerate}
\item\label{prop: weyl it:1}
\( \mfW \) is unital and \( W(0) = I \) is the unit of \( \mfW \);

\item\label{prop: weyl it:2}
for each \( v \in \bbR^{2d} \) the element \( W(v) \) is unitary;

\item\label{prop: weyl it:3}
for each non-zero \( v\in \bbR^{2} \) the spectrum of \( W(v) \) is the entire unit circle of \( \bbC \).
\end{enumerate}
\end{prop}
\begin{proof}
For \ref{prop: weyl it:1}, observe from the Weyl relations that
\begin{align}
W(v)W(0) =W(v) = W(0)W(v)
\qquad \big( v \in \bbR^{2d} \big).
\end{align}
Since each element in \( \mfW \) is a limit of finite complex linear combinations of the Weyl elements, for every \( A \in \mfW \) we have
\(    A W(0) = W(0)A = A,
\) and thus \( W(0) = I \) is the unit of the Weyl algebra; in particular, \( \mfW \) is unital.

For \ref{prop: weyl it:2}, we use again the Weyl relations to see that
\begin{align}
     W(v)W(v)^{\star} = W(v)W(-v) = I = W(-v)W(v) = W(v)^{\star}W(v)
     \qquad \big( v \in \bbR^{2d} \big).
\end{align}
Thence, \( W(v) \) is unitary.

For \ref{prop: weyl it:3}, observe that since \( W(v) \) is unitary, by Proposition \ref{prop:spectral radius and norm in cstar} \ref{it:c}, its spectrum is a subset of the unit circle in \( \bbC \) .
Suppose \( v \neq 0 \). Let \( \lambda \) be in the spectrum of \( W(v) \) and \( \alpha \) be a real number in the interval \( [0, 2\pi) \). Choose \( w \in \bbR^{2d} \) such that \( \sigma(v,w) = \alpha / \hbar \), which is always possible since \( v \neq 0 \) and \( \sigma \) is linear and non-degenerate. Then,
\begin{align}
    W(v) - e^{\imath \alpha} \lambda
    = e^{-\imath \alpha} \big(e^{\imath \alpha} W(v) - \lambda \big)
    = e^{-\imath \alpha} \Big(W(w)\big(W(v) -\lambda \big)W(w)^{\star} \Big).
\end{align}
The right hand side does not have a two-sided inverse, and consequently \( e^{\imath \alpha} \lambda\) is in the spectrum of \( W(v) \). It follows that the spectrum is invariant under rotation and since \( \lambda \) is an element of the unit circle, the spectrum must be the entire unit circle in \( \bbC \).
\end{proof}

\begin{info}[Remark.]
The Theorem \ref{thm:unique weyl} holds even if \( V \) is infinite dimensional. In this case, the Weyl algebra describes a system with infinitely many degrees of freedom. Such a theory is called (non-relativistic) \textit{quantum field theory}. Opposed to it, quantum mechanics describes systems with a finite number of degrees of freedom and thus the Weyl algebra is labeled by a finite dimensional \( V \).
\end{info}

\paragraph{To conclude:}
Using the canonical commutation relations that follow from the Heisenberg uncertainty principle we motivated a \Cs-algebra for the quantum particle. This algebra is essentially unique and we call it the Weyl algebra \( \mfW \). To connect our algebraic formalism to the usual (Schr\"odinger) quantum mechanics, we need to represent the algebra of observables as bounded linear operators on a Hilbert space. As we have seen, we can use the GNS construction for this. Nevertheless, we have also seen, that not all representations are necessarily unitarily equivalent. And if the Weyl algebra has different unitarily inequivalent representations, then these representations will yield physically different descriptions of the quantum particle. Which one should we then choose?
In particular, if different inequivalent representations exist, what is their physical significance? And then, how many different  inequivalent representations of the Weyl algebra can we construct. All these questions were answered by von Neumann. The answer is known as the \textit{Stone-von Neumann theorem}.

\subsection{Stone-von Neumann theorem}
We now want to investigate unitarily inequivalent representations of the Weyl algebra. Recall that, unitarily inequivalent representations lead to different Hilbert spaces on which we can represent the \Cs-algebra as a (sub)algebra of bounded operators. If the \Cs-algebra suppose to define a physical system, then the Hilbert spaces (especially their inner products) define the expectation values. Thus they are physically relevant. If several inequivalent representations exist then we need to select the one which is best suited for the description of the system.

Recall, that irreducible representations are the building blocks of all representations. Thus, we only need to characterize irreducible representations of the Weyl algebra. We can restrict the class of ``interesting'' representations even further by requiring a mild but physically well motivated regularity condition. To formulate this condition we need to refresh the definition of the strong operator topology.
\begin{defi}[\textbf{Strong operator topology}]
Let \( B(H) \) be the set of bounded linear operators on the Hilbert space \( H \). For every \( x \in H \) define a map from \( B(H) \) into \( H \) by
\begin{align}
    A \to Ax.
\end{align}
The strong operator topology is the weakest topology such that each of these maps is continuous.
In this topology, a sequence of operators \( (A_{n}) \) converges to an operator \( A \), if the sequence of vectors \( ( A_{n}x ) \) converges to the vector \( Ax\) for every \( x \in H \). It is a consequence of the uniform boundedness principle that a strong limit of a sequence of bounded linear operators is again a bounded operator.
\end{defi}

\begin{defi}
An operator-valued function \( U(t) \) defined on \( \bbR \) and acting on the Hilbert space \( H \) is called a \textit{strongly continuous one-parameter unitary group} if
\begin{enumerate}
\item
for each \( t,s \in \bbR \), the operator \( U(t) \) is unitary, and \( U(t+s) = U(t)U(s) \);

\item
and for \( x \in H \), the vector \( U(t) x \) converges to the vector \( U(t_{0})x \) as \( t \to t_{0} \).
\end{enumerate}
\end{defi}

We mention some well known results of functional analysis to point out the importance of the strong operator topology.
\begin{info}[Stone's theorem.]
    We state the theorems and refer to e.g.\ \cite[thm. VIII.7 and thm. VIII.8]{reed2012methods} for a more complete statement and the proofs.

    \begin{thm}\label{thm:converse of stone}
Let \( A \) be a self-adjoint operator on a Hilbert space \( H \). The operator-valued function
\begin{align}
U(t) = e^{\imath tA}
\qquad \big( t \in \bbR \big),
\end{align}
defines a strongly continuous one-parameter unitary group. Moreover, for \( x \in dom(A) \), the limit \( \lim\limits_{t \to 0} \tfrac{1}{t} \big( U(t)x - x \big)\) exists and is equal to \( \imath Ax \); conversely, if \( \lim\limits_{t\to 0} \tfrac{1}{t} \big( U(t)x - x \big) \) exists, then \( x \in dom(A) \).
    \end{thm}
The Stone's theorem gives the converse of Theorem \ref{thm:converse of stone}.
\begin{thm}
If \( U(t) \) is a strongly continuous one-parameter unitary group on a Hilbert space \( H \), then \( U(t) = e^{\imath t A} \) for some self-adjoint operator \( A \) on \( H \).
\end{thm}
\end{info}

Stone's theorem states that we can ``differentiate'' strongly continuous one-parameter unitary groups with respect to the parameter. The derivative being the \textit{generator} of the group. The generator need not be bounded, and Theorem \ref{thm:converse of stone} provides a useful criteria for \( x \in H \) to be in the domain of the generator.
Since we motivated the Weyl algebra by considering operators \( e^{\imath \alpha x} \) and \( e^{\imath \beta p} \), we would expect that the generators of the Weyl elements would correspond to the position and momentum operators. However, in order for the generators to exist, we need to make sure that the Weyl operators are represented as a strongly continuous unitary group. This is the reason for the following definition.

\begin{defi}
A representation \( \pi \) of the Weyl algebra \( \mfW \) on a separable Hilbert space \( H \) is called \textit{regular} if \( \pi \big(W(v) \big)\) is strongly continuous in \( v \in \bbR^{2d} \).
\end{defi}

For \( v \in \bbR^{2d}\) of the form \( v_{n}=( 0,\dots, 0, t,0, \dots, 0) \in \bbR^{2d} \) (where all except the \( n \)th element are zero), the Weyl operators define one-parameter unitary groups. In conjunction with the Stone's theorem, the regularity condition assures that the operators \( \pi \big( W(v_{n}) \big) \) will have generators. Hence, we have a chance to define positions and momenta from the Weyl algebra. Since this is desirable on physical grounds, we restrict our attention only to regular representations. Thus we aim to characterize all regular irreducible representations of the Weyl algebra.

\begin{info}[Remark:]
The regularity condition says that if \( (v_{n}) \) is a sequence in \( \bbR^{2d} \) that converges to a vector \( v \), then the sequence of operators \( ( \pi\big(W(v_{n}) \big)) \) will strongly converge to an operator \( \pi\big(W(v)\big) \). Notice, that this is not necessarily the case in the \Cs-norm since
\begin{align*}
    \| W(v_{n}) - W(v) \|^{2} = \| W(-v) \big(e^{\imath \hbar \frac{\sigma(v_{n}, v)}{2}} W(v_n - v)  - I \big) \|^{2}
    =
    \|e^{\imath \hbar \frac{\sigma(v_{n}, v)}{2}} W(v_n - v)  - I \|^{2}.
\end{align*}
Assume that \( v_{n} - v \neq 0 \) for all \( n \in \bbN \), then by Proposition \ref{prop: weyl} \ref{prop: weyl it:3}, the spectrum of \( W(v_{n} - v) \) is the entire unit circle in \( \bbC \) and we have
\begin{align}
    \|e^{\imath \hbar \frac{\sigma(v_{n}, v)}{2}} W(v_n - v)  - I \|
    \geq  \sup \{ |\lambda - 1| \colon \lambda \in \bbC, \ | \lambda| = 1 \} = 2.
\end{align}
Thus the sequence \( \big(W(v_{n}) \big) \) is not even Cauchy in the \Cs-norm.
\end{info}

If \( U \colon \bbR^{2d} \to B(H) \) is a strongly continuous operator-valued function bounded such that \( \sup \{ \| U(v) \| \colon v \in \bbR^{2d} \} = \|U \| \leq \infty \), and \( f \) is a continuous and integrable function with respect to the measure \( \mu \) on \( \bbR^{2d} \), then the operator
\begin{align}
    A = \int_{\bbR^{2d}} f(v) U(v) \ d\mu(v),
\end{align}
can be defined as a limit of Riemann sums. We will use this construction in the following proof, therefore we review it in some detail now.

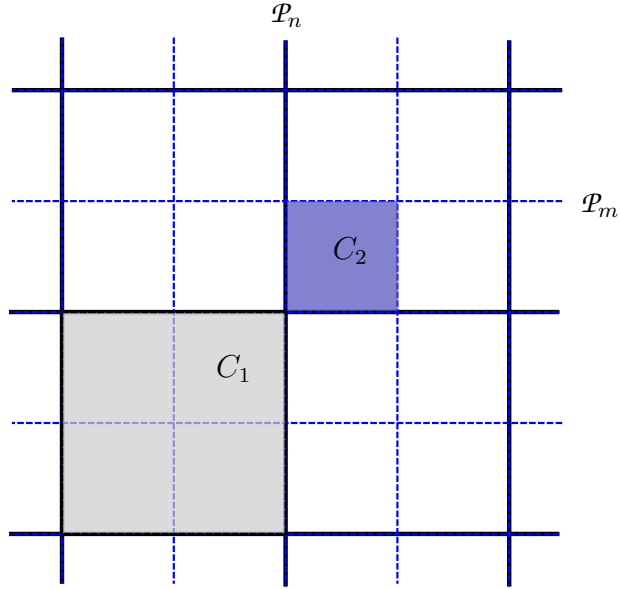
\begin{figure}[ht]
\centering
\def\svgwidth{8cm}
\begingroup%
  \makeatletter%
  \providecommand\color[2][]{%
    \errmessage{(Inkscape) Color is used for the text in Inkscape, but the package 'color.sty' is not loaded}%
    \renewcommand\color[2][]{}%
  }%
  \providecommand\transparent[1]{%
    \errmessage{(Inkscape) Transparency is used (non-zero) for the text in Inkscape, but the package 'transparent.sty' is not loaded}%
    \renewcommand\transparent[1]{}%
  }%
  \providecommand\rotatebox[2]{#2}%
  \newcommand*\fsize{\dimexpr\f@size pt\relax}%
  \newcommand*\lineheight[1]{\fontsize{\fsize}{#1\fsize}\selectfont}%
  \ifx\svgwidth\undefined%
    \setlength{\unitlength}{309.99756427bp}%
    \ifx\svgscale\undefined%
      \relax%
    \else%
      \setlength{\unitlength}{\unitlength * \real{\svgscale}}%
    \fi%
  \else%
    \setlength{\unitlength}{\svgwidth}%
  \fi%
  \global\let\svgwidth\undefined%
  \global\let\svgscale\undefined%
  \makeatother%
  \begin{picture}(1,0.95194382)%
    \lineheight{1}%
    \setlength\tabcolsep{0pt}%
    \put(0,0){\includegraphics[width=\unitlength,page=1]{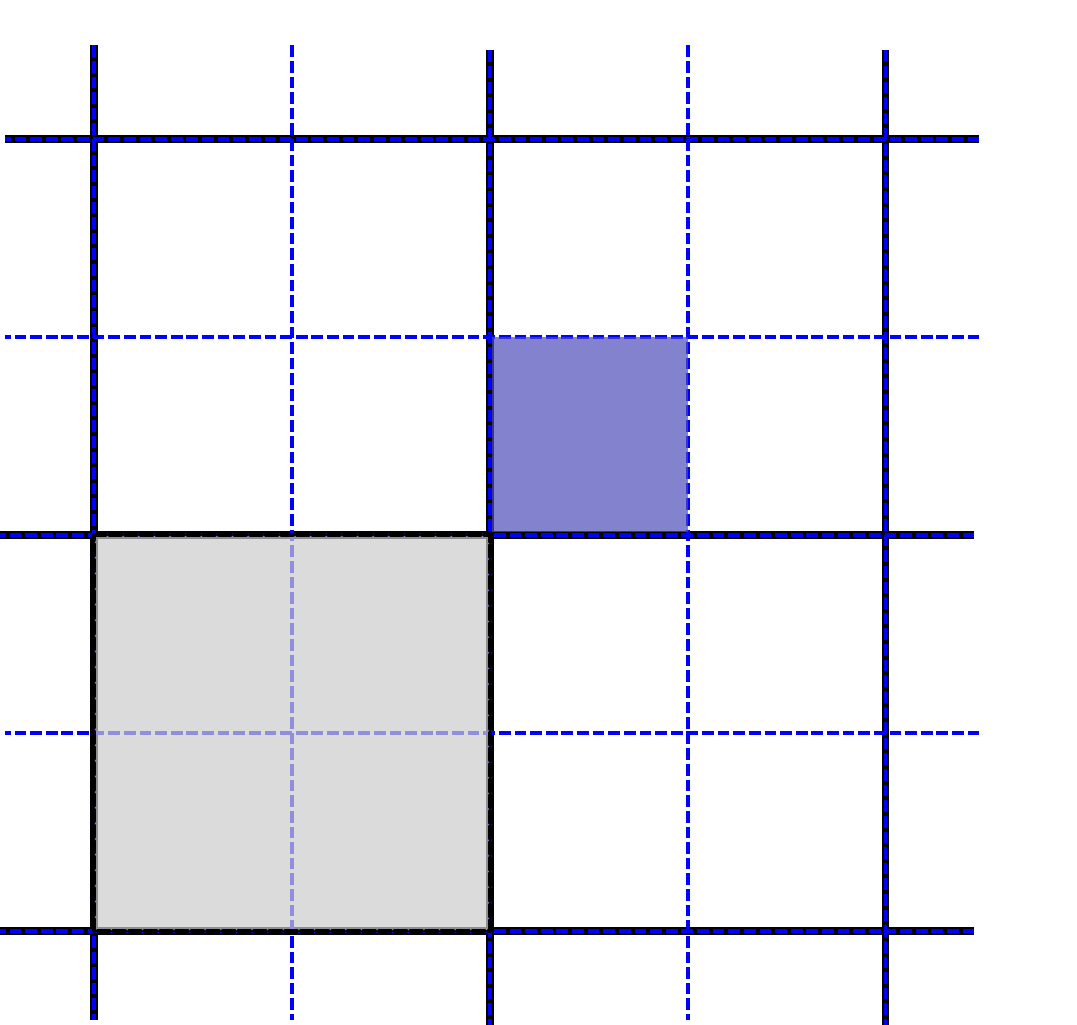}}%
    \put(0.42993219,0.9360975){\color[rgb]{0,0,0}\makebox(0,0)[lt]{\lineheight{1.25}\smash{\begin{tabular}[t]{l}$\pP_{n}$\end{tabular}}}}%
    \put(0.94031611,0.62089524){\color[rgb]{0,0,0}\makebox(0,0)[lt]{\lineheight{1.25}\smash{\begin{tabular}[t]{l}$\pP_{m}$\end{tabular}}}}%
    \put(0.53237946,0.54217993){\color[rgb]{0,0,0}\makebox(0,0)[lt]{\lineheight{1.25}\smash{\begin{tabular}[t]{l}$C_{2}$\end{tabular}}}}%
    \put(0.34086169,0.34709323){\color[rgb]{0,0,0}\makebox(0,0)[lt]{\lineheight{1.25}\smash{\begin{tabular}[t]{l}$C_{1}$\end{tabular}}}}%
  \end{picture}%
\endgroup%

\caption{Example of two partitions \( \pP_{n}, \, \pP_{m} \) of \( \bbR^{2} \). Cells of partition \( \pP_{n} \) are the large squares, like the cell \( C_{1} \). Cells of partition \( \pP_{m} \) are the smaller squares, like the cell \( C_{2} \). Partition \( \pP_{m} \) is finer than partition \( \pP_{n} \) and thus we can order them as \( \pP_{m} \geq \pP_{n} \).}
\label{fig:partition}
\end{figure}

\begin{info}[Integrals via Riemann sums.]
    For a bounded set \( O \subset \bbR^{2d} \)
    we write \( \pP_{n}(O) = \{ C_{1}, \dots C_{n}\} \) to denote a partition of \( O \) into \( n \) (open) cells \( C_{i} \subset O \).
    For two partitions \( \pP_{1}(O) \) and \( \pP_{2}(O) \) we write \( \pP_{2}(O) \geq \pP_{1}(O) \) and say that ``\( \pP_{2}(O) \) is finer than \( \pP_{1}(O) \)'' if each cell \( C \) of \( \pP_{1}(O) \) partitions into a collection \( \pP_{2}(C) \) of cells in \( \pP_{2}(O) \) such that (see Figure \ref{fig:partition})
    \begin{align}\label{eq:subpartition}
        \overline{C} = \overline{\bigcup_{K \in \pP_{2}(C)} K}
        &&
        \mu(C) = \sum_{K \in \pP_{2}(C)} \mu(K).
    \end{align}
    Given a partition \( \pP_{n}(O) \), choose one vector \( v_{C} \) from each cell \( C \) and define an operator
\begin{align}
    A_{n} = \sum_{C \in \pP_{n}(O)} f(v_{C}) U(v_{C}) \ \mu(C).
\end{align}

Observer, that since \( f U \colon \bbR^{2d} \to B(H)\) is strongly continuous, for \( x \in H \) and \( \epsilon >0 \) there is a fine enough partition \( \pP_{0} (O)\) such that: for each cell \( C \) of \( \pP_{0}\)
\begin{align}\label{eq:fine partition}
    v,w \in C \qquad \text{ implies } \qquad \| f(v)U(v) - f(w) U(w) \| \leq \frac{\epsilon}{\mu(O)}.
\end{align}
Then for any \( \pP_{1} (O) \geq \pP_{0}(O) \) we have
\begin{align*}
    \| A_{0}x - A_{1}x \|
    &= \Big\| \sum_{C \in \pP_{0}(O)} f(v_{C}) U(v_{C}) \mu(C)
    - \sum_{K \in \pP_{1}(O)} f(v_{K}) U(v_{K}) \mu(K)\Big\| \\
    &= \Big\| \sum_{C \in \pP_{0}(O)} \sum_{K \in \pP_{1}(C)}
    \big(
    f(v_{C}) U(v_{C})
    -  f(v_{K}) U(v_{K})
    \big)\mu(K)
    \Big\|\\
    &\leq
     \sum_{C \in \pP_{0}(O)} \sum_{K \in \pP_{1}(C)}
    \Big\| f(v_{C}) U(v_{C})
    -  f(v_{K}) U(v_{K})
    \Big\|\mu(K) \\
    &\leq
    \epsilon,
\end{align*}
where the second equality holds by \eqref{eq:subpartition} since \( \pP_{1}(O) \) is finer than \( \pP_{0}(O) \).
Thus, if \( \big(\pP_{n} (O) \big)  \) is an increasing sequence of partitions, i.e.\ \( \pP_{n}(O) > \pP_{m} (O) \) when \( n > m \), then for every \( x \in H \) there is a \( k \in \bbN \) for which \( \pP_{k}(O) \) is fine enough to satisfy \eqref{eq:fine partition} and the sequence of vectors \( (A_{n}x) \) is Cauchy; its limit defines a linear mapping \( A_{O} \) such that,
\begin{align}\label{eq:int on O}
    x \mapsto \lim_{n\to \infty } A_{n}x = \int_{O}  f(v) U(v)x \ d \mu(v) \eqqcolon A_{O}x
    \qquad \big( x \in H \big).
\end{align}
Since the strong limit of bounded linear operators is bounded and thus \( A_{O} \) defines the integral of \( fU \) over \( O \).

To see that the limit does not depend on the partition sequence, choose two increasing sequences \( \big(\pP_{n}(O) \big) \) and \( \big( \tilde{\pP}_{n}(O) \big) \). Then there is a sequence \( \big( \overline{\pP}_{n} (O) \big)\) that is finer than both \( \pP_{n}(O) \) and \( \tilde{\pP}_{n}(O) \); for example define the cells of \( \overline{\pP}_{n}(O) \) as non-empty intersections of each cell of \( \pP_{n}(O) \) with each cell of \( \tilde{\pP}_{n}(O) \). These three partition sequences give rise to three sequences of operators \( (A_{n}) \) for the partitions \( \big(\pP_{n}(O) \big) \), \( (B_{n}) \) for the partitions \( \big(\tilde{\pP}_{n}(O )\big) \), and \( (C_{n}) \) for the partitions \( \big( \overline{\pP}_{n}(O) \big) \). From the above, these operator sequences converge to \( A \), \( B \), and \( C \), respectively.
Since \( \overline{\pP}_{n} (O) \) is finer than other two we have that \( C_{n} \to A \) and \( C_{n} \to B \), and by the uniqueness of limit points we get \( A = C = B \).

Now we need to ``extend'' \( O \) to the entire space \( \bbR^{2d} \). For this let \( \bbB(n) \) denote an open ball in \( \bbR^{2d} \) with radius \( n \) centered at the origin. Then for \( n > m \)
\begin{align*}
    \Big\| \int_{\bbB(n)} f(v) U(v) \ d\mu(v) - \int_{\bbB(m)} f(v) U(v) \ d\mu(v) \Big\|
    &\leq \int_{\bbB(n) /\bbB(m)} |f(v) | \| U(v) \| \ d \mu(v)\\
    &\leq \| U \| \int_{\bbB(n) /\bbB(m)} |f(v) | \ d \mu(v).
\end{align*}
Because \( f \) is integrable, the right hand side of the above estimate goes to zero as \( m \) goes to infinity.
Therefore the sequence of integrals over increasing balls is a Cauchy sequence of operators. The limit point of that sequence defines our desired integral
\begin{align}
    A = \int_{\bbR^{2d}}  f(v) U(v) \ d \mu(v).
\end{align}
Notice that he matrix elements of \( A \) are the usual (Riemann) integrals of matrix elements of \( U(v) \) since matrix elements of Riemann sums are the Riemann sums of matrix elements. Explicitly, for \( x,y \in H \) we have
\begin{align}\label{eq:weak integral}
    \scalp{x, A y} = \int_{\bbR^{2d}} f(v) \scalp{x, U(v) y} \ d \mu(v).
\end{align}
Since \(v \to \scalp{x,U(v)y} \) is a continuous bounded function, the integral is well defined. Indeed we could use \eqref{eq:weak integral} as a \textit{definition} for the operator \( A \) in terms of its matrix elements.
This construction of integral operators is often useful when we need to define operators that are invariant under certain transformations. For this one has to find the integral kernel \( f(v) \) such that the effect of the transformation can be ``averaged out''.
\end{info}

Now, we go back to  the regular irreducible representations of the Weyl algebra. The question is how many regular irreducible unitary inequivalent representations does a Weyl algebra have.
The Stone-von Neumann theorem provides the answer.
To simplify the notation, for the rest of this section we will assume \( \hbar =1 \), which is always possible to do by suitably redefining the units of meter, second, and kilogram.  This is why \( \hbar \) will not appear in our commutation relations in the following.

\begin{thm}[\textbf{Stone-von Neumann}]\label{thm:stone von neumann}
All regular irreducible representations of the Weyl algebra are unitarily equivalent.
\end{thm}

Before starting the proof notice that each state on the Weyl algebra is uniquely defined by its action on the Weyl elements. To see this, observe that each element \( A \) of the Weyl algebra \( \mfW \) is by definition a norm limit of linear combination of Weyl elements of the form \eqref{eq:elements of Weyl algebra} such that
\begin{align}
    A = \sum_{n\in \bbN} a_{n} W(v_{n}).
\end{align}
Thence, if \( \omega \) is a state on \( \mfW \) then for each \( \epsilon > 0 \) there is an \( N_{\epsilon } \) such that
\begin{align}
    \epsilon \geq \| A - \sum_{n = 0}^{N_{\epsilon}} a_{n} W (v_n) \|
    \geq \Big| \omega \Big( A - \sum_{n = 0}^{N_{\epsilon}} a_{n} W (v_n) \Big) \Big|
    =
    \Big|\omega \big( A \big) -  \sum_{n = 0}^{N_{\epsilon}} a_{n} \omega \big( W (v_n) \big) \Big|,
\end{align}
and \( \omega (A) \) is uniquely defined by the expectation values on the Weyl elements.

\begin{proof}[Proof of Theorem \ref{thm:stone von neumann}]
The idea of proof is to show that if \( \pi \) is a regular irreducible representation of \( \mfW \) on the Hilbert space \( H \) then there exists a vector \( \fockv \in H \) such that
\begin{align}\label{eq:fock vector}
    \scalp{\fockv, \pi \big(W(v)\big) \fockv} = e^{-|v|^{2}/4}.
\end{align}
Then such a vector defines an (algebraic) state that is uniquely specified by the relation
\begin{align}\label{eq:fock state}
    \fock \big( W(v) \big)
    = \scalp{\fockv, \pi \big( W(v) \big) \fockv}
    = e^{-|v |^{2}/4}
    \qquad \big( v \in \bbR^{2d} \big),
\end{align}
and if \( \fockv \) is unique (up to normalization) then by the GNS construction, \( \pi \) is unitarily equivalent to the GNS representation of the state \( \fock \). In this case, we would have shown that every regular irreducible representation of the Weyl algebra \( \mfW \) is equivalent to the GNS representation of the state \( \fock \) (Figure \ref{fig:sonte-von-neumann}).
\begin{figure}
\def\svgwidth{13cm}
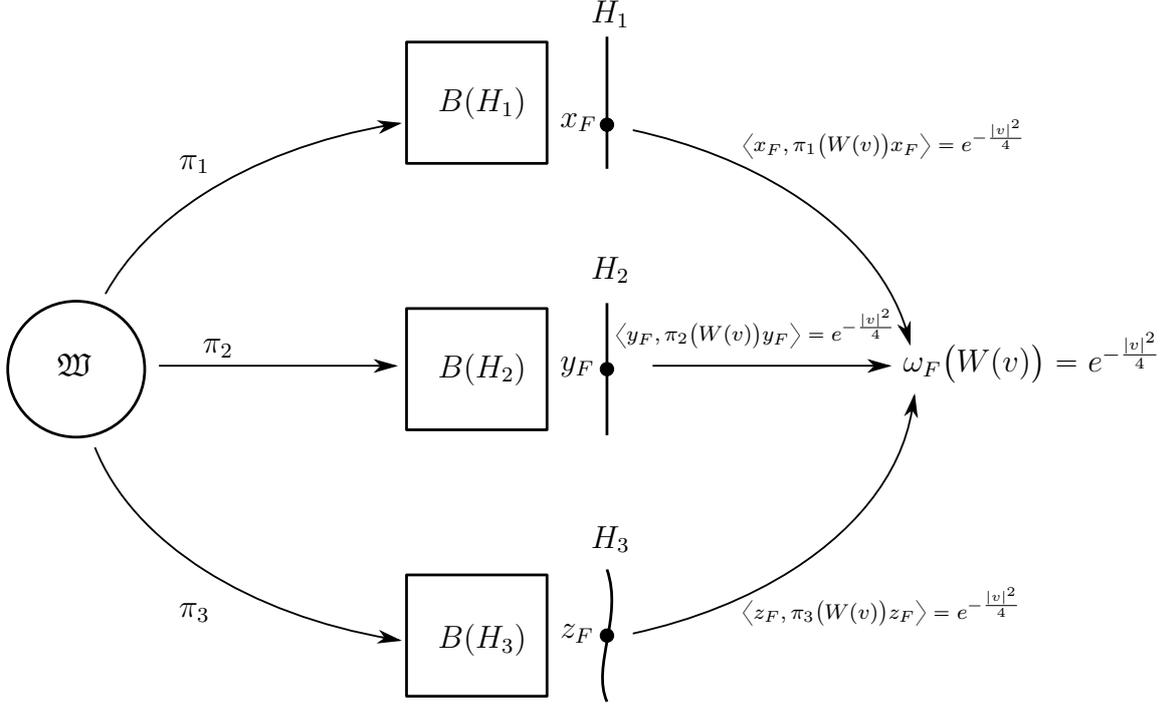
\caption{For every regular irreducible representation there is a vector \( \fockv \) (resp. \( y_{F} \), \( z_{F} \)), that defines the same algebraic state \( \fock \). By the GNS construction all the representations are then equivalent, since they appear as the GNS representations of the state \( \fock \).}
\label{fig:sonte-von-neumann}
\end{figure}

To construct such a vector \( \fockv \) we use the integral construction that we discussed above.
Let \( \pi \) be a regular irreducible representation of \( \mfW \) on \( H \). By Proposition \ref{prop: weyl} \ref{prop: weyl it:2} and since \( \pi \) is a \( \star \)-homomorphism the operator \( \pi \big(W(v) \big) \) is unitary and thus bounded for all \( v \in \bbR^{2d} \) such that \( \sup \{ \| \pi \big( W(v) \big) \| \colon v \in \bbR^{2d} \}\leq1 \).
Since the function \( v \to e^{- \| v \|^{2} / 4} \) is continuous and integrable on \( \bbR^{2d} \) the conditions of our integral construction are satisfies and we can define the operator
\begin{align}
    \proj = \frac{1}{2\pi} \int_{\bbR^{2d}} e^{-\frac{| v|^{2}}{4}} \, \pi \big(W(v) \big) \ d\mu(v).
\end{align}
For \( x,y \in H \) we have
\begin{align}
    \scalp{x, \proj y}
    &= \frac{1}{2\pi}
        \int_{\bbR^{2d}}
            e^{
            \frac{
             - | v |^{2}
             }{%
             4
            }
            } \ \scalp{x,\pi \big( W(v) \big) y } \ d \mu \\
    &= \frac{1}{2\pi}
        \int_{\bbR^{2d}}
            e^{
            \frac{
             - | v |^{2}
             }{%
             4
            }
            } \ \scalp{\pi \big( W(-v) \big)x, y } \ d \mu\\
    &= \frac{1}{2\pi}
        \int_{\bbR^{2d}}
            e^{
            \frac{
             - | v |^{2}
             }{%
             4
            }
            } \ \scalp{\pi \big( W(v) \big)x, y } \ d \mu
    = \scalp{ \proj x , y},
\end{align}
implying that \( \proj \) is self-adjoint.
Moreover, \( \proj \) is not identically zero, otherwise for every \( w \in \bbR^{2d} \) and \( x,y \in H \) we had
\begin{align}\label{eq: zero fourier}
    0
    = \scalp{\pi \big( W(w)   \big)x, \proj \pi \big( W(w) \big) y}
    = \frac{1}{2\pi} \int_{\bbR^{2d}} e^{- \frac{| v |^{2}}{4}}  \scalp{x, \pi \big(W(v) \big)y} \,  e^{\imath \sigma(v, w)} \ d \mu(v).
\end{align}
Writing \( v = (\alpha, \beta) \) and using the definition of \( \sigma \) we see that the right hand side of the above equation is the Fourier transform of the element \( e^{- \frac{| v |^{2}}{4}}  \scalp{x, \pi \big(W(v) \big)y} \). Thus \eqref{eq: zero fourier} would imply that for every \( x,y \in H \) the matrix element \(\scalp{x, \pi \big( W(v) \big)y} \) must vanish for every \( v \) yielding that \(  \pi \big(W(v) \big) = 0 \) for every \( v \in \bbR^{2d} \), which is absurd.

Next, we notice that by looking at the matrix elements \(     \scalp{\proj x,  \pi \, \big( W(z) \big) \proj y} \) for every \( x,y \in H \) we get the operator equality
\begin{align*}
    \proj \, \pi \, \big( W(z) \big) \proj
    &=
        \frac{1}{4\pi^{2}} \int_{\bbR^{2d}\times \bbR^{2d}}
        e^{- \frac{| v |^{2}}{4} - \frac{| w |^{2}}{4} }
        \pi \big( W(z + v + w) \big) \
        e^{-\imath (\sigma(z, w) + \sigma(v, z + w))}
 \ d\mu(w) d\mu(v)
\end{align*}
Using the coordinate transformation \( v \to \gamma \tfrac{1}{2}(\gamma - z + \delta) \) and \( w \to \tfrac{1}{2}(\gamma - z - \delta) \), a lot of patience, and Gaussian integrals, we eventually get (for a details see e.g.\ \cite[p. 575]{Neumann:1931ve})
\begin{align}
    \proj \, \pi \big( W(z) \big) \, \proj =
    e^{- \frac{| z |^{2}}{4}} \proj.
\end{align}
Setting \( z = 0 \) we get
\begin{align}\label{eq:proj idem}
    \proj^{2} = \proj.
\end{align}
Since \( \proj \) is self-adjoint and non-zero it defines a non-trivial orthogonal projection on \( H \). In particular, there exists at least one vector \( x \in H  \) such that \( \proj x \neq 0 \). Define
\begin{align}
    \fockv = \frac{\proj x}{ \| \proj x \|}.
\end{align}
Then \( \fockv = \proj \fockv \) and for every \( v \in \bbR^{2d} \) we get,
\begin{align}
    \scalp{\fockv, \pi \big( W(v) \big) \fockv}
    =\scalp{ \fockv, \proj \pi \big( W(v) \big) \proj \fockv}
    = e^{- \frac{| v |^{2}}{4}}  \scalp{ \fockv, \proj \fockv}.
    =e^{- \frac{| v |^{2}}{4}}
\end{align}

To finish the proof, we need to show that such an \( \fockv \) is unique, that is that the range of \( \proj \) is one-dimensional. Assume this is not the case, then there is a vector \( y_{F} \in H \) in the range of \( \proj \) (i.e.\ \( \proj y_{F} = y_{F} \)), and orthogonal to \( \fockv \). Then,
\begin{align*}
    \scalp{y_{F}, \pi \big( W(v) \big) \fockv}
    = \scalp{y_{F} , \proj \, \pi \big( W(v) \big) \proj \fockv}
    = e^{
     \frac{
      -| v |^{2}
      }{%
       4
     }
     } \scalp{y_{F}, \fockv} = 0,
\end{align*}
for every \( v \in \bbR^{2d} \), and thus by continuity \( \scalp{y_{F}, \pi(A)\fockv} \) for every \( A \in \mfW \). If \( y_{F} \neq 0 \), then \( y_{F} \) is in the orthogonal complement of the closure of the set of vectors \( \{ \pi(A)\fockv \colon A \in \mfW\} \). In particular, this set is a proper subset of \( H \), invariant under the action of \( \pi (\mfW) \). This, however, contradicts irreducibility of \( \pi \). Therefore, \( y_{F} = 0 \).
\end{proof}

\begin{info}[Remark:]
It is a general fact that if \( \pi \) is an irreducible representation of an algebra \( \mfA \) on a Hilbert space \( H \) then every vector of \( H \) is cyclic. To see this assume that there exists a non-cyclic vector \( z \in H \). In this case, the closure of \(\{ \pi(A)z \colon A \in \mfW\} \subset H \) is a closed proper subspace of \( H \) that is invariant under the action of \( \pi(A) \) for each \( A \in \mfW \). This contradicts irreducibility.
\end{info}

The state \( \fock \) from the proof of the Stone-von Neumann Theorem is called the \textit{Fock state}. It is almost exclusively the only state that we use for explicit calculations in physics.

The GNS representation induced by the Fock state is called \textit{the Fock representation} and the GNS Hilbert space is (called) \textit{the Fock space}.

\paragraph{In conclusion:}
To summarize the entire construction of quantum mechanics let us review the steps:
\begin{enumerate}
\item
Using physical intuition we argued that a generic physical system can be described by a \Cs-algebra that defines observables of that system.

\item
We realized that classical mechanics is described by commutative \Cs-algebras, and quantum mechanics is described by a non-commutative \Cs-algebra.

\item
To construct an explicit algebra for quantum mechanics we used the CCRs as a guideline. From commutation of ``position'' and ``momentum'' we ``guessed'' the product relations for the generating elements of the algebra (the Weyl relations). Then theorem \ref{thm:unique weyl} showed us that there is essentially a unique \Cs-algebra generated by elements that satisfy the Weyl relations --- the Weyl algebra.

\item
The Stone-von Neumann theorem showed us that there is essentially only one regular irreducible representation of the Weyl algebra. So there is only one Hilbert space (up to isometric isomorphism) that describes a quantum particle.

\item
Using the GNS construction we can explicitly construct the representation. At is the point we are fully equipped to do explicit calculations.
\end{enumerate}

\subsection{The Schr\"odinger representation of the Weyl algebra}
We finish this lecture by introducing the Schr\"odinger representation of quantum mechanics. It is the formulation of quantum mechanics that is commonly used in physics. Opposed to the physics approach, however, in the algebraic approach we derived this formulation from fundamental physical principles and rigorous mathematical analysis. In particular, we know that the Schr\"odinger quantum mechanics is a \textit{representation} of an underlying algebra that describes a quantum particle.
In this section, we only show the general concepts, as we do not have time develop the details of the construction.

We begin with the Hilbert space \( \lt{\bbR^{d}} \) of square integrable functions with respect to the Lebesgue measure.
In the following, it will be convenient to get back to the original Weyl operators \( U \) and \( V \). So with \( v \) of the form \( v = (\alpha, 0) \in \bbR^{2d} \) we define \( U(\alpha) = W(v) \) and with \( w \) of the form \( w = (0,\beta) \in \bbR^{2d} \) we define \( V(\beta) = W(w) \).
Now we define a representation of the Weyl algebra. That is, we define the \textit{action} of the Weyl operators on the Hilbert space \( \lt{\bbR^{d}} \). For \( \psi \in \lt{\bbR^{d} }\) we set
\begin{align}\label{eq:schroedinger rep}
    \big( U(\alpha) \psi \big)(x) = e^{\imath \alpha x} \psi(x),
    &&
    \big( V(\beta) \psi \big)(x) = \psi(x + \beta),
\end{align}
where the equality is meant to hold almost everywhere.
After analyzing this action we arrive at the conclusion that \( U(\alpha) \) and \( V(\beta) \) define a strongly continuous irreducible representation of the Weyl algebra.

Using the Stone theorem, we can derive the operators \( U(\alpha) \) and \( V(\beta) \) with respect to their parameters. The derivatives at the origin yield the generators of the Weyl operators in the Schr\"odinger representation. We can explicitly calculate,
\begin{align}\label{eq:Schroedinger operators}
    \Big(\Big.-\imath \frac{\partial U(\alpha)}{\partial \alpha_{i}}\Big\vert_{\alpha = 0} \psi \Big)(x) = x \psi(x)
&&
    \Big(\Big.-\imath \frac{\partial V(\beta)}{\partial \beta{i}}\Big\vert_{\beta=0} \psi \Big)(x) = \Big(- \imath \frac{\partial}{\partial x_{i}}\Big) \psi (x),
\end{align}
for \( i \in \{1, \dots , d \} \). For each \( i \) we call the generator of \( U(\alpha_{i}) \) the position operator \( q_i \) and group all these generators as \( q = (q_{1},\dots, q_{d}) \) so that \( U(\alpha) = e^{\imath \alpha q} \). Similarly, for each \( i \) we call the generator of \( V(\beta_{i}) \) the momentum operator \( p_{i}\) and group them to \( p = (p_{1},\dots, p_{d}) \) so that \( V(\beta) = e^{\imath \beta p} \). If we use our notation \( \id(x) = x \) to denote the identity function and \( \nabla \) to denote the Euclidean gradient on \( \bbR^{d} \) then \eqref{eq:Schroedinger operators} defines the action of \( q \) and \( p \) on \( L^{2} \)-functions as
\begin{align}
    q\psi  = \id \cdot \psi
    &&
    p \psi  = - \imath \nabla \psi.
\end{align}
These operators map \( L^{2} \)-functions into \( d \)-copies of \( L^{2} \), that is into the space of square integrable vector fields \( \lt{\bbR^{d}; \bbR^{d}} \).
In particular, we see that neither \( q \) nor \( p \) are bounded. Their domains of definition follow from the Stone's theorem such that
\begin{align}
    dom(q) &= \{ \psi \in \lt{\bbR^{d}} \colon x \psi \in \lt{\bbR^{d}; \bbR^{d}} \},\\
    dom(p) &= \{ \psi \in \lt{\bbR^{d}} \colon \nabla \psi \in \lt{\bbR^{d}; \bbR^{d}}\} = H^{1}(\bbR^{d}),
\end{align}
where \( H^{1}(\bbR^{d}) \) denotes the Sobolev space of square integrable functions, whose derivatives are also square integrable.
Especially, on these domains \( q \) and \( p \) are self-adjoint. Their product can be defined on the common core domain of rapidly decreasing functions --- the Schwarz space \( \mcS(\bbR^{d}) \) --- and on that space we can easily calculate
\begin{align}
    (qp - pq) f =
    - \imath x \nabla f + \imath \mathds{1}f + \imath x \nabla f
    =
    \imath \mathds{1}f
    \qquad \big( f \in \mcS(\bbR^{d}) \big).
\end{align}
So the CCRs are satisfied on the core of \( q \) and \( p \).
\begin{info}[Remark:]
Notice, that we could as easily choose a different representation in which \( U(\alpha) \) would act by transition and \( V(\beta) \) would act by multiplication. In this case, we would obtain the Fourier transform of the above choice. Notice, that this also implies that the definition of position operator as a multiplication operator is ambiguous. It really depends on the Hilbert space we choose. In fact, because of this symmetry between representations, it is hardly possible to distinguish between position and momentum operator at all. The only distinguishing feature will be the dynamics of the system, so we can realize which is which, by looking at their time evolution.
\end{info}

Pure state on the algebra of observables correspond to unit-normalized functions in \( \lt{\bbR^{d}} \).
For each unit vector \( \psi \in H \) we can define the state \( \omega_{\psi} \) such for \( A \in B(H) \) we get
\begin{align}
    \omega_{\psi}(A) = \int_{\bbR^{d}} \overline{\psi}(x) (A\psi)(x) \ d \lambda(x).
\end{align}
If \( A \) is a multiplication operator by some function \( F(x) \) then \( A \) generates a commutative \Cs-subalgebra of \( B(H) \), and \( \omega_{\psi} \) defines a state on this algebra. In particular, as in the case of the classical mechanics we get the probabilistic interpretation of the state \( \omega_{\psi} \) as
\begin{align}
    \omega_{\psi} (A) =
    \int_{\bbR^{d}} \overline{\psi}(x) (A\psi)(x)
    =
    \int_{\bbR^{d}} | \psi(x) |^{2} F(x) \ d \lambda (x).
\end{align}
where \( | \psi |^{2}(x) \) defines the probability distribution finding the particle at position \( x \).
This interpretation is possible for the algebra generated by \( A \). In general, however, it is not possible to find a state that has a probabilistic interpretation for all observables. This distinguishes quantum mechanics from the classical case.

The next step is be to define the dynamics of the particle by defining the evolution operator.
This, however, is beyond the scope of this lecture.

\newpage
\printbibliography
\end{document}